\documentclass[runningheads]{llncs}

\usepackage{a4wide}

\usepackage[latin1]{inputenc}
\usepackage{graphicx}
\usepackage{amsfonts}
\usepackage{amssymb}
\usepackage{amsmath}
\usepackage{framed}
\usepackage{latexsym}
\usepackage{breakcites}
\usepackage{color}
\usepackage{url}
\usepackage{tikz}
\usetikzlibrary{shapes}
\usetikzlibrary{patterns}
\usetikzlibrary{decorations.pathreplacing}
\usepackage{algorithm}
\usepackage[noend]{algpseudocode}
\usepackage{enumerate}
\usepackage{stmaryrd}
\usepackage{booktabs}

\usepackage{thm-restate}

\newtheorem{fact}{Fact}

\makeatletter
\newcounter{restate}

\makeatother

\newcommand{\cA}{\mathcal{A}}
\newcommand{\cC}{\mathcal{C}}
\newcommand{\cD}{\mathcal{D}}
\newcommand{\cE}{\mathcal{E}}

\newcommand{\cI}{\mathcal{I}}
\newcommand{\cJ}{\mathcal{J}}

\newcommand{\cS}{\mathcal{S}}
\newcommand{\cU}{\mathcal{U}}

\newcommand{\Th}[1]{\Theta\left( #1 \right)}

\newcommand{\av}{{\mathbf{a}}}
\newcommand{\bv}{{\mathbf{b}}}
\newcommand{\cv}{{\mathbf{c}}}
\newcommand{\dv}{{\mathbf{d}}}
\newcommand{\ev}{{\mathbf{e}}}

\newcommand{\hv}{{\mathbf{h}}}
\newcommand{\mv}{{\mathbf{m}}}
\newcommand{\rv}{{\mathbf{r}}}
\newcommand{\sv}{{\mathbf{s}}}
\newcommand{\uv}{{\mathbf{u}}}
\newcommand{\vv}{{\mathbf{v}}}
\newcommand{\xv}{{\mathbf{x}}}
\newcommand{\yv}{{\mathbf{y}}}

\newcommand{\Am}{{\mathbf{A}}}
\newcommand{\Bm}{{\mathbf{B}}}
\newcommand{\Cm}{{\mathbf{C}}}
\newcommand{\Dm}{{\mathbf{D}}}
\newcommand{\Gm}{{\mathbf{G}}}
\newcommand{\Hm}{{\mathbf{H}}}

\newcommand{\Imat}{{\mathbf{I}}}
\newcommand{\Mm}{{\mathbf{M}}}
\newcommand{\Pm}{{\mathbf{P}}}
\newcommand{\Sm}{{\mathbf{S}}}

\newcommand{\mat}[1]{\ensuremath{\boldsymbol{#1}}}
\newcommand{\un}{{\mat{1}}}
\newcommand{\zero}{{\mat{0}}}
\newcommand{\sigv}{{\mat{\sigma}}}

\newcommand{\Diag}{\mathbf{Diag}}

\newcommand{\Cc}{{\mathcal C}}
\newcommand{\Dc}{{\mathcal D}}
\newcommand{\Ec}{{\mathcal E}}

\newcommand{\Hc}{{\mathcal H}}
\newcommand{\Ic}{{\mathcal I}}
\newcommand{\Jc}{{\mathcal J}}
\newcommand{\Lc}{{\mathcal L}}

\newcommand{\Qc}{{\mathcal Q}}
\newcommand{\Sc}{{\mathcal S}}
\newcommand{\Uc}{{\mathcal U}}

\newcommand{\Cpub}{\Cc_{\text{pk}}}
\newcommand{\Cr}{\Cc_{\text{rand}}}
\newcommand{\Hpub}{\Hm_{\textup{pk}}}
\newcommand{\trHpub}{\transpose{\Hm}_{\textup{pk}}}
\newcommand{\Hsec}{{\Hm_{\textup{sk}}}}

\newcommand{\Psucc}{{P_{\text{succ}}}}

\newcommand{\WF}{\mathrm{WF}}

\newcommand{\DV}{\Call{DecodeV}}
\newcommand{\DU}{\Call{DecodeU}}
\newcommand{\DVv}{\Call{VarDecodeV}}
\newcommand{\DUv}{\Call{VarDecodeU}}

\newcommand{\DUVt}{\calltxt{DecodeUV}}
\newcommand{\DGUV}{\text{\tt{D}}_{UV}}
\newcommand{\Dgen}{\text{\tt {D}}_{\text{gen}}}

\newcommand{\lw}{m_{1}}

\newcommand{\Lcts}{S}
\newcommand{\Lcth}{H_\textup{true}}
\newcommand{\Lcf}{H_\textup{false}}

\newtheorem{notation}{Notation}
\newcommand{\Lap}[3]{\textup{Lap}({#1},{#2},{#3})}

\newcommand{\ltypu}{m_{\textup{target}}^{\textup{max}}}

\newcommand{\Z}{\mathbb{Z}}
\newcommand{\F}{\mathbb{F}}
\newcommand{\Fq}{\F_q}

\newcommand{\Unif}{\hookleftarrow}

\newcommand{\hsp}{\odot}

\makeatletter
\newcommand*{\transp}{{\mathpalette\@transpose{}}}
\newcommand*{\@transpose}[2]{\raisebox{\depth}{$\m@th#1\intercal$}}
\makeatother
\newcommand{\transpose}[1]{{#1}^{\transp}}
\newcommand{\tran}[1]{\transpose{#1}}

\newcommand{\prob}{\mathbb{P}}
\newcommand{\esp}{\mathbb{E}}
\newcommand{\eqdef}{\mathop{=}\limits^{\triangle}}
\DeclareMathOperator*{\var}{{\bf Var}}

\DeclareMathOperator*{\punc}{Punc}
\DeclareMathOperator*{\Sp}{Supp}
\DeclareMathOperator*{\supp}{Supp}

\DeclareMathOperator*{\Gen}{\text{\tt{Gen}}}
\DeclareMathOperator*{\hash}{\text{\tt{Hash}}}
\DeclareMathOperator*{\trap}{\text{\tt{Trapdoor}}}
\DeclareMathOperator*{\sampPre}{\text{\tt{InvertAlg}}}
\newcommand{\DOOM}{\ensuremath{\mathrm{DOOM}}}
\newcommand{\Sgnsk}{\ensuremath{\mathtt{Sgn}^{\mathrm{sk}}}}
\newcommand{\Vrfypk}{\ensuremath{\mathtt{Vrfy}^{\mathrm{pk}}}}

\newcommand{\Mrs}{M^{\text{rs}}}

\newcommand{\wt}[1]{|#1|}
\newcommand{\listM}{L_\mv}
\newcommand{\qhash}{q_{\textup{hash}}}
\newcommand{\qsig}{q_{\textup{sign}}}

\newcommand{\Dpub}{\Dc_{\textup{pub}}}
\newcommand{\Drand}{\Dc_{\textup{rand}}}

\newcommand{\Dpubw}{\Dc^{\textup{pub}}_w}
\newcommand{\Dpubwq}{\Dc^{\textup{pub}}_{w,q}}
\newcommand{\Dsw}[1]{\Dc_w^{#1}}
\newcommand{\UV}{(U,U+V)}
\newcommand{\IInt}[2]{\llbracket #1, #2 \rrbracket}

\newcommand{\wm}{w^-}
\renewcommand{\wp}{w^+}
\newcommand{\omegam}{\omega^-}
\newcommand{\omegap}{\omega^+}
\newcommand{\wme}{w^-_{\text{easy}}}
\newcommand{\wpe}{w^+_{\text{easy}}}
\newcommand{\omegame}{\omega^-_{\text{easy}}}
\newcommand{\omegape}{\omega^+_{\text{easy}}}
\newcommand{\omegaUVm}{\omega^-_{\text{UV}}}
\newcommand{\omegaUVp}{\omega^+_{\text{UV}}}
\newcommand{\wUVm}{w^-_{\text{UV}}}
\newcommand{\wUVp}{w^+_{\text{UV}}}

\newcommand{\vectspace}[1]{\langle{#1}\rangle}
\newcommand{\calltxt}[2]{{\sc #1}{\rm(#2)}}

\newcommand{\uni}[1]{#1^\textup{unif}}

\newcommand{\evu}{\uni{\ev}}

\newcommand{\qu}{\uni{q}}
\newcommand{\tphi}{\varphi}

\newcommand{\sk}{\mathrm{sk}}
\newcommand{\pk}{\mathrm{pk}}

\makeatletter
\renewcommand*\env@matrix[1][*\c@MaxMatrixCols c]{  \hskip -\arraycolsep
  \let\@ifnextchar\new@ifnextchar
  \array{#1}}
\makeatother

\newcommand{\wave}{Wave}
\begin{document}

\title{Wave: A New Family of Trapdoor One-Way Preimage Sampleable Functions Based on Codes
	\thanks{This work was supported by the ANR CBCRYPT project, grant ANR-17-CE39-0007 of the French Agence Nationale de la Recherche.}
      }
\author{Thomas Debris-Alazard\inst{1,2}   \and Nicolas Sendrier\inst{2} \and Jean-Pierre Tillich \inst{2}}
\institute{Sorbonne Universit\'{e}s, UPMC Univ Paris 06 \and Inria, Paris\\
	\email{\{thomas.debris,nicolas.sendrier,jean-pierre.tillich\}@inria.fr}}
\titlerunning{Wave: a New Family of Trapdoor One-Way PSF Based on Codes}
\maketitle

\begin{abstract}
  We present here a new family of trapdoor one-way Preimage Sampleable
  Functions (PSF) based on codes, the Wave-PSF family. The trapdoor
  function is one-way under two computational assumptions: the
  hardness of generic decoding for high weights and the
  indistinguishability of generalized $\UV$-codes. Our proof follows
  the GPV strategy \cite{GPV08}.  By including rejection sampling, we
  ensure the proper distribution for the trapdoor inverse
  output. The domain sampling property of our family is
  ensured by using and proving a variant of the left-over hash
  lemma.  We instantiate the new Wave-PSF family with ternary
  generalized $\UV$-codes to design a ``hash-and-sign'' signature
  scheme which achieves {\em existential unforgeability under adaptive
    chosen message attacks} (EUF-CMA) in the random oracle model. For
  128 bits of classical security, signature sizes are in the order of
  15 thousand bits, the public key size in the order of 4 megabytes,
  and the rejection rate is limited to one rejection every 10 to 12
  signatures.
\end{abstract}

\section{Introduction} 

\subsubsection*{Code-Based Signature Schemes.}
It is a long standing open problem to build an efficient and secure
digital signature scheme based on the hardness of decoding a linear
code which could compete with widespread schemes like DSA or
RSA. Those signature schemes are well known to be broken by quantum
computers and code-based schemes could indeed provide a valid quantum
resistant replacement. A first answer to this question was given by
the CFS scheme proposed in \cite{CFS01}. It consisted in finding
parity-check matrices $\Hm\in\F_2^{r \times n}$ such that the solution
$\ev$ of smallest weight of the equation
\begin{equation}
\label{eq:decoding}
\ev\transpose{\Hm}=\sv.
\end{equation}
could be found for a non-negligible proportion of all $\sv$ in
$\mathbb{F}_{2}^{r}$.  This task was achieved by using high rate Goppa
codes. This signature scheme has however two drawbacks: (i) for high
rates Goppa codes the indistinguishability assumption used in its
security proof has been invalidated in \cite{FGOPT11}, (ii) security scales
only weakly superpolynomially in the keysize for polynomial time signature time.  A crude extrapolation of
parallel CFS \cite{F10} and its implementations \cite{LS12,BCS13}
yields for 128 bits of classical security a public key size of several
gigabytes and a signature time of several seconds. Those figures even
grow to terabytes and hours for quantum-safe security levels, making
the scheme unpractical.

This scheme was followed by other proposals using other code families 
such as for instance \cite{BBCRS13,GSJB14,LKLN17}. All of them were broken, 
see for instance \cite{PT16,MP16}.
Other signature schemes based on codes were also given in the
literature such as for instance the KKS scheme \cite{KKS97,KKS05},
its variants \cite{BMS11,GS12} or the RaCoSS proposal  \cite{FRXKMT17} to the NIST.  But they can be considered at best to
be one-time signature schemes and great care has to be taken to choose the parameters
of these schemes in the light of the attacks given in
\cite{COV07,OT11,HBPL18}. 
Finally, 
another possibility is to use the Fiat-Shamir
heuristic. For instance by turning the Stern zero-knowledge authentication scheme
\cite{S93} into a signature scheme but this leads to rather large
signature lengths (hundred(s) of kilobits).
There has been some recent progress in this area for another metric,
namely the rank metric. A hash and sign signature scheme was proposed, RankSign \cite{GRSZ14}, that enjoys remarkably small key sizes, but it got broken too in \cite{DT18b}. On the other hand, following the Schnorr-Lyubashevsky \cite{L09_sv} approach, a new scheme was recently proposed, namely Durandal \cite{ABGHZ18}. This scheme enjoys small key sizes and managed to meet the challenge of adapting the Lyubashevsky \cite{L09} approach for code-based cryptography. However, there is a lack of genericity in its security reduction, the security of Durandal is reduced to a rather convoluted problem, namely PSSI$^{+}$ (see \cite[\S 4.1]{ABGHZ18}), capturing the problem of using possibly information leakage in the signatures to break the secret key. 
This is due to the fact that it is not proven in their scheme that their signatures do not leak information.

\paragraph{\bf One-Way Preimage Sampleable Trapdoor Functions.}
There is a very powerful tool for building a hash-and-sign signature scheme. It is based on the notion of 
 {\em one-way trapdoor preimage sampleable function} \cite[\S 5.3]{GPV08} (PSF in short). 
Roughly speaking, this is a family of trapdoor
one-way functions $(f_a)_a$ such that with overwhelming probability
over the choice of $f_a$ (i) the distribution of the
images $f_a(x)$ is very close to the uniform distribution over its range 
(ii) the distribution of the output of the
trapdoor algorithm inverting $f_a$ 
samples from all possible preimages in an appropriate way. This trapdoor inversion algorithm should namely
sample  for any $x$ in the output domain of $f_a$ its outputs $e$ such that the distribution of $e$ is indistinguishable in a statistical sense from the 
input distribution to $f_a$ conditioned on $f_a(e)=x$. 
This notion and its lattice-based instantiation allowed in \cite{GPV08} 
to give a full-domain hash (FDH) signature scheme
with a tight security reduction based on lattice assumptions, namely that the Short Integer Solution (SIS) problem is hard on average.
Furthermore, this approach also allowed to build the first identity
based encryption scheme that could be resistant to a quantum computer.  
We will call in this paper, this approach for obtaining a FDH scheme, the GPV strategy (the authors of \cite{GPV08} are namely 
Gentry, Peikert and Vaikuntanathan). This
strategy has also been adopted in Falcon \cite{FHKLPPRSWZ}, a lattice
based signature submission to the NIST call for post-quantum
cryptographic primitives that was recently selected as a second round candidate. 

This PSF primitive is notoriously difficult to obtain when the functions $f_a$ are not trapdoor permutations but many-to-one 
functions. This is typically the case when one wishes quantum resistant primitives based on lattice 
based assumptions. 
The reason is the following. The hard problem on which this primitive relies is the
SIS problem where we want to find for a matrix $\Am$ in $\Z_q^{n \times m}$ (with $m \geq n$) and an element 
$\sv \in \Z_q^n$ 
a short enough (for the Euclidean norm) solution $\ev \in \Z_q^m$ to the equation
\begin{equation}\label{eq:SIS}
\ev \transpose{\Am} = \sv \mod{q}.
\end{equation}
Such a matrix defines a corresponding PSF function as $f_{\Am}(\ev) = \ev \transpose{\Am}$ and the 
input to this function is chosen according to a Gaussian distribution that outputs 
$\ev$ 
of large enough euclidean norm $W$ so that \eqref{eq:SIS} has a solution. Obtaining a nearly uniform distribution for the
$f_{\Am}(\ev)$'s over its range requires 
to choose $W$ large enough so that there are actually {\em exponentially many} solutions to \eqref{eq:SIS}.
It is a highly non-trivial task to build in this case a trapdoor inversion algorithm that samples appropriately 
among all possible preimages, i.e. that is oblivious of the trapdoor.

The situation is actually exactly the same if we want to use another candidate problem for building this PSF primitive
for being resistant to a quantum computer, namely the decoding problem in code-based cryptography. Here we rely on the difficulty of finding
a solution $\ev$ of Hamming weight {\em exactly w} with coordinates in a finite field field $\Fq$ for the equation
\begin{equation}
\label{eq:ourdecoding}
\ev \transpose{\Hm} = \sv.
\end{equation}
where $\Hm$ is a given matrix and $\sv$ (usually called a syndrome) a given vector with entries in $\Fq$.
The weight $w$ has to be chosen large enough so that this equation has always exponentially many solutions 
(in $n$ the length of $\ev$). As in the lattice based setting, it is non-trivial to build trapdoor candidates with
a trapdoor inversion algorithm for $f_{\Hm}$ (defined as $f_{\Hm}(\ev)=\ev\transpose{\Hm}$)  that is oblivious of the trapdoor.

\paragraph{\bf Our Contribution: a Code-Based PSF Family and an FDH Scheme.}
Our main contribution is to give here a code-based PSF family that relies on the difficulty of solving
\eqref{eq:ourdecoding}. We derive from it an FDH signature scheme which is shown to be 
existentially unforgeable under a chosen-message attack
(EUF-CMA) with a tight reduction to solving two code-based problems: 
one is a  distinguishing problem related to the trapdoor used
in our scheme, the other one is a multiple target version of the
decoding problem \eqref{eq:ourdecoding}, the so called ``Decoding One Out
of Many'' problem (DOOM in short) \cite{S11}. 
In \cite{GPV08} a signature scheme based on preimage sampleable
functions is given that is shown to be strongly existentially
unforgeable under a chosen-message attack if in addition the preimage
sampleable functions are also collision resistant. With our choice of
$w$ and $\Fq$, our preimage sampleable functions are not collision
resistant. However, as observed in \cite{GPV08}, collision resistance
allows a tight security reduction but is not necessary: a security
proof could also be given when the function is ``only'' preimage
sampleable. Moreover, contrarily to the lattice setting where the size of the alphabet $q$ grows
with $n$, our alphabet size will be constant in our proposal, it is fixed to $q=3$.

\paragraph{\bf  Our Trapdoor: Generalized $\UV$-Codes.}
In \cite{GPV08} the trapdoor consists in a short basis of the lattice
considered in the construction. Our trapdoor will be of a different
nature, it consists in choosing parity-check matrices of generalized
$\UV$-codes. In our construction, $U$ and $V$ are chosen as random codes.
The number of such generalized $(U,U+V)$-codes of dimension $k$ and length $n$ is of the same order 
as the number of linear codes with the same parameters, namely $q^{\Th{n^2}}$ when $k=\Th{n}$. A
generalized $(U,U+V)$ code $\Cc$ of length $n$ over $\Fq$ is built from two codes $U$ and $V$ of length $n/2$ and $4$ vectors $\av, \bv, \cv$ and $\dv$ in $\Fq^{n/2}$ as the following ``mixture'' of $U$ and $V$: 
$$
\Cc = \{(\av \hsp \uv + \bv \hsp \vv,\cv \hsp \uv + \dv \hsp \vv): \uv \in U,\;\vv \in V\}
$$
where $\xv \hsp \yv$ stands here for the component-wise product, also called the Hadamard or Schur product. 
It is defined as:
$$
\xv \hsp \yv \eqdef (x_{1}y_{1},\cdots,x_{n/2}y_{n/2}). 
$$
Standard $(U,U+V)$-codes correspond to $\av=\cv=\dv=\un_{n/2}$ and $\bv=\zero_{n/2}$, 
the all-one and the all-zero vectors respectively.

The point of introducing such codes is that they have a natural decoding algorithm $\DGUV$ 
solving the decoding problem \eqref{eq:ourdecoding} that is
based on a generic decoding algorithm $\Dgen$ for linear codes.
$\DGUV$ works by combining the decoding of $V$ with $\Dgen$
with the decoding of $U$ by $\Dgen$.
 The nice feature is that $\DGUV$ is more powerful than 
$\Dgen$ applied directly on the generalized $\UV$-code:
the weight of the error produced by $\DGUV$  can be much smaller than the weight of the error produced by $\Dgen$ applied directly to the generalized $\UV$-code. In our case, 
$\Dgen$ will be here a very 
simple decoder,
namely a variation of the Prange decoder \cite{P62} that is able to
produce for {\em any} parity-check matrix $\Hm \in \Fq^{r \times n}$ at will a solution of
\eqref{eq:ourdecoding} when $w$ is in the range
$\IInt{\frac{q-1}{q}r}{n-\frac{r}{q}}$.  Note that this algorithm
works in polynomial time and that outside this range of weights, the
complexity of the best known algorithms is exponential in $n$ for
weights $w$ of the form $w= \omega n$ where $\omega$ is a constant
that lies outside the interval
$[\frac{q-1}{q}\rho, 1 - \frac{\rho}{q}]$ where
$\rho \eqdef \frac{r}{n}$.  
The point of using $\DGUV$ is that it produces errors outside this interval. 
This is in essence the
trapdoor of our signature scheme.  A tweak in this decoder
consisting in performing only a small amount of rejection sampling (with our choice of parameters one rejection every $10$ or $12$ signatures)
allows to obtain solutions that are uniformly distributed over the
words of weight $w$. This is the key for obtaining a PSF family
and a signature scheme from it. 

 Finally, a
variation of the proof technique of \cite{GPV08} allows to
give a tight security proof of our signature scheme that relies only
on the hardness of two problems, namely
\begin{description}
\item[Decoding Problem:] Solving at least one instance of the decoding
  problem \eqref{eq:decoding} out of multiple instances for a certain
  $w$ that is outside the range $\IInt{\frac{q-1}{q}r}{n-\frac{r}{q}}$
\item[Distinguishing Problem:] Deciding whether a linear code is a
  permuted generalized $\UV$ code or not.
\end{description}

Interestingly, some recent work \cite{CD17} has shown that these two
properties (namely statistical indistinguishability of the signatures and
the syndromes associated to the code family chosen in the scheme) are
also enough to obtain a tight security reduction in the Quantum Random
Oracle Model (QROM) for generic code-based signatures. The security reduction is made to a problem that is 
called the Claw with Hash problem. It can be viewed as an adaptation of the DOOM problem to the quantum setting. 
In this case, an adversary has access to a quantum oracle for producing the instances that he wants to decode. In other words, this can
be used to give a tight security proof of our generalized $\UV$-codes
in the QROM.

\paragraph{\bf Hardness of the Decoding Problem.} 
All code-based cryptography relies upon that problem.  Here we are in
a case where there are multiple solutions of \eqref{eq:ourdecoding} and
the adversary may produce any number of instances of
\eqref{eq:ourdecoding} with the same matrix $\Hm$ and various syndromes
$\sv$ and is interested in solving only one of them. This relates to
the, so called, Decoding One Out of Many (DOOM) problem. This problem
was first considered in \cite{JJ02}.  It was shown there how to adapt
the known algorithms for decoding a linear code in order to solve this
modified problem. This modification was later analyzed in
\cite{S11}. The parameters of the known algorithms for solving
\eqref{eq:ourdecoding} can be easily adapted to this scenario where we
have to decode simultaneously multiple instances which all have
multiple solutions.

\paragraph{\bf Hardness of the Distinguishing Problem.}
This problem might seem at first sight to be ad-hoc. However, even in
the very restricted case of
$\UV$-codes, deciding whether a code is a permuted $\UV$-code or not is an NP-complete problem.
Therefore the Distinguishing Problem is also
NP-complete for generalized $\UV$-codes. This theorem is proven in the
case of binary $\UV$-codes in \cite[\S 7.1, Thm 3]{DST17b} and the
proof carries over to an arbitrary finite field $\Fq$.  However as
observed in \cite[p. 3]{DST17b}, these NP-completeness reductions hold
in the particular case where the dimensions $k_U$ and $k_V$ of the
code $U$ and $V$ satisfy $k_U < k_V$. If we stick to the binary case,
i.e. $q=2$, then in order that our $\UV$ decoder works outside the
integer interval $\IInt{\frac{r}{2}}{n-\frac{r}{2}}$ it is necessary
that $k_U > k_V$.  Unfortunately in this case there is an efficient
probabilistic algorithm solving the distinguishing problem that is
based on the fact that in this case the hull of the permuted
$\UV$-code is typically of large dimension, namely $k_U - k_V$ (see
\cite[\S1 p.1-2]{DST17}). This problem can not be settled in the
binary case by considering generalized $\UV$-codes instead of just
plain $\UV$-codes, since it is only for the restricted class of $\UV$-codes that the decoder considered in \cite{DST17} is able to
work properly outside the critical interval
$\IInt{\frac{r}{2}}{n-\frac{r}{2}}$. This explains why the ancestor Surf  \cite{DST17}
of the scheme proposed here that relies on binary $\UV$-codes can not work.

This situation changes drastically when we move to larger finite
fields. In order to have a decoding algorithm $\DGUV$ that has an advantage 
over the generic decoder $\Dgen$ we do not need to have 
$\av=\cv=\dv=\un_{n/2}$ and $\bv=\zero_{n/2}$ (i.e. $\UV$-codes) we just need that 
$\av \hsp \cv$ and $\av \hsp \dv - \bv \hsp \cv$ are vectors with only non-zero components.
This freedom of choice for the $\av,\bv,\cv$ and $\dv$ thwarts completely the attacks based 
on hull considerations and changes completely the nature of the distinguishing problem.
In this case, it seems that the best
approach for solving the distinguishing problem is based on the
following observation. The generalized $\UV$-code has codewords of weight
slightly smaller than the minimum distance of a random code of the
same length and dimension. It is very tempting to conjecture that the
best algorithms for solving the Distinguishing Problem come from
detecting such codewords. This approach can be easily thwarted by
choosing the parameters of the scheme in such a way that the best
algorithms for solving this task are of prohibitive complexity. Notice
that the best algorithms that we have for detecting such codewords are
in essence precisely the generic algorithms for solving the Decoding
Problem. In some sense, it seems that we might rely on the very same
problem, namely solving the Decoding Problem, even if our proof
technique does not show this.

\paragraph{\bf $q=3$ and Large weights Decoding.}
 In terms of simplicity of the decoding procedure used in the signing process, it
seems that defining our codes over the finite field $\F_3$ is
particularly attractive. In such a case, the biggest advantage of $\DGUV$ over $\Dgen$ is obtained for large weights
rather than for small weights (there is an explanation for this asset in the paragraph 
{\em ``Why is the trapdoor more powerful for large weights than for small weights?''} \S \ref{subsec:genUVcodes}).
This is a bit unusual in code-based cryptography to rely on the difficulty of finding solutions 
of large weight to the decoding problem. However, it also opens the issue whether it would not be advantageous to 
make certain (non-binary) code-based primitives rely on the hardness of solving the decoding problem for large weights rather than for small weights. Of course these two problems are equivalent in the binary case, i.e. $q=2$, but this is not the case for larger alphabets anymore and still everything seems to point to the direction that large weights problem is by no means easier
than its small weight counterpart.

All in all, this gives the first practical signature scheme based on
ternary codes which comes with a security proof and which scales well
with the parameters: it can be shown that if one wants a security
level of $2^\lambda$, then signature size is of order $O(\lambda)$,
public key size is of order $O(\lambda^2)$, signature generation is of
order $O(\lambda^3)$, whereas signature verification is of order
$O(\lambda^2)$. It should be noted that contrarily to the current
thread of research in code-based or lattice-based cryptography which
consists in relying on structured codes or lattices based on ring
structures in order to decrease the key-sizes we did not follow this
approach here. This allows for instance to rely on the NP-complete
Decoding Problem which is generally believed to be hard on average
rather that on decoding in quasi-cyclic codes for instance whose
status is still unclear with a constant number of circulant
blocks. Despite the fact that we did not use the standard approach for
reducing the key sizes relying on quasi-cyclic codes for instance, we
obtain acceptable key sizes (about 3.8 megabytes for 128 bits of
security) which compare very favorably to unstructured lattice-based
signature schemes such as TESLA for instance \cite{ABBDEGKP17}. This
is due in part to the tightness of our security reduction.

 \section{Notation} 
\label{sec:nota} 
We provide here some notation that will be used throughout the paper.
\newline

{\noindent \bf General Notation.}
The notation $x \eqdef y$ means
that $x$ is defined to be equal to $y$. We denote by $\mathbb{F}_{q}$
the finite field with $q$ elements and by $S_{w,n}$, or $S_w$ when $n$
is clear from the context, the subset of $\F_q^n$ of words of weight
$w$. For $a$ and $b$ integers with $a \leq b$, we denote by 
$\IInt{a}{b}$ the set of integers $\{a,a+1,\dots,b\}$.
\newline

{\noindent \bf Vector and Matrix Notation.} 
Vectors will be written with  bold letters (such as $\ev$) and  uppercase bold letters are used to denote matrices (such as $\Hm$). Vectors are in row notation.
Let $\xv$ and $\yv$ be two vectors, we will write $(\xv,\yv)$ to denote their concatenation.
We also denote by $\xv_\cI$ the vector whose coordinates are those of $\xv=(x_i)_{1 \leq i \leq n}$ which are indexed by $\cI$, i.e. 
$
\xv_\cI = (x_i)_{i \in cI}
$. We will denote by $\Hm_{\cI}$ the matrix whose columns are those of $\Hm$ which are indexed by $\cI$. 
Sometimes we denote for a vector $\xv$ by $\xv(i)$ its $i$-th entry, or for a matrix $\Am$, by $\Am(i,j)$ its entry in row $i$ and column $j$. We define the support of $\xv = (x_i)_{1 \leq i \leq n}$ as
$$
\supp(\xv) \eqdef \{ i \in \{1,\cdots,n \} \mbox{ such that } x_{i} \neq 0 \}
$$
The Hamming weight of $\xv$ is denoted by 	
$|\xv|$.
By some abuse of notation, we will use the same notation 
to denote the size of a finite set: $|S|$ stands for the size of the finite set $S$.  
It will be clear from the context whether $|\xv|$ means the Hamming weight or the size of a finite set. 
Note that
$
|\xv| = |\supp(\xv)|.
$
For a vector $\av \in \Fq^n$, we denote by $\Diag(\av)$ the $n \times n$ diagonal matrix $\Am$ with its entries
given by $\av$, i.e. $\Am(i,i)=a_i$ for all $i \in \IInt{1}{n}$ and $\Am(i,j) = 0$ for $i \neq j$.
\newline

{\noindent \bf Probabilistic Notation.} Let $S$ be a finite set, then $x \Unif S$ means 
that $x$ is assigned to be a random element chosen uniformly at random in $S$. For two random variables $X,Y$, $X \sim Y$ 
means that $X$ and $Y$ are identically distributed. We will also use the same notation for a random variable 
and a distribution $\Dc$, where $X \sim \Dc$ means that that $X$ is distributed according to $\Dc$.
We denote the uniform distribution on $S_{w}$ by $\mathcal{U}_{w}$. 

The statistical distance between two discrete probability distributions over a same space $\mathcal{E}$ is defined as:
$
\rho(\cD_0,\cD_1) \eqdef \frac{1}{2} \sum_{x \in \mathcal{E}} |\cD_0(x)-\cD_1(x) |.
$
Recall that a function $f(n)$ is said to be negligible, and we denote this by $f \in \textup{negl}(n)$, if for all polynomials $p(n)$, $|f(n)| < p(n)^{-1}$ for all sufficiently large $n$.
\newline

{\noindent \bf Coding Theory.}
For any matrix $\Mm$ we denote by $\vectspace{\Mm}$ the vector space
spanned by its rows. A $q$-ary linear code $\cC$ of length $n$ and
dimension $k$ is a subspace of $\mathbb{F}_{q}^{n}$ of dimension $k$
and is often defined by a {\em parity-check matrix} $\Hm$ over $\F_q$
of size $r \times n$ as
$$
\Cc = \vectspace{\Hm}^\perp = \left\{ \xv \in \mathbb{F}_{q}^{n}:  \xv \transpose\Hm=\mathbf{0}\right\}.
$$ 	
When $\Hm$ is of full rank (which is usually the case) we have
$r = n-k$. A {\em generator matrix} of $\Cc$ is a $k \times n$ full
rank matrix $\Gm$ over $\F_q$ such that $\vectspace{\Gm}=\Cc$. The
code rate, usually denoted by $R$, is defined as the ratio ${k}/{n}$.

An {\em information set} of a code $\Cc$ of length $n$ is a set of $k$
coordinate indices $\Ic\subset\llbracket 1,n \rrbracket$ which indexes $k$
independent columns on any generator matrix. Its complement indexes
$n-k$ independent columns on any parity check matrix. For any
$\sv\in\Fq^{n-k}$, $\Hm\in\Fq^{(n-k)\times n}$, and any information
set $\Ic$ of $\Cc=\vectspace{\Hm}^\perp$, for all $\xv\in\Fq^{n}$
there exists a unique $\ev\in\Fq^n$ such that $\ev \transpose\Hm=\sv$
and $\xv_\Ic=\ev_\Ic$.

 \section{The \wave-family of Trapdoor One-Way Preimage Sampleable Functions} 
\label{sec:genSig}

\subsection{One-way Preimage Sampleable Code-based Functions}\label{subsec:WPS}

In this work we will use the 
FDH paradigm
\cite{BR96,C02} using as one-way the syndrome function:
\begin{displaymath}
\begin{array}{lccc}
f_{w,\Hm} :  &\ev \in S_{w} & \longmapsto & \ev\transpose{\Hm}\in \Fq^{n-k}\\
\end{array}
\end{displaymath}
The corresponding FDH signature uses a trapdoor to choose
$\sigv \in f_{w,\Hm}^{-1}(\hv)$ where $\hv$ is the digest of the message to be
signed. Here, the signature domain is $S_w$ 
and its range is the set of
syndromes $\Fq^{n-k}$ according to $\Hm$, an $(n-k) \times n$ parity
check matrix of some $q$-ary linear $[n,k]$ code. The weight $w$ is
chosen such that the one-way function $f_{w,\Hm}$ is surjective but
not bijective. Building a secure FDH signature in this situation can
be achieved by imposing additional properties \cite{GPV08} to the
one-way function (we will speak of the GPV strategy). This is mostly
captured by the notion of Preimage Sampleable Functions (PSF), see
\cite[Definition 5.3.1]{GPV08}. We express below this notion in our
code-based context with a slightly weaker definition that drops the collision resistance
condition.
This will be sufficient for proving the security of our code-based FDH scheme.
 The key feature is a trapdoor inversion of $f_{w,\Hm}$ which
achieves (close to) uniform distribution over the domain $S_w$.

\begin{definition}[One-way Preimage Sampleable Code-based
  Functions] \label{def:WPS} It is a pair of probabilistic
  polynomial-time algorithms $(\trap,\sampPre)$ together with a triple
  of functions $(n(\lambda),k(\lambda),w(\lambda))$
  growing polynomially with the security parameter $\lambda$
  and giving the length and dimension of the codes and the weights we
  consider for the syndrome decoding problem, such that
  \begin{itemize}
  \item $\trap$ when given $\lambda$, outputs $(\Hm,T)$ where $\Hm$ is
    an $(n-k) \times n $ matrix over $\Fq$ and $T$ the trapdoor
    corresponding to $\Hm$. Here and elsewhere we drop the dependence
    in $\lambda$ of the functions $n,k$ and $w$.
  \item $\sampPre$ is a probabilistic algorithm which takes as input $T$
    and an element $\sv \in \Fq^{n-k}$ and outputs an $\ev \in
    S_{w,n}$ such that $\ev\tran{\Hm} = {\sv}$.
  \end{itemize}
  The following properties have to hold for all but a negligible
  fraction of $\Hm$ output by $\trap$.
  \begin{enumerate}
  \item \textup{Domain Sampling with uniform output:} 
    $$\rho(\ev\tran{\Hm},{\sv}) \in \textup{negl}(\lambda)$$
    where $\ev$ and $\sv$ are two random variables, with $\ev$ being uniformly distributed over 
    $S_{w,n}$ and $\sv$ being uniformly distributed over $\Fq^{n-k}$.	
  \item \textup{Preimage Sampling with trapdoor:} for every $\sv \in \Fq^{n-k}$, we have
    $$\rho\left( \sampPre(\sv,T),\ev_s \right) \in \textup{negl}(\lambda),$$
    where $\ev_s$ is uniformly distributed over the set $\{\ev \in S_{w,n}:\ev \transpose{\Hm}=\sv\}$.
  \item \textup{One wayness without trapdoor:} for any probabilistic
    poly-time algorithm $\cA$
    outputting an element $\ev
    \in S_{w,n}$ when given $\Hm \in \Fq^{(n-k) \times n }$ and $\sv
    \in \Fq^{n-k}$, the probability that $\ev\tran{\Hm} =
    {\sv}$ is negligible, where the probability is taken over the
    choice of $\Hm$,
    the target value $\sv$
    chosen uniformly at random, and $\cA$'s random coins.
  \end{enumerate}
\end{definition}
Given a one-way preimage sampleable code-based function
$(\trap,\sampPre)$ we easily define a code-based FDH signature scheme
as follows. We generate the public/secret key as
$(\pk,\sk)=(\Hm,T) \leftarrow \trap(\lambda)$.
We also select a cryptographic hash function $\hash: \{0,1\}^{*}
\rightarrow \Fq^{n-k}$ and a salt $\rv$ of size $\lambda_{0}$.
The algorithms \Sgnsk\ and \Vrfypk\ are defined as follows
\begin{center}
	\begin{tabular}{l@{\hspace{3mm}}|@{\hspace{3mm}}l}
		$\Sgnsk(\mv)\!\!: \qquad \qquad \qquad$ & $\Vrfypk(\mv,(\ev',\rv))\!\!:$ \\
		$\quad \rv \Unif \{ 0,1 \}^{\lambda_{0}}$ &$\quad \sv \leftarrow \hash(\mv ,\rv)$ \\
		$\quad \sv \leftarrow \hash(\mv ,\rv)$ & $\quad \texttt{if } \ev'\transpose{\Hm} = {\sv} \texttt{ and } |\ev'| = w \texttt{ return } 1$\\ 
		$\quad\ev \leftarrow \sampPre(\sv,T) $ &$\quad \texttt{else return } 0 $ \\
		$\quad \texttt{return}(\ev,\rv)$& \\			
	\end{tabular} 
\end{center}
A tight security reduction in the random oracle model is given in
\cite{GPV08} for PSF signature schemes. It requires collision
resistance. Our construction uses a ternary alphabet $q=3$ together with large
values of $w$ and collision resistance is not met. Still, we achieve a
tight security proof by considering in \S\ref{sec:securityProof} a reduction to
the multiple target decoding problem.

\subsection{The Wave Family of One-Way Trapdoor Preimage Sampleable Functions} 
\label{subsec:waveTrap}

The trapdoor family of codes which  gives an advantage for inverting $f_{w,\Hm}$ is built upon the following
transformation:
\begin{definition}
Let $\av$, $\bv$, $\cv$ and $\dv$ be vectors of $\Fq^{n/2}$. We define
\begin{eqnarray*}
\varphi_{\av,\bv,\cv,\dv} :\Fq^{n/2} \times \Fq^{n/2} & \rightarrow &  \Fq^{n/2} \times \Fq^{n/2}\\
(\xv,\yv) & \mapsto & (\av \hsp \xv+\bv \hsp \yv,\cv \hsp \xv + \dv \hsp \yv).
\end{eqnarray*}
We will say that $\varphi_{\av,\bv,\cv,\dv}$ is UV-normalized if
\begin{equation} \label{eq:cdtInv}
  \forall i \in \llbracket 1,n/2 \rrbracket, \quad a_{i}d_{i} -
  b_{i}c_{i} = 1, \mbox{ } a_{i}c_{i} \neq 0.
\end{equation}
For any two subspaces $U$ and $V$ of $\Fq^{n/2}$, we extend the notation
\begin{displaymath}
  \varphi_{\av,\bv,\cv,\dv} (U,V) \eqdef \left\{ \varphi_{\av,\bv,\cv,\dv}(\uv, \vv) : \uv \in U, \vv \in V\right\}
\end{displaymath}
\end{definition}
\begin{proposition}[Normalized Generalized $\UV$-code]\label{prop:genUV}
  Let $n$ be an even integer and let
  $\varphi=\varphi_{\av,\bv,\cv,\dv}$ be a UV-normalized mapping. The
  mapping $\varphi$ is bijective with
  \begin{displaymath}
    \varphi^{-1}(\xv,\yv) = (\dv \hsp \xv -\bv \hsp \yv,-\cv \hsp \xv + \av \hsp \yv).
  \end{displaymath}
  For any two subspaces $U$ and $V$ of $\Fq^{n/2}$ of parity check
  matrices $\Hm_U$ and $\Hm_V$, the vector space $\varphi(U,V)$ is
  called a {\em normalized generalized $\UV$-code}. It has dimension
  $\dim U + \dim V$ and admits the following parity check matrix
  \begin{equation}\label{eq:pcmUV}
    \Hc(\varphi,\Hm_U,\Hm_V) \eqdef \begin{pmatrix}[r|r]
      \Hm_U \Dm & - \Hm_U \Bm\\ \hline
      - \Hm_V \Cm &  \Hm_V \Am
    \end{pmatrix}
  \end{equation}
  where $\Am \eqdef \Diag(\av)$, $\Bm \eqdef \Diag(\bv)$,
  $\Cm \eqdef \Diag(\cv)$ and $\Dm\eqdef \Diag(\dv)$.
\end{proposition} 
In the sequel, a UV-normalized mapping $\varphi$ implicitly
defines a quadruple of vectors $(\av,\bv,\cv,\dv)$ such that
$\varphi=\varphi_{\av,\bv,\cv,\dv}$. We will use this implicit
notation and drop the subscript whenever no ambiguity may arise.

\begin{remark}
\begin{itemize}
\item This construction can be viewed as taking two codes of length
  $n/2$ and making a code of length $n$ by ``mixing'' together a
  codeword $\uv$ in $U$ and a codeword $\vv$ in $V$ as the vector
  formed by the  set of $a_i u_i + b_i v_i$'s and
  $c_i u_i + d_i v_i$'s.
\item The condition $a_i c_i \neq 0$ is here to ensure that
  coordinates of $U$ appear in all the coordinates of the normalized
  generalized $\UV$ codeword. This is essential for having a decoding
  algorithm for the generalized $\UV$-code that has an advantage over
  standard information set decoding algorithms for linear codes. The
  trapdoor of our scheme builds upon this advantage. It can really be
  viewed as the ``interesting'' generalization of the standard $\UV$
  construction.
\item We have fixed $a_{i}d_{i} - b_{i}c_{i} = 1$ for every $i$ to
  simplify some of the expressions in what follows.  It is readily
  seen that any generalized $\UV$-code that can be obtained in the
  more general case $ a_{i}d_{i} - b_{i}c_{i} \neq 0$ can also be
  obtained in the restricted case $a_{i}d_{i} - b_{i}c_{i} = 1$ by
  choosing $U$ and $V$ appropriately.
\end{itemize}
\end{remark}

\subsubsection{Defining $\trap$ and $\sampPre$.}
From the security parameter $\lambda$, we derive the system parameters
$n,k,w$ and split $k=k_U+k_V$ as described in \S\ref{sec:ch_params}.
The secret key is a tuple $\sk=(\varphi,\Hm_U,\Hm_V,\Sm,\Pm)$ where
$\varphi$ is a UV-normalized mapping,
$\Hm_U\in\Fq^{(n/2-k_U)\times n/2}$,
$\Hm_V\in\Fq^{(n/2-k_V)\times n/2}$, $\Sm\in\Fq^{(n-k)\times (n-k)}$
is non-singular with $k=k_U+k_V$, and $\Pm\in\Fq^{n\times n}$ is a
permutation matrix. Each element of $\sk$ is chosen randomly and
uniformly in its domain.

From $(\varphi,\Hm_U,\Hm_V)$ we derive the parity check matrix
$\Hsec=\Hc(\varphi,\Hm_U,\Hm_V)$ as in
Proposition~\ref{prop:genUV}. The public key is $\Hpub=\Sm\Hsec\Pm$.
Next, we need to produce an algorithm $D_{\varphi,\Hm_U,\Hm_V}$ which
inverts $f_{w,\Hsec}$. The parameter $w$ is such that this can be
achieved using the underlying $\UV$ structure while the generic
problem remains hard. In \S\ref{sec:rejSampl} we will show how to use
rejection sampling to devise $D_{\varphi,\Hm_U,\Hm_V}$ such that its
output is uniformly distributed over $S_w$ when $\sv$ is uniformly
distributed over $\Fq^{n-k}$.  This enables us to instantiate
algorithm $\sampPre$. To summarize:
\begin{displaymath}  
\left.
  \begin{array}{rcl}
    \sk & \gets & (\varphi,\Hm_U,\Hm_V,\Sm,\Pm) \\
    \pk & \gets & \Hpub \\
    \left(\pk,\sk\right)  & \gets & \trap(\lambda)
  \end{array}~~~\right|~~~
  \begin{array}{l}
     \sampPre(\sk,\sv)   \\
    \quad\ev \leftarrow D_{\varphi,\Hm_U,\Hm_V}(\sv\transpose{\left(\Sm^{-1}\right)})\\
    \quad \texttt{return}~ \ev\Pm
  \end{array} 
\end{displaymath}
As in \cite{GPV08}, putting this together with a domain sampling
condition --which we prove in \S\ref{sec:domSampl} from a variation of
the left-over hash lemma-- allows us to define a family of trapdoor
preimage sampleable functions, later referred to as the Wave-PSF
family.

 \section{Inverting the Syndrome Function}
\label{sec:trapdoor}

This section is devoted to the inversion of $f_{w,\Hm}$. It amounts to solve the following problem.
\begin{problem}[Syndrome Decoding with fixed weight]
  \label{prob:CSD}
  Given $\Hm\in\Fq^{(n-k)\times n}$, $\sv\in\Fq^{n-k}$, and an integer
  $w$, find $\ev\in\Fq^n$ such that $\ev\transpose{\Hm}={\sv}$ and $\wt{\ev}=w$.
\end{problem}
We consider three nested intervals $\IInt{\wme}{\wpe} \subset
\IInt{\wUVm}{\wUVp} \subset \IInt{\wm}{\wp}$ for $w$ such that for $\sv$
randomly chosen in $\Fq^{n-k}$:
\begin{itemize}\vspace{-1em}
\item $f^{-1}_{w,\Hm}(\sv)$ is likely/very likely to exist if $w\in
  \IInt{\wm}{\wp}$ (Gilbert-Varshamov bound)
\item $\ev\in f^{-1}_{w,\Hm}(\sv)$ is easy to find if $w\in
  \IInt{\wme}{\wpe}$ for all $\Hm$ (Prange algorithm)
\item $\ev\in f^{-1}_{w,\Hm}(\sv)$ is easy to find if
  $w\in \IInt{\wUVm}{\wUVp}$ and $\Hm$ is the parity check matrix of a
  generalized $\UV$-code. This is the key for exploiting
  the underlying $\UV$ structure as a trapdoor for
   inverting $f_{w,\Hm}$.
\end{itemize}

\subsection{Surjective Domain of the Syndrome Function}
The issue is here for which value of $w$ we may expect that $f_{w,\Hm}$
is surjective. This clearly implies that $|S_w| \geq q^{n-k}$.  In
other words we have:
\begin{fact}
	\label{fac:lower_bound}
	If $f_{w,\Hm}$ is surjective, then $w \in \IInt{\wm}{\wp}$ where $\wm<\wp$ are the extremum of the set $\left\{ w \in \llbracket 0,n \rrbracket\mid\binom{n}{w}(q-1)^{w} \geq q^{n-k} \right\}.$
\end{fact}
For a fixed rate $R=k/n$, let us define 
$
\omegam \eqdef \mathop{\lim}\limits_{n \to + \infty} \wm/n$ and $\omegap \eqdef \mathop{\lim}\limits_{n \to + \infty} \wm/n.
$
Note that $\omegam$ is known as the asymptotic Gilbert-Varshamov
distance. A straightforward computation of the expected number of
errors $\ev$ of weight $w$ such that $\ev\transpose{\Hm} = {\sv}$ when
$\Hm$ is random shows that we expect an exponential number of
solutions when $w/n$ lies in $(\omegam,\omegap)$.
However, coding theory has never come up with an efficient algorithm for finding
a solution to this problem in the whole range $(\omegam,\omegap)$.

\subsection{Easy Domain of the Syndrome Function} \label{subsec:prangeStep}

The subrange of $(\omegam,\omegap)$ for which we know how to solve
efficiently Problem \ref{prob:CSD} is given by the condition
$w/n \in [\omegame,\omegape]$ where
\begin{eqnarray}
\omegame & \eqdef & \frac{q-1}{q} (1-R) \quad \mbox{and} \quad \omegape \eqdef  \frac{q-1}{q} + \frac{R}{q}, 
\end{eqnarray}	 
where $R \eqdef \frac{k}{n}$.  This is achieved by a sightly
generalized version of the Prange decoder \cite{P62}.  We want to find
for a given $\sv$ and error $\ev$ of weight $w$ such that
$\ev\transpose{\Hm} = {\sv}$. The matrix $\Hm$ is a
full-rank matrix and it therefore contains an invertible submatrix
$\Am$ of size $(n-k)\times (n-k)$.  We choose a set of positions $\cI$
of size $n-k$ for which $\Hm$ restricted to these positions is a full
rank matrix. For simplicity assume that this matrix is in the first
$n-k$ positions: $\Hm = \begin{pmatrix} \Am | \Bm\end{pmatrix}$. We
look for an $\ev$ of the form $\ev = (\ev'',\ev')$ where
$\ev' \in \Fq^{k}$ and $\ev'' \in \Fq^{n-k}$. We should therefore have
${\ev''} = ({\sv} - \ev'\transpose{\Bm})\transpose{(\Am^{-1})}$. In
this way we can arbitrarily choose the error $\ev'$ of length
$k$ but in any case we expect for the remaining part a vector $\ev''$ with about 
$\frac{q-1}{q}(n-k)$ positions that are non zero. Therefore, the
weights that are easily attainable by this strategy are between
$\frac{q-1}{q}(n-k) = n \omegame$ and
$k + \frac{q-1}{q}(n-k) = n \omegape$ by choosing appropriately the weight 
of $\ev'$ between $0$ and $k$. This procedure, that we call
$\Call{PrangeOne}{\cdot}$, is formalized in
Algorithm~\ref{algo:Prangesdd}.

\begin{algorithm}[htb]
  \caption{\calltxt{PrangeOne}{$\Hm,\sv$} --- One iteration of the Prange decoder}\label{algo:Prangesdd}
  Parameters: $q,n,k$, $\Dc$ a distribution over $\llbracket 0,k\rrbracket$
  \begin{algorithmic}[1]
    \hrule
    \Require $\Hm\in\Fq^{(n-k)\times n}$, $\sv\in\Fq^{n-k}$
    \Ensure $\ev\transpose{\Hm}=\sv$
    \State $t\Unif\Dc$
    \State $\Ic\gets\Call{InfoSet}{\Hm}$
    \Comment {{\em \Call{InfoSet}{$\Hm$} returns an information set of $\vectspace{\Hm}^\perp$}}
    \State $\xv\Unif\{\xv\in\Fq^n\mid\wt{\xv_\Ic}=t\}$
    \State $\ev\gets\Call{PrangeStep}{\Hm,\sv,\Ic,\xv}$
    \State \Return $\ev$
  \end{algorithmic}
  \smallskip
  
  \hrule
  {\bf function} \Call{PrangeStep}{$\Hm,\sv,\Ic,\xv$} --- Prange vector completion
  \hrule
  \begin{algorithmic}
    \Require $\Hm\in\Fq^{(n-k)\times n}$, $\sv\in\Fq^{n-k}$, $\Ic$ an
    information set of $\vectspace{\Hm}^\perp$, $\xv\in\Fq^n$
    \Ensure $\ev\transpose{\Hm}=\sv$ and $\ev_\Ic=\xv_\Ic$
    \State $\Pm\gets$ any $n\times n$ permutation matrix sending $\Ic$ on the last
    $k$ coordinates
    \State $(\Am\mid\Bm)\gets \Hm\Pm$
    \Comment $\Am\in\Fq^{(n-k)\times(n-k)}$
    \State $(\zero\mid\ev')\gets \xv$
    \Comment $\ev'\in\Fq^{k}$
    \State $\ev\gets\left(\left(\sv - 
        \ev'\tran{\Bm}\right)\tran{\left(\Am^{-1}\right)},\ev'\right)\tran{\Pm}$
    \State \Return $\ev$
  \end{algorithmic}
\end{algorithm}

\begin{proposition} \label{propo:Prange} When $\Hm$ is chosen
  uniformly at random in $\Fq^{(n-k)\times n}$ and $\sv$ uniformly at
  random in $\Fq^{n-k}$, for the output $\ev$
  of \calltxt{PrangeOne}{$\Hm,\sv$} we have
  $$
  |\ev| = S+T
  $$
  where $S$ and $T$ are independent random variables,
  $S \in \IInt{0}{n-k}$, $T \in \IInt{0}{k}$, $S$ is the Hamming
  weight of a vector that is uniformly distributed over $\Fq^{n-k}$
  and $\prob(T=t) = \cD(t)$. The distribution of $|\ev|$ is given by
  \begin{eqnarray*}
    \prob\left(|\ev|=w \right) & = & \sum_{t=0}^{w} \frac{\binom{n-k}{w-t}(q-1)^{w-t}}{q^{n-k}} \cD(t),\quad  \esp(|\ev|) =  \overline{\cD} + \textstyle{\frac{q-1}{q}} (n-k) = \overline{\cD} + n \omegame\label{eq:probaprange}
  \end{eqnarray*}
  where $\overline{\cD} = \sum_{t=0}^k t\cD(t)$.
\end{proposition}
From this proposition, we deduce immediately that any weight $w$ in
$\IInt{\omegame n}{\omegape n}$ can be reached by this Prange decoder
with a probabilistic polynomial time algorithm that uses a
distribution $\cD$ such that $\overline{\cD} = w - \omegame n$ and which is sufficiently concentrated around its expectation. It
will be helpful in what follows to be able to choose a probability
distribution $\cD$ as this gives a rather large degree of freedom in
the distribution of $|\ev|$ that will come very handy to simulate
an output distribution that is uniform over the words of weight $w$ in
the generalized  $\UV$-decoder that we will consider in what
follows.

To summarize this discussion we have shown that when we want to ensure that $f_{\Hm}$ is surjective, $w$ has to verify $\wm \leq w \leq \wp$. However, in a cryptographic setting $w/n$ cannot
lie in $[\omegame,\omegape] \subseteq [\omegam,\omegap]$ otherwise
anybody that uses the generalized Prange algorithm would be able to
invert $f_{\Hm}$. All of this is summarized in Figure
\ref{fig:distSgn} where we draw the above different areas
asymptotically in $n$ of $w/n$ when $k/n$ is fixed.

\begin{figure}
	\caption{Areas of relative signature distances. \label{fig:distSgn}} 
	\centering
	\includegraphics[height=10cm]{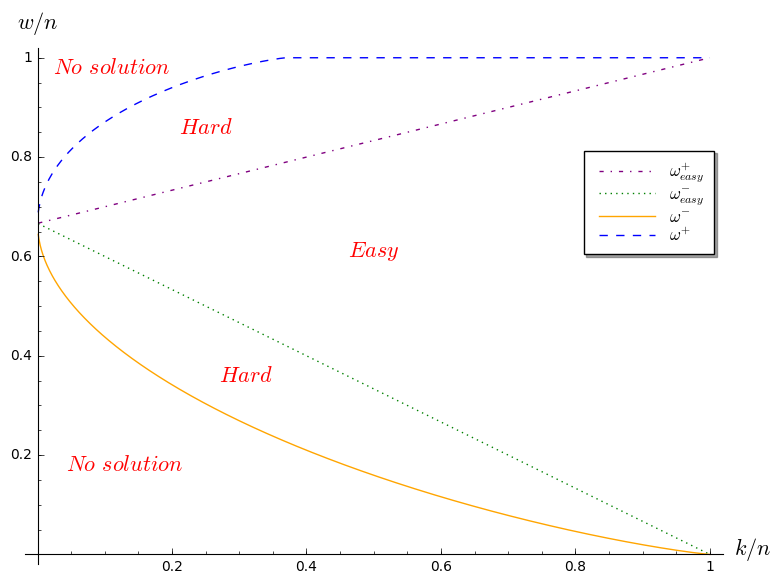}
\end{figure}

\subsubsection{Enlarging the Easy Domain $\IInt{\wme}{\wpe}$.}

Inverting the syndrome function $f_{w,\Hm}$ is the basic problem upon which all code-based cryptography relies. This problem has been studied for a long time for relative weights $\omega \eqdef \frac{w}{n}$ in $(0,\omegame)$ and despite many efforts the best algorithms \cite{S88,D91,B97b,MMT11,BJMM12,MO15,DT17,BM18} for solving this problem are all exponential in $n$ for such fixed relative weights. 
In other words, after  
more than fifty years of research, none of those algorithms came up with a polynomial complexity for relative weights 
$\omega$ in $(0, \omegame)$. Furthermore, by adapting all the previous algorithms beyond this point we observe for them the same behaviour: they are all polynomial in the range of relative weights $[\omegame,\omegape]$ and become exponential once again when $\omega$ is in $(\omegape,1)$. All these results point towards the fact 
that  inverting $f_{w,\Hm}$ in polynomial time on a larger range is fundamentally a hard problem.
 In the following subsection we present a trapdoor on the matrices $\Hm$ that enables to invert in polynomial time $f_{w,\Hm}$ on a larger range  by tweaking the Prange decoder.

\subsection{Solution with Trapdoor} \label{subsec:genUVcodes}

Let us recall that our trapdoor to invert $f_{w,\Hm}$ is given by the
family of normalized generalized $\UV$-codes (see Proposition
\ref{prop:genUV} in \S\ref{subsec:waveTrap}). As we will see in what
follows, this family comes with a simple procedure which enables to
invert $f_{w,\Hm}$ with errors of weight which belongs to
$\IInt{\wUVm}{\wUVp} \subset \IInt{\wm}{\wp}$ but with
$\IInt{\wme}{\wpe} \subsetneq \IInt{\wUVm}{\wUVp}$. We summarize this
situation in Figure \ref{fig:rewUV}.

We wish to point out here, to avoid any misunderstanding that led the
authors of \cite{BP18a} to make a wrong claim that they had an attack
on Wave, that the procedure we give here is not the one we use at the
end to instantiate Wave, but is merely here to give the underlying
idea of the trapdoor. Rejection sampling will be needed as explained
in the following section to avoid any information leakage on the
trapdoor coming from the outputs of the algorithm given here.

\begin{figure}[htb]
	\centering
	\begin{tikzpicture}[scale=0.83]
	\tikzstyle{valign}=[text height=1.5ex,text depth=.25ex]
	\draw[line width=2pt,gray] (0,2) -- (1,2);
	\draw (2.5,2.5) node[red]{{\sf hard}};
	\draw (12,2.5) node[right,red]{{\sf hard}};
	\draw (12,2.5) node[right,red]{{\sf hard}};
	\draw[line width=2pt,red!50] (1,2) -- (5,2);
	\draw[line width=2pt,blue!50] (5,2) --
	node[above,midway,blue,valign]{{\sf easy}} (11,2);
	\draw[line width=2pt,red!50] (11,2) -- (13,2);
	\draw[->,>=latex,line width=2pt,gray] (13,2) -- (14,2)
	node[right,black] {$\displaystyle w$};
	\tikzstyle{valign}=[text height=2ex]
	\draw[thick] (1,1.9) node[below,valign]{$0$} -- (1,2.1);
	\draw[thick] (5,1.9) node[below,valign]{$\wme$} -- (5,2.1);
	\draw[thick] (11,1.9) node[below,valign]{$\wpe$~~} -- (11,2.1);
	\draw[thick] (13,1.9) node[below,valign]{$n$} -- (13,2.1);
				
	\draw[thick] (4,1.9) node[below,valign]{$\wUVm$} -- (4,3.1);
	\draw[thick] (11.75,1.9) node[below,valign]{~~$\wUVp$} -- (11.75,3.1);
	\draw[<->,>=latex,thin,blue!50] (4,3) -- node[above,blue,midway]{{\sf
			easy with \UV{} trapdoor}} (11.75,3);
	\draw[<->,>=latex,thin,red!50] (1,3) -- (4,3);
	\draw[<->,>=latex,thin,red!50] (11.75,3) -- (13,3);
	\end{tikzpicture}
	\caption{Hardness of $\UV$ Decoding}
	\label{fig:rewUV}
\end{figure}
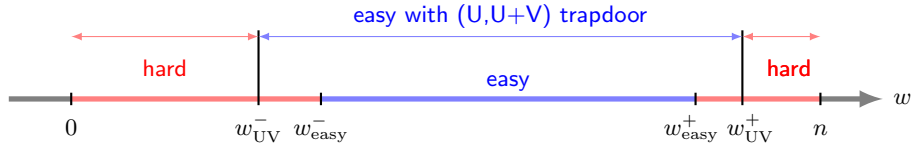

It turns out that in the case of a normalized generalized  $\UV$-code, a simple tweak of
the Prange decoder will be able to reach relative weights $w/n$
outside the ``easy'' region $[\omegame,\omegape]$. It exploits
the fundamental leverage of the Prange decoder : it consists in
choosing the error $\ev$ satisfying $\ev \tran{\Hm} = {\sv}$ as we
want in $k$ positions when the code that we decode is random and of dimension
$k$. When we want an error of low weight, we put zeroes on those
positions, whereas if we want an error of large weight, we put
non-zero values. This idea leads to even smaller or larger weights in the case of a normalized 
generalized 
$\UV$-code. 
To explain this point, recall that we want to solve the following decoding problem in this case.
\begin{problem}[decoding problem for normalized generalized $\UV$-codes]\label{prob:decodingNGUV}
  Given a normalized generalized $\UV$ code $(\varphi,\Hm_U,\Hm_V)$
  (see Proposition \ref{prop:genUV}) of parity-check matrix
  $\Hm = \Hc(\varphi,\Hm_U,\Hm_V)\in\Fq^{(n-k)\times n}$, and a
  syndrome $\sv \in \Fq^{n-k}$, find $\ev \in \Fq^n$ of weight $w$
  such that $\ev \transpose{\Hm} = \sv.$
\end{problem}
The following notation will be very useful to explain how we solve
this problem.
\begin{notation} \label{nota:euv}
For a vector $\ev$ in $\Fq^n$, we denote by $\ev_U$ and $\ev_V$ the vectors in $\Fq^{n/2}$ such that 
$$(\ev_U,\ev_V)=\tphi^{-1}(\ev).$$
\end{notation}
The decoding algorithm we will consider recovers $\ev_V$ and then $\ev_U$. From $\ev_U$ and $\ev_V$ we recover 
$\ev$ since $\ev=\tphi(\ev_U,\ev_V)$. The point of introducing such an $\ev_U$ and a $\ev_V$ is that
\begin{restatable}{proposition}{prop:decomposition}
\label{prop:decomposition}
Solving the decoding problem \ref{prob:decodingNGUV} is equivalent to find an $\ev \in \Fq^n$ of weight 
$w$ satisfying
\begin{eqnarray}
  {\ev_U} \tran{\Hm}_U & = & {\sv^U}  \label{eq:U}\\ 
  {\ev_V} \tran{\Hm}_V & = & {\sv^V} \label{eq:V}
\end{eqnarray}
where $\sv = (\sv^U,\sv^V)$ with $\sv^U \in \Fq^{n/2-k_U}$ and $\sv^V \in \Fq^{n/2-k_V}$. 
\end{restatable}
\begin{remark}
We have put $U$ and $V$ as superscripts in $\sv^U$ and $\sv^V$ to avoid any confusion with the notation we have just introduced for 
$\ev_U$ and $\ev_V$.
\end{remark}
\begin{proof}
Let us observe that
$
\ev =  \tphi(\ev_U,\ev_V)
 =  (\av \hsp \ev_U+ \bv \hsp \ev_V,\cv \hsp \ev_U + \dv \hsp \ev_V)
=  (\ev_U \Am + \ev_V \Bm, \ev_U \Cm + \ev_V \Dm)
$
with $\Am = \Diag(\av),\Bm = \Diag(\bv),\Cm = \Diag(\cv), \Dm = \Diag(\dv)$.
By using this, $\ev \transpose{\Hm} = \sv$ translates into
\begin{eqnarray*}
\left\{
\begin{array}{lcr}
\ev_U \Am \transpose{\Dm} \transpose{\Hm}_U + \ev_V \Bm \transpose{\Dm} \transpose{\Hm}_U  -
\ev_U \Cm  \transpose{\Bm} \transpose{\Hm}_U - \ev_V \Dm \transpose{\Bm} \transpose{\Hm}_U & = & \sv^U\\
-\ev_U \Am \transpose{\Cm} \transpose{\Hm}_V - \ev_V \Bm \transpose{\Cm} \transpose{\Hm}_V  +
\ev_U \Cm  \transpose{\Am} \transpose{\Hm}_V + \ev_V \Dm \transpose{\Am} \transpose{\Hm}_V & = & \sv^V
\end{array}
\right.
\end{eqnarray*}
which amounts to 
$
\ev_U (\Am \Dm - \Bm \Cm)\transpose{\Hm}_U   =  \sv^U$ and $
\ev_V (\Am \Dm - \Bm \Cm) \transpose{\Hm}_V  =  \sv^V
$, since $\Am$, $\Bm$, $\Cm$, $\Dm$ are diagonal matrices, they are therefore symmetric and commute
with each other. We finish the proof by observing that $\Am \Dm - \Bm \Cm = \Imat_{n/2}$, the identity matrix of 
size $n/2$. \qed 
\end{proof}
Performing the two decoding
\eqref{eq:U} and \eqref{eq:V} independently with the Prange algorithm
gains nothing. However if we first solve \eqref{eq:V} with the Prange
algorithm, and then seek a solution of \eqref{eq:U} which properly
depends on $\ev_V$ we increase the range of weights accessible in polynomial time
for $\ev$.  It then turns out that the range $[\omegaUVm,\omegaUVp]$
of relative weights $w/n$ for which the $\UV$-decoder works in polynomial time is
larger than  $[\omegame,\omegape]$. 
This will
provide an advantage to the trapdoor owner.

\paragraph{Tweaking the Prange Decoder for Reaching Large Weights.}
When $q=2$, small and large weights play a symmetrical role. This is
not the case anymore for $q \geq 3$. In what follows we will suppose that
$
q\geq 3.
$
In order to find a solution $\ev$ of large weight to the decoding problem
$\ev \transpose{\Hm} = \sv$, we use Proposition \ref{prop:decomposition} and first find an 
arbitrary solution $\ev_V$ to $\ev_V \transpose{\Hm_V} = \sv^V$.
The idea, now for performing the second decoding $\ev_U \transpose{\Hm_U} = \sv^U$, 
 is to take advantage of $\ev_V$ to 
find a solution $\ev_U$ that maximizes the weight of $\ev=\tphi(\ev_U,\ev_V)$.
On any information set of the $U$ code, we can fix arbitrarily $\ev_U$.
Such a set is of size $k_U$ and on those positions $i$ we can always choose 
$\ev_U(i)$ such that this induces {\em simultaneously} two positions in $\ev$ that are non-zero. 
These are $\ev_i$ and $\ev_{i+n/2}$. We just have to choose $\ev_U(i)$ so that we have simultaneously
	$$
	\left\{
	\begin{array}{ll} 
		a_i\ev_U(i)+b_i\ev_V(i) \neq  0 \\ 
		c_i\ev_U(i)+d_i\ev_V(i) \neq  0.
	\end{array} 
	\right.
	$$
This is always possible since $q \geq 3$ and it gives an expected weight of $\ev$:
\begin{eqnarray}
\esp(|\ev|) = 2\left( k_U + \frac{q-1}{q}(n/2-k_U)\right) = \frac{q-1}{q} n + \frac{2k_U}{q}
\end{eqnarray}
The best choice for $k_U$ is to take $k_U=k$ up to the point where
$\frac{q-1}{q} n + \frac{2k}{q}=n$, that is $k=n/2$ and for larger values
of $k$ we choose $k_U=n/2$ and $k_V = k-k_U$.

\paragraph{Why Is the Trapdoor More Powerful for Large Weights than for Small Weights?}	 
This strategy can be clearly adapted for small weights. However, it is less powerful in this case.
Indeed, to minimize the weight of the final error we would like to choose $\ev_U(i)$ in $k_U$ positions such that 
	$$
	\left\{
	\begin{array}{ll} 
		a_i\ev_U(i)+b_i\ev_V(i) =  0 \\ 
		c_i\ev_U(i)+d_i\ev_V(i) =  0
	\end{array} 
	\right.
	$$
	Here as $a_id_i - b_ic_i = 1$ and $a_ic_i \neq 0$ in the family of codes we consider, this is possible if and only if $\ev_V(i) = 0$. Therefore, contrarily to the case where we want to reach errors of large weight, the area of positions where we can gain twice is constrained to be of size $n/2 - |\ev_V|$. The minimal weight for $\ev_V$ we can reach in polynomial time with the Prange decoder is given by $\frac{q-1}{q}(n/2-k_V)$. In this way the set of positions where we can double the number of $0$ will be of size $n/2 - \frac{q-1}{q}(n/2-k_V)= \frac{n}{2q} + \frac{q-1}{q}k_V$. It can be verified that this strategy would give the following expected weight for the final error we get:
	$$
	\mathbb{E}(|\ev|) = \left\{
	\begin{array}{ll} 
	\frac{q-1}{q}n - 2\frac{q-1}{q}k_U \quad \mbox{if } k_U \leq \frac{n}{2q} + \frac{q-1}{q}k_V  \\ 
	\frac{2(q-1)^2}{(2q-1)q}(n-k) \quad\mbox{  }\mbox{ else.}
	\end{array} 
	\right.
	$$

This discussion is summarized in Figure \ref{fig:trapDist}
where we draw $\omegaUVm$ and $\omegaUVp$ which are the highest and
the smallest relative distances that our decoder can reach
asymptotically in $n$ when $k/n$ is fixed and $q = 3$.

\begin{figure}[h!]
  \caption{Areas of relative signature distances with our trapdoor when $q = 3$ \label{fig:trapDist}}
  \centering
  \includegraphics[width=10cm]{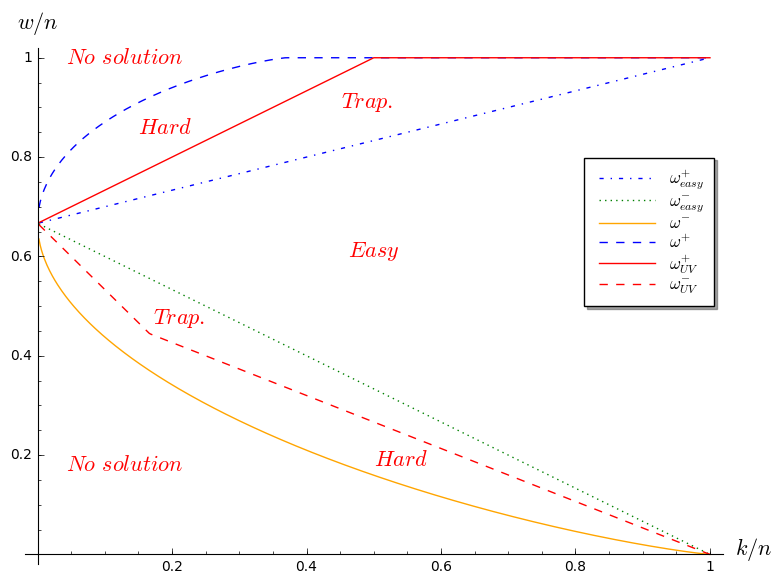}
\end{figure}

 \section{Preimage Sampling with Trapdoor: Achieving a Uniformly Distributed Output}
	\label{sec:rejSampl}
	
	 We restrict here our study to the case $q=3$ but it can be generalized to larger values of $q$. To be a trapdoor one-way preimage sampleable function, we have to enforce that
	the outputs of our algorithm, which inverts our trapdoor function, are very close to be uniformly distributed over $S_w$.
	The procedure described in the previous section using directly the Prange decoder, does
	not meet this property. As we will prove, by changing it slightly, we will
	achieve this task by still keeping the property to output errors of
	weight $w$ for which it is hard to solve the decoding problem for this
	weight. However, the parameters will have to be chosen carefully and the area of weights $w$ for which we can output errors in polynomial time decreases. Figure \ref{fig:uniDimRej} gives a rough picture of what will happen.

	\begin{figure}[htb]
		\centering
		\begin{tikzpicture}[scale=0.83]
		\tikzstyle{valign}=[text height=1.5ex,text depth=.25ex]
		\draw[line width=2pt,gray] (0,2) -- (1,2);
		\draw (2.5,2.5) node[red]{{\sf hard}};
		\draw (12,2.5) node[right,red]{{\sf hard}};
		\draw (12,2.5) node[right,red]{{\sf hard}};
		\draw[line width=2pt,red!50] (1,2) -- (5,2);
		\draw[line width=2pt,blue!50] (5,2) --
		node[above,midway,blue,valign]{{\sf easy}} (11,2);
		\draw[line width=2pt,red!50] (11,2) -- (13,2);
		\draw[->,>=latex,line width=2pt,gray] (13,2) -- (14,2)
		node[right,black] {$\displaystyle w$};
		\tikzstyle{valign}=[text height=2ex]
		\draw[thick] (1,1.9) node[below,valign]{$0$} -- (1,2.1);
		\draw[thick] (5,1.9) -- node[above,valign]{$\wme$} (5,2.1);
		\draw[thick] (11,1.9) -- node[above,valign]{$\wpe$} (11,2.1);
		\draw[thick] (13,1.9) node[below,valign]{$n$} -- (13,2.1);
								
		\draw[thick] (4,1.9) node[below,valign]{$\wUVm$} -- (4,3.1);
		\draw[thick] (11.75,1.9) node[below,valign]{\mbox{ }$\wUVp$} -- (11.75,3.1);
		\draw[<->,>=latex,thin,blue!50] (4,3) -- node[above,blue,midway]{{\sf
				easy with \UV{} trapdoor}} (11.75,3);
		\draw[<->,>=latex,thin,red!50] (1,3) -- (4,3);
		\draw[<->,>=latex,thin,red!50] (11.75,3) -- (13,3);
		\draw[<->,>=latex,thin,purple!50] (5,1) -- node[below,purple,midway]{{\sf	no leakage with $\UV{}$ trapdoor}} (11.4,1);
		\draw[thick] (5,0.9) node[below,valign]{$ $} -- (5,2.1);
		\draw[thick] (11.4,0.9) node[below,valign]{$ $} -- (11.4,2.1);
		\end{tikzpicture}
		\caption{Hardness of $\UV$ Decoding with no leakage of signature}
		\label{fig:uniDimRej}
	\end{figure}
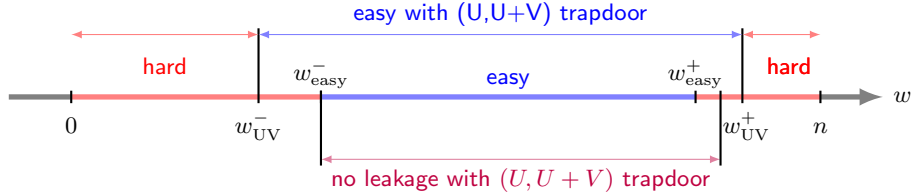

	\subsection{Rejection Sampling to reach Uniformly Distributed Output} 
	\label{subsec:rej}
	
We will tweak slightly the generalized $\UV$-decoder from the previous section by performing in particular rejection 
sampling on $\ev_U$ and $\ev_V$ in order to obtain an error $\ev$ satisfying $\ev \transpose{\Hm} = \sv$ that is uniformly distributed over the words of weight $w$ when the syndrome
$\sv$ is randomly chosen in $\F_3^{n-k}$. Solving the decoding problem \ref{prob:decodingNGUV} of the generalized $\UV$-code will be done by solving 
\eqref{eq:U} and \eqref{eq:V} through an algorithm whose skeleton is given in Algorithm \ref{algo:realskeleton}.
$\DV{\Hm_V,\sv^V}$ returns a vector satisfying $\ev_V \transpose{\Hm_V} = \sv^V$, whereas
$\DU{\Hm_U,\varphi,\sv^U,\ev_V}$ is assumed to return a vector satisfying $\ev_U \transpose{\Hm_U} = \sv^U$ {\em and} such that $|\tphi(\ev_U,\ev_V)|=w$.
Here $\sv = (\sv^U,\sv^V)$ with $\sv^U \in \F_3^{n/2-k_U}$ and $\sv^V \in \F_3^{n/2-k_V}$.

	\begin{algorithm}[htb]
		\caption{\DUVt{$\Hm_V,\Hm_U,\varphi,\sv$}}
		\label{algo:realskeleton}
		\begin{algorithmic}[1]
			\Repeat
			\State $\ev_V\gets\DV{\Hm_V,\sv^V}$ \label{ske:ev}
			\Until{Condition 1 is met}\label{skerej:V}
			\Repeat
			\State $\ev_U \gets\DU{\Hm_U,\varphi,\sv^U,\ev_V}$ \label{ske:U}\Comment{We assume that $|\tphi(\ev_U,\ev_V)|=w$ here.}
			\State $\ev \gets \tphi(\ev_U,\ev_V)$ \label{ske:e}
			\Until{Condition 2 is met}\label{skerej:U}
			\State \Return $\ev$
		\end{algorithmic}
	\end{algorithm} 

What we want to achieve by rejection sampling is that the distribution of $\ev$ output by this algorithm is the same as
the distribution of $\evu$ that denotes a vector that is chosen uniformly at random among the words of weight $w$ in $\F_3^n$.
This will be achieved by ensuring that
\begin{enumerate}
\item the $\ev_V$ fed into $\DU{\cdot}$ at Step \ref{ske:U} has the same distribution as $\evu_V$,
\item the distribution of $\ev_U$  surviving to Condition 2 at Step \ref{skerej:U} conditioned on the value of $\ev_V$
is the same as the distribution of $\evu_U$ conditioned on $\evu_V$.
\end{enumerate}

There is a property of the decoders $\DV{\cdot}$ and $\DU{\cdot}$ derived from Prange decoders that we will consider that will be very helpful here. They will namely be very close to meet the following conditions.
\begin{definition}\label{def:weightU}
$\DV{\cdot}$ is said to be weightwise uniform if the output $\ev_V$ of $\DV{\Hm_V,\sv^V}$ is such that 
$\prob(\ev_V)$ is just a function of $|\xv|$
		when  $\sv^V$ is
		chosen uniformly at random in $\F_3^{n/2-k_V}$.
$\DU{\cdot}$ is $\lw$-uniform if the outputput $\ev_U$ of $\DU{\Hm_U,\varphi,\sv^U,\ev_V}$ satisfies that the conditional probability $\prob(\ev_U|\ev_V)$ is
just a function of the pair $(|\ev_V|,\lw(\varphi(\ev_U,\ev_V))$ where
	\begin{displaymath}
		\lw(\xv) \eqdef \left|\left\{ 1 \leq i \leq n/2 : |(x_i,x_{i+n/2})| = 1 \right\}\right|.
		\end{displaymath}
\end{definition}
It is readily observed that $\prob(\evu_V)$ and $\prob(\evu_U|\evu_V)$ are also only functions of $|\evu_V|$ and $(|\evu_V|,\lw(\evu))$ respectively. From this it is readily seen that we obtain the
right distributions for $\ev_V$ and $\ev_U$ conditioned on $\ev_V$ by just ensuring that 
the distribution of $|\ev_V|$ follows the same distribution as $|\evu_V|$ and that the distribution of $\lw(\ev)$ conditioned on $|\ev_V|$ is the 
same as the distribution of $\lw(\evu)$ conditioned on $|\evu_V|$. This is shown by the following lemma.

	\begin{lemma}\label{lemm:rejSampl} Let $\ev$ be the output of Algorithm \ref{algo:realskeleton} when $\sv^V$ and $\sv^U$ are chosen uniformly at random in $\F_3^{n/2-k_V}$ and 
		$\F_3^{n/2-k_U}$ respectively. 
		Assume that $\DU{\cdot}$ is $\lw$-uniform whereas $\DV{\cdot}$ is weightwise uniform. 
		If 	for any possible $y$ and $z$, 
		\begin{equation} 
		|\ev_V| \sim |\evu_V| \mbox{ and } 
		\prob(\lw(\ev)=z \mid |\ev_V| =y)=\prob(\lw(\evu)=z \mid |\evu_V|=y)
		\end{equation} 
		then 
		$
		\ev \sim \evu.
		$ The probabilities are taken here over the choice of $\sv^U$ and $\sv^V$ and over the internal coins
of $\DU{\cdot}$ and $\DV{\cdot}$.
	\end{lemma}

	\begin{proof} 
			We have for any $\xv$ in $S_w$
\begin{eqnarray}
\prob(\ev = \xv)  &=&\prob(\ev_U=\xv_U\mid\ev_V=\xv_V)\prob(\ev_V=\xv_V) \nonumber\\
  & = & \prob(\DU{\Hm_U,\varphi,\sv^U,\ev_V}=\xv_U\mid\ev_V=\xv_V)\prob(\DV{\Hm_V,\sv^V}=\xv_V)
        \nonumber \\
& = & \frac{\prob(\lw(\ev)=z \mid
      |\ev_V|=y)}{n(y,z)}\frac{\prob(|\ev_V|=y)}{n(y)} \eqdef P \label{eq:uniformity}
\end{eqnarray}
where $n(y)$ is the number of vectors of $\F_3^n$ of weight $y$ and $n(y,z)$ is the number of vectors $\ev$ 
in $\F_3^n$ such that $\ev_V=\xv_V$ and such that $\lw(\ev)=z$ (this last number only depends on $\xv_V$ through
its weight $y$). Equation \eqref{eq:uniformity} is here a consequence of the weightwise uniformity of $\DV{\cdot}$ on one hand and 
the $\lw$-uniformity of $\DU{\cdot}$ on the other hand.
We conclude by noticing that
\begin{eqnarray} 
P & = & 
\frac{\prob(\lw(\evu)=z \mid |\evu_V|=y)}{n(y,z)}\frac{\prob(|\evu_V|=y)}{n(y)}\label{eq:passage} \\
& = & \prob(\evu_U=\xv_U\mid\evu_V=\xv_V)\prob(\evu_V=\xv_V)\nonumber \\
& =& \prob(\evu = \xv).
\end{eqnarray}
Equation \eqref{eq:passage} follows from the assumptions on the distribution of $|\ev_V|$ and of the conditional distribution
of $\lw(\ev)$ for a given weight $|\ev_V|$. 
 \qed	\end{proof} 

This shows that in order to obtain that $\ev$ is uniformly distributed over $S_w$ it is enough to perform rejection sampling based on the weight $|\ev_V|$ for 
$\DV{\cdot}$ and based on the pair $(|\ev_V|,\lw(\ev))$ for $\DU{\cdot}$. In other words, our decoding algorithm with rejection sampling will use a rejection vector $\rv_V$ on the weights of
$\ev_V$ for $\DV{\cdot}$ and a two-dimensional rejection vector $\rv_U$ for the values of $(|\ev_V|,\lw(\ev))$ for $\DU{\cdot}$. The corresponding algorithm is specified 
in Algorithm \ref{algo:skeleton}.

	\begin{algorithm}[htb]
		\caption{\DUVt{$\Hm_V,\Hm_U,\tphi,\sv$}}
		\label{algo:skeleton}
 		\begin{algorithmic}[1]
			\Repeat
 			\State $\ev_V\gets\DV{\Hm_V,\sv^V}$ \label{alg:ev}
			\Until{rand$([0,1]) \leq \rv_{V}(|\ev_V|)$}\label{rej:V}
 			\Repeat \label{alg:U}
			\State $\ev_U \gets\DU{\Hm_U,\varphi,\sv^U,\ev_V}$
 			\State $\ev \gets \tphi(\ev_U,\ev_V)$ \label{alg:e}
			\Until{rand$([0,1]) \leq \rv_U(|\ev_V|,\lw(\ev))$}\label{rej:U}
			\State \Return $\ev$
		\end{algorithmic}
	
 	\end{algorithm} 
Standard results on rejection sampling yield the following proposition:
	\begin{restatable}{proposition}{propoRejUnif} \label{propo:rejDistrib}
Let \begin{equation} 
		q_{1}(i) \eqdef \prob\left(  |\ev_V| = i \right) \mbox{ };\mbox{ } \qu_1(i) \eqdef \prob\left(  |\evu_V| = i \right)
		\end{equation} 
		\begin{equation} 
		q_{2}(s,t) \eqdef \prob\left( \lw(\ev) = s \mid  |\ev_V| = t \right) \mbox{ };\mbox{ } \qu_{2}(s,t) \eqdef \prob\left( \lw(\evu) = s \mid  |\evu_V| = t \right)
		\end{equation} 
		for any $i,t \in \llbracket 0,n/2 \rrbracket$ and $s \in \llbracket 0,t \rrbracket$.
		Let $\rv_{V}$ and $\rv_{U}$ be defined as 
		\begin{displaymath}
		r_{V}(i) \eqdef \frac{1}{\Mrs_V} \frac{\qu_1(i)}{q_{1}(i)} \quad \mbox{and} \quad r_{U}(s,t) \eqdef \frac{1}{\Mrs_U(t)} \frac{\qu_{2}(s,t)}{q_{2}(s,t)} 
		\end{displaymath}
		with 
		\begin{equation*}
		\Mrs_V  \eqdef \mathop{\max}\limits_{\substack{0 \leq i \leq n/2}}\frac{\qu_1(i)}{q_1(i)} \quad \mbox{and} \quad \Mrs_U(t) \eqdef \mathop{\max}\limits_{\substack{0 \leq s \leq t}}\frac{\qu_2(s,t)}{q_2(s,t)}
		\end{equation*}
		Then if $\DV{\cdot}$ is weightwise uniform and $\DU{\cdot}$ is $\lw$-uniform, the output $\ev$ of Algorithm \ref{algo:skeleton}
		satisfies
		$
		\ev \sim \evu.
		$
	\end{restatable}

	\subsection{Application to the Prange Decoder} \label{subsec:prangeDecUV}

	To instantiate rejection sampling, we have to provide here $(i)$ how $\DV{\cdot}$ and $\DU{\cdot}$ are instantiated and $(ii)$ how $\qu_1,\qu_2, q_1$ and $q_2$ are computed.
Let us begin by the following proposition which gives $\qu_1$ and $\qu_2$.	
		\begin{restatable}{proposition}{propoqu} \label{propo:qu} Let $n$ be an even integer, $w \leq n$, $i,t \leq n/2$ and $s \leq t$ be integers. We have,	
		\begin{equation} 
 		\qu_1(i) = \frac{\binom{n/2}{i}}{\binom{n}{w}2^{w/2}} \mathop{\sum}\limits_{\substack{p=0 \\
 				w+p \equiv 0 \mod 2}}^{i}\binom{i}{p}\binom{n/2-i}{(w+p)/2-i}2^{3p/2}
		\end{equation}
		\begin{equation}
		\qu_2(s,t) =  \left\{
		\begin{array}{ll}
		\frac{\binom{t}{s}\binom{n/2 - t}{\frac{w+s}{2}-t}2^{\frac{3s}{2}}}{\sum\limits_{p} \binom{t}{p}\binom{n/2-t}{\frac{w+p}{2}-t}2^{\frac{3p}{2}}} &\mbox{if } w+s \equiv 0 \mod 2. \\
		0  &\mbox{ else} 
		\end{array}
		\right.
		\end{equation} 
		
	\end{restatable}

	The proof of this proposition is given in Appendix \ref{app:usefulDistribs}.  Algorithms $\DV{\cdot},\DU{\cdot}$ are described in Algorithms \ref{algo:DV} and $\ref{algo:DU}$. They use the rejection vectors given in Proposition \ref{propo:rejDistrib} which is 
	based on the expressions given in Proposition \ref{propo:qu}.

	\begin{algorithm}\label{algoV}
		\caption{\calltxt{DecodeV}{$\Hm_V,\sv^V$} the Decoder outputting an $\ev_V$ such that $\ev_V \Hm_V^{\intercal}=\sv^V$. \label{algo:DV}}
		\begin{algorithmic}[1]
			\State\label{D_V:J}$\cJ,\cI \gets\Call{FreeSet}{\Hm_V}$
						\State $\ell\Unif\Dc_V$
			\State\label{DV_:x}$\xv_V\Unif\left\{\xv\in\F_3^{n/2}\mid\wt{\xv_{\cJ}}=\ell,\Sp(\xv) \subseteq \cI \right\}$
			\Comment  $(\xv_V)_{\cI \mbox{\textbackslash} \cJ}$ is random
			\State $\ev_V \gets\Call{PrangeStep}{\Hm_V,\sv^V,\cI,\xv_V}$ \label{line:eVoutput}
			\State \Return $\ev_V$
		\end{algorithmic}
	 \smallskip	
	\hrule
	{\bf function} \Call{FreeSet}{$\Hm$} 
	\hrule
	\begin{algorithmic}[1]
		\Require $\Hm\in\F_3^{(n-k)\times n}$
		\Ensure $\Ic$ an
		information set of $\vectspace{\Hm}^\perp$ and $\cJ \subset \cI$ of size $k - d$
		\Repeat 
		\State $\cJ \Unif \llbracket 1,n \rrbracket$ of size $k - d$
		\Until the rank of the columns of $\Hm$ indexed by $ \llbracket 1,n \rrbracket \mbox{\textbackslash} \cJ$ is $n-k$
		\Repeat 
		\State  $\cJ' \Unif  \llbracket 1,n \rrbracket \mbox{\textbackslash} \cJ$ of size $d$
		\State $\cI \leftarrow \cJ \sqcup \cJ'$ 
		\Until $\cI$ is an information set of $\vectspace{\Hm}^\perp$
		\State \Return $\cJ,\cI$
	\end{algorithmic}
\end{algorithm}

	\begin{algorithm}\label{algoU}
		\caption{\calltxt{DecodeU}{$\Hm_U,\varphi,\sv^U,\ev_V$} the U-Decoder outputting an $\ev_U$ such that $\ev_U \Hm_U^{\intercal}=\sv^U$ and $|\tphi(\ev_U,\ev_V)|=w$. \label{algo:DU}}
		\begin{algorithmic}[1]
			\State $t\gets |\ev_V|$
			\State $k_{\neq 0}\Unif\Dc_U^{t}$\label{line:kneq}
			\State $k_0 \gets k_U' - k_{\neq 0}$
			\Comment{$k_U' \eqdef k_U -d$}
			\Repeat 
			\State\label{line:infSet1}$\cJ,\cI\gets\Call{FreeSetW}{\Hm_U,\ev_V,k_{\neq 0}}$
			
			 			\State \label{line:infSet2} $\xv_U \Unif\{\xv\in\F_3^{n/2}\mid \forall j \in \cJ, \mbox{ } \xv(j) \notin \{ -\frac{b_i}{a_i}\ev_V(i), -\frac{d_i}{c_i}\ev_V(i) \} \mbox{ and } \Sp(\xv) \subseteq \cI \}$
			\State \label{line:eUoutput} $\ev_U\gets\Call{PrangeStep}{\Hm_U,\sv^U,\Ic,\xv_U}$
			\Until $|\tphi(\ev_U,\ev_V)| = w$
			\State \Return $\ev_U$
	\end{algorithmic}
	\smallskip	
	\hrule
	{\bf function} \Call{FreeSetW}{$\Hm,\xv,k_{\neq 0}$} 
	\hrule
	\begin{algorithmic}[1]
		\Require $\Hm\in\Fq^{(n-k)\times n}, \xv \in \Fq^{n}$ and $k_{\neq 0} \in \llbracket 0,k \rrbracket$.
		\Ensure $\cJ$ and $\Ic$ an information set of $\vectspace{\Hm}^\perp$ such that $\left|\{i \in \cJ: x_i \neq 0\}\right| =k_{\neq 0}$ and $\cJ \subset \cI$ of size $k - d$.  
		\Repeat 
		\State $\cJ_1 \Unif \Sp(\xv)$ of size $k_{\neq 0}$
		\State $\cJ_2 \Unif  \llbracket 1,n \rrbracket \mbox{\textbackslash} \Sp(\xv)$ of size $k - d - k_{\neq 0}$. 
		\State $\cJ \leftarrow \cJ_1 \sqcup \cJ_2$ 
		\Until the rank of the columns of $\Hm$ indexed by $ \llbracket 1,n \rrbracket \mbox{\textbackslash} \cJ$ is $n-k$
		\Repeat 
		\State  $\cJ' \Unif  \llbracket 1,n \rrbracket \mbox{\textbackslash} \cJ$ of size $d$
	\State $\cI \leftarrow \cJ \sqcup \cJ'$ 
\Until $\cI$ is an information set of $\vectspace{\Hm}^\perp$
\State \Return $\cJ,\cI$
	\end{algorithmic}
\end{algorithm}

	These two algorithms both use the Prange decoder in the same way as we did with the procedure described in \S\ref{subsec:genUVcodes} to reach large weights, except that here we introduced some internal distributions $\cD_V$ and the $\cD_U^t$'s. These distributions are here to tweak the weight distributions of $\DV{\cdot}$ and $\DU{\cdot}$ in order to reduce 
the rejection rate. We have:

	\begin{restatable}{proposition}{propoq} \label{propo:q}
	 Let $n$ be an even integer, $w \leq n$, $i,t,k_U \leq n/2$ and $s \leq t$ be integers. Let $d$ be an integer, $k_V' \eqdef k_V - d$ and $k_U' \eqdef k_U - d$. Let $X_V$ (resp. $X_U^{t}$) be a random variable distributed according to $\cD_V$ (resp. $\cD_U^{t}$). We have,
		
		\begin{equation} 
		q_1(i) = \sum_{t=0}^{i} \frac{\binom{n/2-k_V'}{i-t}2^{i-t}}{3^{n/2-k_V'}} \mathbb{P}(X_V = t)
		\end{equation}
		\begin{equation}
				q_{2}(s,t) =  \left\{
		\begin{array}{ll}
		\mathop{\sum}\limits_{\substack{t + k_U' - n/2 \leq k_{\neq 0} \leq t \\ k_0 \eqdef k_U' - k_{\neq 0} }} \frac{\binom{t - k_{\neq 0}}{s}\binom{n/2 - t - k_0}{\frac{w+s}{2} - t - k_0}2^{\frac{3s}{2}}}{\mathop{\sum}\limits_{p} \binom{t - k_{\neq 0}}{p}\binom{n/2 - t - k_0}{\frac{w+p}{2} - t - k_0}2^{\frac{3p}{2}} } \mathbb{P}(X_U^{t} = k_{\neq 0}) &\;\mbox{if } w \equiv s \bmod 2. \\
		\quad\quad 0  &\quad\mbox{else} 
		\end{array}
		\right.
		\end{equation} 
	\end{restatable} 
	
	The information set $\cI$ is also chosen by first choosing randomly  a set $\cJ$ of size $k-d$ where $k$ is the size of the information set and $d$ will be chosen so that $3^d \approx 2^\lambda$ where $\lambda$ is the security parameter.
	 Then $d$ positions are added to $\cJ$ until finding an information set. 
	The reason for this is the following: by choosing $\cI$ as this, we ensure that $\cI$ contains $k-d$ almost completely random
	positions (the probability that $\cJ$ gets rejected will be of order $\frac{1}{2^\lambda}$). On these positions $\cJ$
	we choose the weight of $\ev_{\cJ}$ according to $\cD$ but $\ev_{\cI \mbox{\textbackslash} \cJ}$ as a random vector and we complete $\ev_{\cI}$ with the Prange algorithm. If instead we had chosen the information set $\cI$ by picking $k$ positions 
	at random, then $\cI$  is rejected with some constant probability. Even if the Prange decoder based on this way of choosing the information set is very likely to be very close to meet the two uniformity conditions of Definition \ref{def:weightU}, this
	constant rejection probability makes a proof that the Prange decoder is close enough to behave uniformly 
	very difficult. This is actually what we need to ensure
	 that the output of these Prange decoders is close after rejection sampling to output $\ev_V$ and $\ev_U$ that 
	are close to be distributed like $\evu_V$ and $\evu_U$. This is circumvented by choosing $\cI$ as we do here.
For this way of forming the information set we can namely prove (see Appendix \ref{app:weightUnif}). 
\begin{theorem}\label{theo:trueRej}
Let $\ev$ be the output of Algorithm \ref{algo:skeleton} based on Algorithms \ref{algo:DV},\ref{algo:DU} and $\evu$ be a uniformly distributed error of weight $w$.  There exists a constant $\alpha >0$ depending on $k_U/n$ and $k_V/n$ such that any integer
		integer $d$ in the range $\IInt{0}{\alpha n}$ we have, 
		\begin{equation*} 
		\mathbb{P}\left( \rho(\ev,\evu) > \frac{1}{3^{d}} \right) \in \textup{negl}(n)
		\end{equation*} 
		where the probability is taken over the choice of matrices $\Hm_V$ and $\Hm_U$. 
	\end{theorem}
A sketch of the proof appears in the appendix in Section \ref{app:weightUnif}.

\subsection{Instantiating the Distributions}
Any choice for the distributions $\cD_V$ and $\cD_U^t$ in Algorithms
\ref{algo:DV} and \ref{algo:DU} will enable uniform sampling by a proper
choice of the rejection vectors $\rv_V$ and $\rv_U$ in Algorithm
\ref{algo:skeleton}. We argue here, through a case study, that an
appropriate choice of the distributions may considerably reduce the
rejection rate. In fact, what matters is to have the smallest possible
values for $\Mrs_V$ and $\Mrs_U(t)$ in
Proposition~\ref{propo:rejDistrib}.

The first step to achieve this is to correctly align the distributions
to their targets, we do that by a proper choice for the mean value or
of the mode ({\em i.e.\ } maximum value) of the distributions. Next we
choose a ``shape'' for the distributions. Here we will take
(truncated) Laplace distributions with a prescribed mean and choose a
variance which minimizes rejection.

For typical parameters with 128 bits of classical security, we will
give a case study with the above strategy, in which the total
rejection rate is about 8\%.

Let $k_V' \eqdef k_V-d$ and $k_U' \eqdef k_U - d$ be parameters of Algorithm \ref{algo:DV} and Algorithm \ref{algo:DU}. 

\subsubsection{Aligning the Distributions:}
\begin{enumerate}
\item For the distribution $\cD_V$. The output of
  Algorithm~\ref{algo:DV} has an average weight
  $\bar{\ell}+2/3(n/2-k_V')$, where $\bar{\ell}$ denotes the mean of
  $\cD_V$. It must be close to $\mathbb{E}(\wt{\evu_V})$. We
  will admit
  $
    \mathbb{E}(|\evu_V|) = \sum_{i=0}^{n/2} i \qu_V(i) = \frac{n}{2}\left( 1 - \left( 1 - \frac{w}{n}\right)^{2} - \frac{1}{2}\left(\frac{w}{n}\right)^{2} \right).
  $
  The mean value $\bar{\ell}$ of $\cD_V$ is chosen (close to) $(1 -
  \alpha) k_V'$ where $\alpha\in[0,1]$ is defined as follows
  \begin{equation}\label{eq:alpha}
    (1 - \alpha) k_V' = \frac{n}{2}\left( 1 - \left( 1 -
        \frac{w}{n}\right)^{2} -
      \frac{1}{2}\left(\frac{w}{n}\right)^{2} \right) - \frac{2}{3}\left(\frac{n}{2}-k_V'\right).
  \end{equation}
\item For the distribution $\cD_U^t$, $0\le t\le n/2$. Here, for every
  $t$, we want to align the functions $s\mapsto q_2(s,t)$ and
  $s\mapsto \qu_2(s,t)$ (see Proposition~\ref{propo:rejDistrib}). We
  get a very good estimate of the $s$ which maximizes $\qu_2(s,t)$ by
  solving numerically the equation $\qu_2(s-1,t)=\qu_2(s+1,t)$, that
  is
  \begin{displaymath}
    {\frac {8\, \left( t-s \right)  \left( t-s+1 \right)  \left( n-w-s+1
 \right) }{ \left( s+1 \right) s \left( w+s+1-2\,t \right) }} = 1
  \end{displaymath}
  We will denote $\ltypu(t)$ the unique real positive root
  of the above polynomial equation.

  We use the notations of Algorithm~\ref{algo:DU}, with in addition
  $\ev=\tphi(\ev_U,\ev_V)$.  We now have to
  determine which value of $k_{\neq0}$ (line~\ref{line:kneq}) will be
  such that $q_2(s,t)$ also reaches its maximum for $s=\ltypu(t)$. For
  a given $t$, $q_2(s,t)$ is the probability to have
  $\lw(\ev)=s$. This number counts the pairs $(i,i+n/2)$ with
  $i\in\llbracket0,n/2\rrbracket$ such that exactly one of $\ev(i)$
  and $\ev(i+n/2)$ is non-zero. This may only happen when
  $i\in\supp(\ev_V)\setminus\cJ$, in which case $\ev(i)$ and
  $\ev(i+n/2)$ are two random distinct elements of $\F_3$ and this
  particular $i$ is counted in $\lw(\ev)$ with probability
  $2/3$. Since $\wt{\supp(\ev_V)\setminus\cJ}=t-k_{\neq0}$, we
  typically have $\lw(\ev)=\frac23(t-k_{\neq0})$ and the best
  alignment is reached when the most probable output of distribution
  $\cD_U^t$ is $k_{\neq0}=t-\frac32\ltypu(t)$.
\end{enumerate}

\subsubsection{Matching the ``Shapes'':} to avoid a high rejection
rate we need to choose distributions so that the tails of the emulated
$q_1$ and $q_2$ are not lower than their respective targets. A bad
choice in this respect could lead to values of $\Mrs_V$ and
$\Mrs_U(t)$ growing exponentially with the block size. We choose
truncated Laplace distributions to avoid this.
\begin{definition}[Truncated Discrete Laplace Distribution (TDLD)]
  Let $\mu,\sigma$ be positive real numbers, let $a$ and $b$ be two
  integers. We say that a random variable $X$ is distributed according
  to the Truncated Discrete Laplace Distribution \textup{(TDLD)} of
  parameters $\mu,\sigma,a,b$, which is denoted
  $X \sim \Lap{\mu}{\sigma}{a,b}$, if for all $i\in\llbracket a,b\rrbracket$,
  $$
    \mathbb{P}\left( X = i \right) = \frac{e^{-\frac{|i-\mu|}{\sigma}}}{N}
  $$
  where $N$ is a normalization factor.
\end{definition}
We choose
\begin{displaymath}
  \left\{
  \begin{array}{rcl}
    \cD_V &\sim& \Lap{\mu_V}{\sigma_V}{0,k_V'}\\
    \cD_U^t &\sim& \Lap{\mu_U(t)}{\sigma_U(t)}{t + k_U' - n/2,t}
  \end{array}
  \right.
  \mbox{ with }
  \left\{
    \begin{array}{lcl}
      \mu_V & = & (1-\alpha)k_V' \\
      \mu_U(t) & = & t-\frac32\ltypu(t)+\varepsilon
    \end{array}
  \right.
\end{displaymath}
and $\sigma_V$ and $\sigma_U(t)$ to minimize $\Mrs_V$ and
$\Mrs_U(t)$. We also observed heuristically that the alignment is
improved by choosing a small $\varepsilon>0$, typically $\varepsilon=2$.

\subsubsection{Case Study:} $n=9078$, $(k_U,k_V)=(3749,1998)$,
$w=8444$, $\alpha=0.5907$ and $d = 162$.  With $\sigma_V=17.6$, we obtain
$\Mrs_V\approx 1.0417$. With $\sigma_U=6.8$ and $\varepsilon=0.2$ for all $t$,
we obtain $\Mrs_U\approx 1.0380$ on average. The result could be
marginally better by selecting the best $\sigma_U(t)$ (and $\varepsilon$)
for each $t$.

\subsection{Choosing the parameters}\label{sec:ch_params}
Using the parameter $\alpha$ introduced in \eqref{eq:alpha} in the previous subsection as
$$
(1-\alpha)k_V' = \frac{n}{2}\left( 1 - \left( 1 -
    \frac{w}{n}\right)^{2} - \frac{1}{2}\left(\frac{w}{n}\right)^{2}
\right) - \frac{2}{3}\left(\frac{n}{2}-k'_V\right).
$$
we may define all the system parameters depending only on $\alpha$,
the code rate $k/n$, $d$ and the block size $n$
\begin{eqnarray}\label{eq:alpha-w}
  w & = &\left\lfloor n \left( 1-\alpha + \frac{1}{3} \sqrt{ (3\alpha - 1) \left( 3\alpha + 4 \frac{k'}{n}  - 1 \right)} \right) \right\rfloor \\
  \label{eq:alpha-V}
  k_V' &= &\left\lfloor \frac{n}{2} \frac{3}{3\alpha - 1} \left( \left( 1 - \frac{w}{n} \right)^{2} + \frac{1}{2}\left( \frac{w}{n} \right)^{2} - \frac{1}{3} \right) \right\rfloor \mbox{ ; }
  k_U' = \left\lfloor \frac{n}{2}\left( -2 + 3 \frac{w}{n} \right) \right\rfloor
\end{eqnarray} 
where $k' \eqdef k_U + k_V - 2d$.

 \section{Achieving Uniform Domain Sampling}                                                 
\label{sec:domSampl}

The following definition will be useful to understand the structure of normalized generalized $\UV$-codes.

\begin{restatable}{definition}{defVblocks}{\textbf{\textup{(number of $V$ blocks of type I).}}}
	\label{def:Vpositions} In a normalized generalized $\UV$-code of length $n$
	associated to $(\av,\bv,\cv,\dv)$,
	the number of $V$ blocks of type $I$, which we denote by $n_I$,  is defined by:
	\begin{displaymath}
	n_I \eqdef \left| \left\{ 1 \leq i \leq n/2 : b_id_i=0
	\right\}\right|.
	\end{displaymath}
\end{restatable}

\begin{remark} 
	\label{rem:nI}
	$n_I$ can be viewed as the number of positions in which a codeword
	of the form $(\bv\hsp\vv,\dv\hsp\vv)$ is necessarily equal to $0$:
	this comes from the fact that on a position where either
	$b_i=0$ or $d_i=0$, the other one is necessarily
	different from $0$ as $a_id_i - b_ic_i = 1$.  In other words we also have	
	\begin{displaymath}
	n_I =  \left| \left\{ 1 \leq i \leq n/2 : b_i=0\right\}\right| +
	\left| \left\{ 1 \leq i \leq n/2 : d_i=0 \right\}\right|.
	\end{displaymath}
\end{remark} 
We denote by $\Hpub$ the public parity-check matrix of a normalized generalized $\UV$-code as described in \S \ref{subsec:waveTrap}. It turns out that $\Hpub$  has enough randomness in it for making 
the syndromes associated to it indistinguishable in the strongest possible sense, i.e.
statistically, from random syndromes as the following proposition shows. In other words,
 our scheme achieves the Domain Sampling property of
Definition \ref{def:WPS}.
Note that the upper-bound
we give here depends on the number $n_I$ we have just introduced.

\begin{restatable}{proposition}{propoDist}
  \label{prop:statDist} 
  Let $\Dsw{\Hm}$ be the distribution of 
  $\ev\transpose{\Hm}$ when $\ev$ is drawn uniformly at random among $S_w$
  and let $\Uc$ be the uniform
  distribution over $\F_3^{n-k}$.  We have
  \begin{displaymath}
    \esp_{\Hpub} \left( \rho(\Dsw{\Hpub}, \Uc) \right) \leq \frac{1}{2} \sqrt{\varepsilon} \quad \mbox{with,}
  \end{displaymath}
  $$
  \varepsilon = \frac{3^{n-k}}{2^{w}\binom{n}{w}} + 3^{n/2-k_V}\sum_{j=0}^{n/2}\frac{ \qu_1(j)^{2}}{2^{j}\binom{n/2}{j}} + 3^{n/2-k_U}
  \sum_{j=0}^{n_I} \frac{\binom{n_I}{j}\binom{n-n_I}{w-j}^{2}}{\binom{n}{w}^{2}2^{j}} 
    $$
  where $\qu_1$ is given in Proposition \ref{propo:qu} in \S\ref{sec:rejSampl}.
  
\end{restatable}

The proof of this proposition relies among other things on the following 
variation of the left-over hash lemma (see
\cite{BDKPPS11})
that is adapted to our case: 
here the hash function to which we apply the left-over hash lemma 
is defined as $h(\ev) = \ev \trHpub$. Functions $h$
do not form a universal family of hash functions (essentially
because the distribution of the $\Hpub$'s is not the uniform
distribution over $\F_3^{(n-k)\times n}$). However in our case we
	can still bound $\varepsilon$ by a direct computation. 
	
\begin{restatable}{lemma}{lemleftoverHash} \label{lem:leftOver}
	Consider a finite family $\Hc = (h_i)_{i \in I}$ of functions from a finite set $E$ to a finite set $F$.
	Denote by $\varepsilon$ the bias of the collision probability, i.e. the quantity such that
	\begin{displaymath}
	\prob_{h,e,e'}(h(e)=h(e')) = \frac{1}{|F|} (1 + \varepsilon)
	\end{displaymath}
	where $h$ is drawn uniformly at random in $\Hc$, $e$ and $e'$ are
	drawn uniformly at random in $E$. Let $\Uc$ be the uniform
	distribution over $F$ and $\Dc(h)$ be the distribution of the
	outputs $h(e)$ when $e$ is chosen uniformly at random in $E$.  We
	have
	\begin{displaymath}
	\esp_h \left( \rho(\Dc(h),\Uc) \right) \leq \frac{1}{2} \sqrt{\varepsilon}.
	\end{displaymath}
\end{restatable}

This lemma is proved in Appendix \S \ref{ss:leftoverhash}.
In order to use this lemma to bound the statistical distance we are interested in, we have proved in Appendix \S\ref{lem:syndromeDistribution} the following lemma:
\begin{restatable}{lemma}{lemSyndromeDistribution}
\label{lem:syndromeDistribution}
	Assume that $\xv$ and $\yv$ are random vectors of $S_w$ that are drawn uniformly at random in this set. We have
	$$
	\mathbb{P}_{\Hpub,\xv,\yv}
		\left( \xv \trHpub = \yv \trHpub \right)  \leq \frac{1}{3^{n-k}} (1 + \varepsilon) \mbox{ with } \varepsilon \mbox{ given in Proposition \ref{prop:statDist}.}  $$
\end{restatable}

 \section{Security Proof}
\label{sec:securityProof}

\subsection{Basic Tools}

\subsubsection{Basic Definitions.}

A {\em distinguisher} between two distributions $\mathcal{D}^{0}$ and
$\mathcal{D}^{1}$ over the same space $\mathcal{E}$ is a randomized
algorithm which takes as input an element of $\mathcal{E}$ that
follows the distribution $\mathcal{D}^{0}$ or $\mathcal{D}^{1}$ and
outputs $b \in \{0,1\}$. It is  characterized by its advantage:
$
  Adv^{\mathcal{D}^{0},\mathcal{D}^{1}}(\cA) \eqdef
  \mathbb{P}_{\xi \sim \mathcal{D}^{0}}\left( \cA(\xi) \mbox{
      outputs } 1 \right) - \mathbb{P}_{\xi \sim
    \mathcal{D}^{1}}\left(\cA(\xi) \mbox{ outputs } 1
  \right).
$

\begin{definition}
  [Computational Distance and Indistinguishability]
  The computational distance between two distributions $\mathcal{D}^{0}$ and $\mathcal{D}^{1}$ in time $t$ is: 
  \begin{displaymath}
    \rho_{c}\left( \mathcal{D}^{0},\mathcal{D}^{1}\right)(t) \eqdef
    \mathop{\max}\limits_{ |\cA| \leq t} \left\{
      Adv^{\mathcal{D}^{0},\mathcal{D}^{1}}(\cA) \right\}
  \end{displaymath}
  where $|\cA|$ denotes the running time of $\cA$ on
  its inputs.
  
\end{definition}
For signature schemes, one of the strongest security notion is {\em
	existential unforgeability under an adaptive chosen message attack}
(EUF-CMA). In this model the adversary has access to all signatures
of its choice and its goal is to produce a valid forgery. A valid
forgery is a message/signature pair $(\mv,\sigma)$ such that
$\Vrfypk(\mv,\sigma)=1$ whereas the signature of $\mv$ has never been
requested.

\begin{definition}
	[EUF-CMA Security] \label{def:EUF-CMA} 	A forger $\cA$
	is a $(t,\qhash,\qsig,\varepsilon)$-adversary in \textup{EUF-CMA} against
	a signature scheme $\cS$ if after at most $\qhash$ queries to the hash oracle, $\qsig$
	signatures queries and $t$ working time, it outputs a valid forgery
	with probability at least $\varepsilon$.
	The \textup{EUF-CMA} success probability against $\cS$ is:
	\begin{displaymath}
	Succ_{\cS }^{\textup{EUF-CMA}}(t,\qhash,\qsig) \eqdef
	\max \left( \varepsilon \mbox{} | \mbox{it exists a }
	(t,\qhash,\qsig,\varepsilon) \mbox{-adversary} \right).
	\end{displaymath}
	
\end{definition}

\subsection{Code-Based Problems}
\label{subsec:cbProb} 
We introduce the code-based problems that will be
used in the security reduction.

\begin{restatable}{problem}{doom}[\textup{DOOM} -- Decoding One Out of Many]
  \label{prob:DOOM}
  For $\Hm\in\F_3^{(n-k)\times n}$,
  $\sv_{1},\cdots,\sv_{N} \in \F_3^{n-k}$, integer $w$, find
  $\ev\in\F_3^{n}$ and $i \in \IInt{1}{N}$ such that 
  $\ev\transpose{\Hm}=\sv_i$ and $\wt{\ev}=w$.
\end{restatable}
We will come back to the best known algorithms
to solve this problem as a function of the distance $w$ in
\S\ref{subsec:messAtt}.
\begin{definition}[One-Wayness of DOOM]
  We define the success of an algorithm $\cA$ against \DOOM\ with the parameters $n,k,N,w$ as:
  \begin{align*} 
    Succ_{\DOOM}^{n,k,N,w}\left( \cA \right) = \prob \big( \cA&\left( \Hm,\sv_{1},\cdots,\sv_{N} \right) \mbox{solution  of } \DOOM \big)
  \end{align*} 
  where $\Hm \Unif \F_3^{(n-k)\times n}$, $\sv_i \Unif \F_3^{n-k}$ and
  the probability is taken over $\Hm$, the $\sv_i$'s and the internal coins of $\cA$. 
  The computational success in time $t$ of breaking \DOOM\ with the parameters $n,k,N,w$ is then defined as:
  	$$
    Succ_{\DOOM}^{n,k,N,w}(t) = \mathop{\max}\limits_{|\cA|\leq t} \left\{
      Succ_{\DOOM}^{n,k,N,w}\left( \cA \right) \right\}.
 	$$
\end{definition}

Another problem appears in the security proof: distinguish
random codes from a code drawn uniformly at random in the family used for public keys
in the signature scheme. In what follows $\Dpub$ denotes the distribution of public keys $\Hpub$ 
whereas 
$\Drand$ denotes the uniform distribution over $\F_3^{(n-k_U-k_V)\times n}$.

\subsection{EUF-CMA Security Proof} 
\label{sec:securityProof3}

\begin{restatable}{theorem}{secuRed}\textbf{\textup{(Security Reduction)}}.
  \label{theo:secRedu}
  Let $\qhash$ (resp. $\qsig$) be the number of queries to the hash
  (resp. signing) oracle. We assume that
  $\lambda_{0} = \lambda + 2\log_{2}(\qsig)$  where $\lambda$ is the security parameter of the signature scheme. We have in the random oracle model  for all time $t$, $t_{c} = t + O \left( \qhash  \cdot n^{2} \right)$ and $\varepsilon$ given in Proposition \ref{prop:statDist}:
  \begin{multline*}
    Succ_{\cS_{\textup{Wave}}}^{\textup{EUF-CMA}}(t,\qhash,\qsig) \leq
    2 Succ_{\DOOM}^{n,k,\qhash,w}(t_{c}) +\rho_{c} \left( \Drand,\Dpub \right)(t_{c}) \\ + \qsig \rho\left( \mathcal{D}_{w},\mathcal{U}_{w} \right) +  \frac{1}{2}\qhash\sqrt{ \varepsilon } + \frac{1}{2^{\lambda}}
  \end{multline*}
where $\mathcal{D}_{w}$ is the output distribution of Algorithm \ref{algo:skeleton} using Algorithms \ref{algo:DV} and \ref{algo:DU} and $\cU_w$ is the uniform distribution over $S_w$.
  \end{restatable}

 \section{Security Assumptions and Parameter Selection}

Our scheme is secure under two security assumptions. One relates to the
hardness decoding and the other to the indistinguishability of
generalized $\UV$-codes.

\subsection{Message Attack -- Hardness of Decoding}\label{subsec:messAtt}
Here we are interested in the hardness of the DOOM problem as stated in Problem \ref{prob:DOOM}
for the case $q=3$ when the target weight $w$
is large. This variant of the problem, including the multiple target
(DOOM) aspect, was recently investigated in \cite{BCDL19}. This work
adapted to this setting  the best generic decoding techniques
\cite{D91,S88,MMT11,BJMM12} which use the so-called PGE+SS framework
(``Partial Gaussian Elimination and Subset Sum''). It also uses Wagner's
generalized birthday algorithm \cite{W02} and the representation
technique \cite{HJ10}.

\subsection{Key Attack -- Indistinguishability of generalized $\UV$-Codes}\label{subsec:keyAtt}

Here we are interested in the hardness of the problem to distinguish
random codes from permuted generalized normalized $(U,U+V)$-code. All the proofs of this subsection are in Appendix \ref{sec:keyAtt}.

A normalized generalized $(U,U+V)$-code where $U$ and $V$ are random
seems very close to a random linear code.  
There is for instance only a
very slight difference between the weight distribution of a random
linear code and the weight distribution of a random normalized
generalized $\UV$-code of the same length and dimension. This
slight difference happens for small and large weights and is due to
codewords where $\vv = \mathbf{0}$ or $\uv = \mathbf{0}$ which are of
the form $(\av \hsp \uv, \cv \hsp \uv)$ where $\uv$ belongs to $U$ or
codewords of the form $(\bv \hsp \vv,\dv \hsp \vv)$ where $\vv$ belongs to
$V$ as shown by the following proposition:

\begin{restatable}{proposition}{propdensity}
	\label{prop:density}
	Assume that we choose a normalized generalized $\UV$-code
	over $\F_3$ with a number $n_I$ of linear combinations of type
	I by picking the parity-check matrices of $U$ and $V$
	uniformly at random among the ternary matrices of size
	$(n/2-k_U) \times n/2$ and $(n/2-k_V) \times n/2$
	respectively.  Let $a_{(\uv,\vv)}(z)$,
	$a_{(\uv,\mathbf{0})}(z)$ and $a_{(\mathbf{0},\vv)}(z)$ be the
	expected number of codewords of weight $z$ that are
	respectively in the normalized generalized $(U,U+V)$-code, of
	the form $(\av\hsp\uv,\cv\hsp\uv)$ where $\uv$ belongs to $U$ and
	of the form $(\bv \hsp \vv,\dv\hsp \vv)$ where $\vv$ belongs to $V$.
	These numbers are given for even $z$ in $\llbracket 0,n\rrbracket$ by
	\begin{displaymath}
	a_{(\uv,\mathbf{0})}(z) = \frac{\binom{n/2}{z/2}2^{z/2}}{3^{n/2 - k_U}} \quad ; \quad a_{(\mathbf{0},\vv)}(z) = \frac{1}{3^{n/2 - k_V}}\mathop{\sum}\limits_{\substack{j=0 \\
			j \text{ even}}}^{z} \binom{n_I}{j}\binom{n/2 -n_I}{\frac{z-j}{2}}2^{(z+j)/2}
	\end{displaymath}
	$$
	a_{(\uv,\vv)}(z) = a_{(\uv,\mathbf{0})}(z) + a_{(\mathbf{0},\vv)}(z) + \frac{1}{3^{n - k_U - k_V}} \left( \binom{n}{z}2^{z} - \binom{n/2}{z/2}2^{z/2} -\mathop{\sum}\limits_{\substack{j=0 \\
			j \text{ even}}}^{z} \binom{n_I}{j}\binom{n/2 - n_I}{\frac{z-j}{2}}2^{(z+j)/2} \right) 
	$$
	and for odd $z \in \llbracket 0,n\rrbracket$ by
	\begin{displaymath}
	a_{(\uv,\mathbf{0})}(z) = 0 \quad ; \quad a_{(\mathbf{0},\vv)}(z) = \frac{1}{3^{n/2 - k_V}}\mathop{\sum}\limits_{\substack{j=0 \\
			j \text{ odd} }}^{z} \binom{n_I}{j}\binom{n/2 - n_I}{\frac{z-j}{2}}2^{(z+j)/2}
	\end{displaymath}
	$$ 
	a_{(\uv,\vv)}(z) =  a_{(\mathbf{0},\vv)}(z) + \frac{1}{3^{n - k_U - k_V}} \left( \binom{n}{z}2^{z} - \mathop{\sum}\limits_{\substack{j=0 \\
			j \text{ odd}}}^{z} \binom{n_I}{j}\binom{n/2 - n_I}{\frac{z-j}{2}}2^{(z+j)/2} \right) 
	$$
	On the other hand, when we choose a linear code of length $n$ over $\F_3$  with a random parity-check matrix of size $(n-k_U-k_V)\times n$ 
	chosen uniformly at random, then the expected number $a(z)$ of codewords of weight $z>0$ is given by
	\begin{displaymath}
	a(z) = \frac{\binom{n}{z}2^{z}}{3^{n-k_U-k_V}}.
	\end{displaymath}
\end{restatable}

We have plotted in Figure \ref{fig:density} the normalized logarithm of the density of codewords of the form $(\av\hsp\uv,\cv\hsp\uv)$ and $(\bv\hsp\vv,\dv\hsp\vv)$ of relative 
{\em even} weight $x \eqdef \frac{z}{n}$ against $x$ in the case where $U$ is of rate $\frac{k_U}{n/2}=0.7$,
$V$ is of rate $\frac{k_V}{n/2}=0.3$ and $\frac{n_I}{n/2} = \frac{1}{2}$. These two relative densities are defined respectively by
\begin{displaymath}
\alpha_{\uv}(z/n) \eqdef \frac{\log_2(a_{(\uv,\mathbf{0})}(z)/a_{(\uv,\vv)}(z))}{n} \quad ; \quad 
\alpha_{\vv}(z/n) \eqdef  \frac{\log_2(a_{(\mathbf{0},\vv)}(z)/a_{(\uv,\vv)}(z))}{n}
\end{displaymath}
We see that for a relative weight $z/n$ below approximately $0.26$ almost all the codewords are of the form $(\av\hsp\uv,\cv\hsp\uv)$.

\begin{figure}
	\centering
	\includegraphics[scale = 0.2,height=6cm]{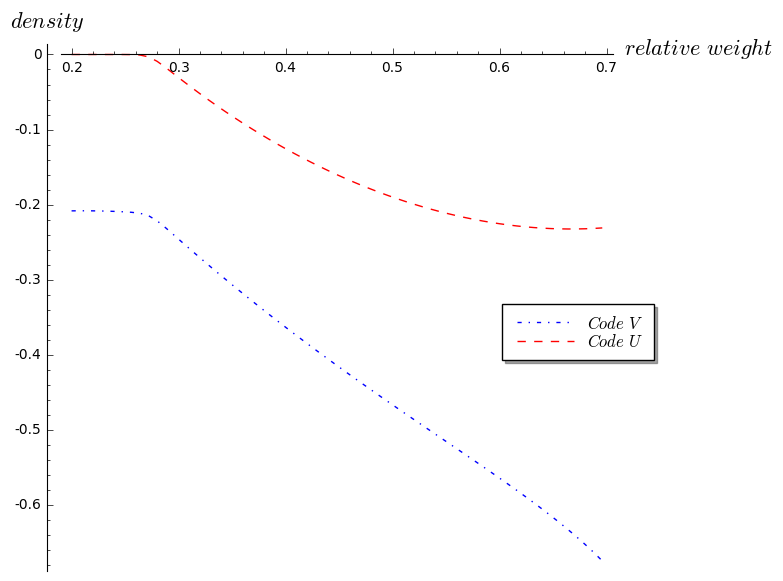}
	\caption{$\alpha_{\uv}(z/n)$ and $\alpha_{\vv}(z/n)$ against $x \eqdef \frac{z}{n}$.\label{fig:density}}
\end{figure}

Since the weight distribution is invariant by permuting the positions, this slight difference also survives
in the permuted version of the normalized generalized $(U,U+V)$-code. These considerations lead to the
best attack we have found for recovering the structure of a permuted normalized generalized $(U,U+V)$-code.
It consists in applying known algorithms aiming at recovering low weight codewords in a linear code.
We run such an algorithm until getting  at some point either a permuted $(\av\hsp\uv,\cv\hsp\uv)$ codeword where $\uv$ is in $U$ or
a permuted $(\bv\hsp\vv,\dv\hsp\vv)$ codeword where $\vv$ belongs to $V$.  
The rationale behind this algorithm is that the
density of codewords of the form $(\av\hsp\uv,\cv\hsp\uv)$ or $(\bv\hsp\vv,\dv\hsp\vv)$ is bigger when the weight of the codeword gets smaller.

Once we have such a codeword we can bootstrap from there
very similarly to what has been done in \cite[Subs. 4.4]{OT11}.
Note that this attack is actually very close in spirit to the attack that was devised on the KKS signature scheme \cite{OT11}.
In essence, the attack against the KKS scheme really amounts to recover the support of the $V$ code. 
The difference with the KKS scheme is that the support of $V$ is much bigger in our case. As explained in the conclusion of \cite{OT11} the attack against the KKS scheme has in essence
an exponential complexity. This exponent becomes really prohibitive in our case when the parameters of $U$ and $V$
are chosen appropriately as we will now explain.
Let us first introduce the following notation that will be useful in the following.
\newline

{\bf Punctured Code.} For a
subset $\cI \subset \llbracket 1,n\rrbracket$ and a code $\cC$ of length $n$, we
denote by $\punc_{\cI}(\cC)$, the code $\cC$ punctured in $\cI$, namely
$\{\cv_{\bar{\cI}}=(c_j)_{j \in \llbracket 1,n\rrbracket \setminus \cI}:\cv \in
\cC\}$.
In other words, the set of vectors obtained by deleting in the
codewords of $\cC$ the positions that belong to $\cI $.

\subsubsection{Recovering the $U$ Code up to Permutation.}	

We consider here the permuted code
\begin{displaymath}
U' \eqdef (\av\hsp U,\cv\hsp U)\Pm = \{(\av\hsp \uv,\cv\hsp\uv)\Pm: \uv \in U\}.
\end{displaymath}
The attack in this case consists in recovering a basis of $U'$. Once this is done, it is easy to recover the $U$ code up to permutation by matching the pairs of coordinates which are either always equal or always sum to $0$ in $U'$. The basic algorithm for recovering the code $U'$ is given in Algorithm \ref{algo:ComputeU}.

\begin{algorithm}[htbp]
	\textbf{Parameters: }  (i) $\ell$ : small integer (typically $\ell \leqslant 40$),\\
	(ii) $p$ : very small integer (typically $1 \leqslant p
	\leqslant 10$).\\
	{\bf Input:} (i) $\Cpub$ the public code used for verifying signatures.\\
	(ii)       $N$ a certain number of iterations\\
	{\bf Output:} an independent set of elements in $U'$
	\begin{algorithmic}[1]
		\Function{ComputeU}{$\Cpub$,$N$}
		\For{$i=1,\dots,N$}
		\State $B \leftarrow \emptyset$
		\State Choose a set $\cI\subset \llbracket 1,n\rrbracket$  of size $n-k-\ell$ uniformly at random
		\State $\Lc \leftarrow$ \Call{Codewords}{$\punc_{\cI}(\Cpub),p$} \label{l:codewords}
		\ForAll{$\xv \in \Lc$}
		\State $\xv \leftarrow$ \Call{Complete}{$\xv,\cI,\Cpub$}
		\If{\Call{CheckU}{$\xv$}}
		\State add $\xv$ to $B$ if $\xv \notin <B>$ 
		\EndIf
		\EndFor
		\EndFor
		\State \Return $B$
		\EndFunction
	\end{algorithmic}
	\caption{\textsc{ComputeU}: algorithm that computes a set of independent elements in $U'$.} \label{algo:ComputeU}
\end{algorithm}

It uses other auxiliary functions
\begin{itemize}
	\item \textsc{Codewords}$(\punc_{\cI}(\Cpub),p)$ which computes all (or a big fraction of) codewords of weight $p$ of the punctured public code 
	$\punc_{\cI}(\Cpub)$. All modern \cite{D91,FS09,MMT11,BJMM12,MO15} algorithms for decoding linear codes perform such 
	a task in their inner loop.
	\item \textsc{Complete}$(\xv,\cI,\Cpub)$ which computes the codeword $\cv$ in $\Cpub$ such that its restriction outside $\cI$ is equal to $\xv$.
	\item \textsc{CheckU}$(\xv)$ which checks whether $\xv$ belongs to $U'$. 
\end{itemize}

\subsubsection{Choosing $N$ Appropriately.} Let us first analyse how we have to choose $N$ such that
\textsc{ComputeU} returns $\Omega(1)$ elements. This is essentially
the analysis which can be found in \cite[\S 5.2]{OT11}.  

\begin{restatable}{proposition}{proporecovU}\label{propo:recovU}
	The probability $\Psucc$ that one iteration of the for loop (Instruction 2) in \textsc{ComputeU} 
	adds elements to the list $B$ is lower-bounded by 
	\begin{equation}
	\Psucc \geq \sum_{z=0}^{n/2}  \frac{\binom{n/2}{z}\binom{n/2-z}{k+\ell-2z}2^{k+\ell-2z}}{\binom{n}{k+\ell}} f\left(\frac{\binom{k+\ell-2z}{p-2i} \binom{z}{i}2^{p-i} }{3^{\max(0,k+\ell-z-k_U)}}\right)
	\end{equation}
	where $f$ is the function 
	defined by 
	$f(x) \eqdef \max \left(x(1-x/2),1-\frac{1}{x} \right)$.
	Algorithm \ref{algo:ComputeU} returns a non zero list with probability $\Omega(1)$ when $N$ is chosen as 
	$N = \Omega\left( \frac{1}{\Psucc}\right)$.
\end{restatable}

\subsubsection{Complexity of Recovering a Permuted Version of $U$.}

The complexity of a call to \textsc{ComputeU} can be estimated as follows. We denote the complexity of 
computing the list of codewords of weight $p$ in a 
code of length $k+\ell$ and dimension $k$ by  $C_1(p,k,\ell)$.  It depends on the particular algorithm used here.
For more details see \cite{D91,FS09,MMT11,BJMM12,MO15}. This is  the complexity of the call \textsc{Codewords}$(\punc_{\cI}(\Cpub),p)$ in Step 
\ref{l:codewords} in Algorithm \ref{algo:ComputeU}. The complexity of  \textsc{ComputeU} and hence the complexity of recovering a permuted version of 
$U$ is clearly lower bounded by
$\Omega\left( \frac{C_1(p,k,\ell)}{\Psucc} \right)$. 
It turns out that the whole complexity of recovering 
a permuted version of $U$ is actually of this order, namely $ \Theta\left( \frac{C_1(p,k,\ell)}{\Psucc} \right)$. This can be done by a combination of two techniques
\begin{itemize}
	\item Once a non-zero element of $U'$ has been identified, it is much easier to find other ones. This uses one of the tricks for breaking the KKS scheme
	(see \cite[Subs. 4.4]{OT11}). The point is the following: if we start again the procedure \textsc{ComputeU}, but this time by choosing a set $\cI$
	on which we puncture the code which contains the support of the codeword that we already found, then the number $N$ of iterations that we have to perform until finding a new element is negligible
	when compared to the original value of $N$. 
	\item The call to \textsc{CheckU} can be implemented in such a way that the additional complexity coming from all the calls to this function is of the same order as the $N$ calls 
	to \textsc{Codewords}. The strategy to adopt depends on the values of the dimensions $k$ and $k_U$. In certain cases, it is easy to detect such codewords since they have 
	a typical weight that is significantly smaller than the other codewords. In more complicated cases, we might have to combine a technique checking first the weight of $\xv$, if it is
	above some prescribed threshold, we decide that it is not in $U'$, if it  is below the threshold, we decide that it is a suspicious candidate and use then the previous trick.
	We namely check whether  
	the support of the codeword $\xv$ can be used to find other suspicious candidates much more quickly than performing $N$ calls to \textsc{CheckU}.
\end{itemize}
To keep the length of this paper within some reasonable limit we avoid here giving the analysis of those steps and we will just use
the aforementioned lower bound on the complexity of recovering a permuted version of $U$.

\subsubsection{Recovering the $V$ Code up to a Permutation}
\label{ss:V}

We consider here the permuted code
\begin{displaymath}
V' \eqdef (\bv\hsp V,\dv\hsp V)\Pm = \{ (\bv\hsp \vv,\dv\hsp \vv)\Pm \mbox{ where }\vv \in V \}.
\end{displaymath}
The attack in this case consists in recovering a basis of $V'$. Once this is achieved, the support $\Sp(V')$ of $V'$ can easily be obtained. Recall that this is the set of positions for which there exists at least one codeword
of $V'$ that is non-zero in this position. This allows to easily recover the code $V$ up to some permutation. The algorithm for recovering
$V'$ is the same as the algorithm for recovering $U'$.  We call the
associated function \textsc{ComputeV} though since they differ in the
choice for $N$. The analysis is slightly different indeed.

\subsubsection{Choosing $N$ Appropriately.} As in the previous subsection let us analyse  how we have to choose $N$ in order that \textsc{ComputeV} returns 
$\Omega(1)$ elements of $V'$. 
We have in this case the following result.

\begin{restatable}{proposition}{proporecovV}\label{propo:recovV}
	The probability $\Psucc$ that one iteration of the for loop (Instruction 2) in \textsc{ComputeV} adds elements to the list $B$ is lower-bounded by 
	\begin{multline*}
	\Psucc \geq \sum_{z=0}^{\min(n-k-\ell,n - n_I)}\sum_{m = 0}^{n/2-n_I}\frac{\binom{\frac{n}{2} - n_I}{m}\binom{n_I}{n-k-\ell-z}}{\binom{n}{n-k-\ell}}\max_{i=0}^{\lfloor p/2 \rfloor}f\left(\frac{\binom{n - n_I - z - 2m}{p-2i}\binom{m}{i}2^{p-i} }{3^{\max(0,n - n_I - z - m - k_V)}}\right) \\  \sum_{j = 0}^{n/2-n_{I} - m} \binom{n/2 - n_I-m}{j}2^{j}\binom{n_I}{z-n + 2n_I + 2m + j}
	\end{multline*}
	where $f$ is the function 
	defined by 
	$f(x) \eqdef \max \left(x(1-x/2),1-\frac{1}{x} \right)$.
	\textsc{ComputeV} returns a non-zero list with probability $\Omega(1)$ when $N$ is chosen as 
	$N = \Omega\left( \frac{1}{\Psucc}\right)$.
\end{restatable}

\subsubsection{Complexity of Recovering a Permuted Version of $V$.} As for recovering the permuted $U$ code, the complexity for recovering the permuted $V$ is of order 
$\Omega\left( \frac{C_1(p,k,\ell)}{\Psucc} \right)$.

\subsubsection{Distinguishing a Generalized $(U,U+V)$-Code}

It is not clear in the second case that from the single knowledge of $V'$ and a permuted version of $V$ we are able to find a permutation of the positions
which gives to the whole code the structure of a generalized $(U,U+V)$-code. However in both cases as single successful call to 
\textsc{ComputeV} (resp. \textsc{ComputeU})  is really distinguishing the code from a random code
of the same length and dimension. In other words, we have a distinguishing attack whose complexity is given by 
the following proposition

\begin{restatable}{proposition}{prcomplexityUV}
  \label{pr:complexity_U_V}
   Algorithm \ref{algo:ComputeU} lead to a distinguishing attack whose complexity is given by 
  $$\min\left(O\left(\min_{p,\ell}C_U(p,\ell)\right),O\left(\min_{p,\ell}C_V(p,\ell)\right)\right)$$
  \begin{equation}
    C_U(p,\ell) \eqdef \frac{C_1(p,k,\ell)}{\mathop{\sum}\limits_{z=0}^{n/2}  \frac{\binom{n/2}{z}\binom{n/2-z}{k+\ell-2z}2^{k+\ell-2z}}{\binom{n}{k+\ell}} \mathop{\max}\limits_{i=0}^{\lfloor p/2 \rfloor}f\left(\frac{\binom{k+\ell-2z}{p-2i} \binom{z}{i}2^{p-i} }{3^{\max(0,k+\ell-z-k_U)}}\right)}\label{eq:secU}
  \end{equation}\vspace{-5mm}
  \begin{multline}\label{eq:secV}
    C_V(p,\ell) \eqdef\\ \frac{C_1(p,k,\ell)}{\mathop{\sum}_\Ic\frac{\binom{\frac{n}{2} - n_I}{m}\binom{n_I}{n-k-\ell-z}}{\binom{n}{n-k-\ell}}\mathop{\max}\limits_{i=0}^{\lfloor p/2 \rfloor}f\left(\frac{\binom{n - n_I - z - 2m}{p-2i}\binom{m}{i}2^{p-i} }{3^{\max(0,n - n_I - z - m - k_V)}}\right)\binom{n/2 - n_I-m}{j}2^{j}\binom{n_I}{z-n + 2n_I + 2m + j}.}
  \end{multline}
  where $C_1(p,k,\ell)$ is the the complexity of a computing a constant
  fraction (say half of them) of the codewords of weight $p$ in a code
  of length $k+\ell$ and dimension $k$ and $f$ is the function
  $f(x) \eqdef \max \left(x(1-x/2),1-\frac{1}{x} \right)$. The sum in
  the denominator of \eqref{eq:secV} is over the domain
  $\Ic=\{(z,m,j)\mid 0\le z\le\min(n-k-\ell,n-n_I), 0\le m
  \le n/2-n_I,0\le j\le n/2-n_{I} - m\}$.
\end{restatable}

We explain in Appendices \S\ref{app:CU} and \S\ref{app:CV} how to estimate $C_U$ and $C_V$. 

\subsection{Parameter Selection}\label{ss:parameter}
With proper rejection sampling, the security of Wave provably reduces
to the two previous hard computational problems. The best known solvers, presented
above, both have an exponential complexity.  For a given set of system
parameters $(n,w,k_U,k_V,k=k_U+k_V)$, their asymptotic complexities
can be expressed as
\begin{itemize}
\item for the message attack, $2^{c_Mn(1+o(1))}$ where $c_M$ is a function of
  $w/n$ and $k/n$
\item for the key attack, $2^{c_Kn(1+o(1))}$ where $c_K$ is a function of
  $k_U/n$ and $k_V/n$
\end{itemize}
Using the relations of \S\ref{sec:ch_params}, both $c_M$ and $c_K$ can
be expressed as functions of the code rate $R=k/n$ and of the
parameter $\alpha$. Minimizing the public key size under the
constraint $c_M(R,\alpha)=c_K(R,\alpha)$, we obtain
\begin{displaymath}
  R = 0.633, \alpha=0.590656, c_M\approx c_K\approx 0.0141.
\end{displaymath}
For $\lambda$ bits of (classical) security we get ($K$ the key size
in bits):
\begin{displaymath}
  n = \frac{\lambda}{0.0141}, ~~ w = 0.9302\, n, ~~ k_U = 0.8259\,
  \frac{n}{2}, ~~ k_V = 0.4402\, \frac{n}{2}, ~~ K = 0.368\, n^2
\end{displaymath}
To reach 128 bits of security we obtain $n=9078$, $w=8444$, $k_U=3749$,
$k_V=1998$ for a public key size of $3.8$ megabytes. We also checked that 
the other terms in the security reduction do not interfere here. For instance, we recommend
to choose the vectors $\av, \bv, \cv, \dv$ uniformly at random among the choices
that give a $\tphi$ that is $UV$-normalized, meaning that for all $i$ in $\IInt{1}{n/2}$ we should have
 $a_id_i-b_ic_i=1$ and
$a_i c_i \neq 0$. We reject choices that lead to a number $n_I$ of V blocks of type I
that are not close to their  expected value $\esp(n_I)=n/6$. By doing so we can control 
the parameter $\varepsilon$ giving an upper-bound on $\esp_{\Hpub} \left( \rho(\Dsw{\Hpub}, \Uc) \right)$. 
In the case $n_I=n/6$ this upper-bound is of order 
$\approx 2^{-254}$.

\subsection{Implementation}
The scheme was implemented in SageMath as a proof of concept.  For the
parameters $(n,w)=(9078,8444)$ each signature is produced in a few
seconds.  This gives a compelling argument to debunk the claim made in
\cite{BP18} to break Wave. The algorithm of \cite{BP18} collects a set
$\Sc$ of signatures, measures for each pair of indices $(i,j)$ the
quantity
$\wt{\{e_i=-e_j\mid\ev\in\Sc\}} - \wt{\{e_i=e_j\mid\ev\in\Sc\}}$ and
selects for each $i$ the pair $(i,j)$ which maximizes this quantity. A
tentative secret key is then derived from the selected pairs. The
first version of this paper \cite{BP18a} proposed an algorithm that
recovers the secret key when rejection sampling was left out from the
$\UV$-decoder. It uses information leakage from a few hundred
signatures to achieve its purpose.  The authors of \cite{BP18a} were
told that the rejection sampling step was critical to ensure uniformly
distributed signatures over $S_w$ and thus resistance against leakage
attack. Subsequent versions of \cite{BP18} claimed that their
algorithm also worked with the rejection sampling step. There was no
implementation of Wave at that time to give a practical refutation of
this conjecture. We could now test our implementation against the
algorithm given in \cite{BP18}.  With a set of $25\,000$ properly
generated signatures the algorithm failed as expected to recover the
secret key.

 \section{Concluding Remarks and Further Work}\label{sec:conclusion}

We have presented Wave the first code-based ``hash-and-sign'' signature
scheme which strictly follows the GPV strategy \cite{GPV08}. This
strategy provides a very high level of security, but because of the
multiple constraints it imposes, very few schemes managed to comply to
it. For instance, only one such scheme based on hard lattice problems
\cite{FHKLPPRSWZ} was proposed to the recent NIST standardization
effort.
Our scheme is secure under two assumptions from coding theory.  Both
of those assumptions relate closely to hard decoding problems. Using
rejection sampling, we have shown how to efficiently avoid key leakage
from any number of signatures.  The main purpose of our work was to
propose this new scheme and assess its security. Still, it has a few
issues and extensions that are of interest.

\smallskip {\noindent \em The Far Away Decoding Problem.} The message security of
\wave{} relates to the hardness of finding a codeword {\em far} from a
given word. A recent work \cite{BCDL19} adapts the best ISD techniques
for low weight \cite{MMT11,BJMM12} and goes even further with a higher
order generalized birthday algorithm \cite{W02}. 
Interestingly enough, in the non-binary case, this work gives a  worst case 
exponent for the far away codeword that is significantly  larger than the close codeword 
worst case exponent. This seems to point to the fact that the far away codeword problem may even be more difficult to solve than the 
close codeword problem. This raises the issue of obtaining code-based primitives with better parameters
that build upon the far away codeword rather than on the usual close codeword problem.

\smallskip {\noindent \em Distinguishability.} Deciding whether a
matrix is a parity check matrix of a generalized $\UV$-code is also a
new problem. As shown in \cite{DST17b} it is hard in the worst case
since the problem is NP-complete. In the binary case, $\UV$ codes have
a large hull dimension for some set of parameters which are precisely
those used in \cite{DST17b}. In the ternary case the normalized
generalized $\UV$-codes do not suffer from this flaw. The freedom of
the choice on vectors $\av,\bv,\cv$ and $\dv$ is very likely to make
the distinguishing problem much harder for generalized $\UV$-codes
than for plain $\UV$-codes. Coming up with non-metric based
distinguishers in the generalized case seems a tantalizing problem
here.

\smallskip{\noindent \em On the Tightness of the Security Reduction.}
It could be argued that one of the reasons of why we have a tight
security-reduction comes from the fact that we reduce to the multiple
instances version of the decoding problem, namely DOOM, instead of the
decoding problem itself. This is true to some extent, however this
problem is as natural as the decoding problem itself. It has already
been studied in some depth \cite{S11} and the decoding techniques for
linear codes have a natural extension to DOOM as noticed in
\cite{S11}. We also note that with our approach, where a message has
many possible signatures, we avoid the tightness impossibility results
given in \cite{BJLS16} for instance.

\smallskip
{\noindent \em Rejection Sampling.} Rejection sampling in our
algorithm is relatively unobtrusive: a rejection every few signatures
with a crude tuning of the decoder. We believe that it can be further
improved. Our decoding has two steps. Each step is parametrized by a
weight distribution which conditions the output weight distribution. We
believe that we can tune those distributions to
reduce the probability of rejection to an arbitrarily small
value. This task requires a better understanding of the distributions
involved. This could offer an interesting trade-off in which the
designer/signer would have to precompute and store a set of
distributions but in exchange would produce a signing algorithm that
emulates a uniform distribution without rejection sampling.

\newcommand{\etalchar}[1]{$^{#1}$}

\newpage
\appendix

\section{Some Useful Distributions}\label{app:usefulDistribs}

	The purpose of this section is to prove Propositions \ref{propo:qu} and \ref{propo:q} which give the distributions $\qu_1,\qu_{2},q_1$ and $q_2$. 
 
	\subsection{Proof of Proposition \ref{propo:qu}}
	Let us first recall the definitions of $\qu_1$ and $\qu_2$. We have 
	$$
	\qu_1(i) = \mathbb{P}(|\evu_V|=i) \quad ; \quad \qu_2(s,t) = \mathbb{P}(\lw(\evu) = s \mid |\ev_V| = t)
	$$
	where 
	\begin{itemize}
	\item $\evu$ is a random vector drawn uniformly at random among the vectors of weight $w$ in $\F_3^n$
	\item $\evu_V \eqdef - \cv \hsp \ev_1 + \av \hsp \ev_2$ with 
	$\ev_1$ and $\ev_2$ being vectors in $\F_3^{n/2}$ such that $\evu=(\ev_1,\ev_2)$ and
	$\av,\bv,\cv$ and $\dv$ are vectors of $\F_3^{n/2}$ verifying the following equations
	\begin{equation} \label{eq:qu} 
	\forall i \in \llbracket 1,n/2 \rrbracket, \quad  a_id_i - b_ic_i = 1 \quad ; \quad a_ic_i \neq 0
	\end{equation} 
	\item $\lw(\xv) \eqdef |\{ 1 \leq i \leq n/2 : |(x_i,x_{i+n/2})|=1 \}|$.
	\end{itemize}
	 Let us prove now Proposition \ref{propo:qu}:
 	\propoqu*

	\begin{proof} Let us first compute the distribution $\qu_1$. The following lemma will be useful:
		
		\begin{lemma}\label{lemm:qu} 
		$|\ev_2-\ev_1| \sim |\evu_V|$.
		\end{lemma}
		
		\begin{proof}[Proof of Lemma \ref{lemm:qu}]
		Let $\ev'_1 \eqdef \cv \hsp \ev_1$, $\ev'_2 \eqdef \av \hsp \ev_2$, $\ev' \eqdef (\ev'_1,\ev'_2)$.
		$\ev'$ is clearly a random vector that is uniformly distributed over the words of weight $w$ in $\F_3^n$ because all the entries of $\av$ and $\cv$ are non-zero. Since $\evu_V = -\cv \hsp \ev_1 + \av \hsp \ev_2 = \ev'_2-\ev'_1$ we deduce that 
		$|\ev_2-\ev_1|$ and $ |\evu_V|=|\ev'_2-\ev'_1|$ have the same distribution. \qed
		\end{proof}

	From this lemma, to compute the distribution $q_1$ it is enough to determine for all $i$ in $\llbracket 1,n/2 \rrbracket$, 
	$\mathbb{P}(|\ev_2-\ev_1|=i)$
	where $(\ev_1,\ev_2)$ is uniformly distributed over the words of weight $w$. Let us define the following quantities:
	\begin{eqnarray} 
	p &\eqdef& \left| \{ 1 \leq i \leq n/2 : (\ev_1(i),\ev_2(i)) \in \{ (1,0),(0,1),(-1,0),(0,-1) \} \} \right| \label{eq:p} \\
	r &\eqdef& \left| \{ 1 \leq i \leq n/2 : (\ev_1(i),\ev_2(i)) \in \{ (1,-1),(-1,1) \} \} \right| \label{eq:r}\\
	l &\eqdef&  \left| \{ 1 \leq i \leq n/2 : (\ev_1(i),\ev_2(i)) \in \{ (1,1),(-1,-1)\} \} \label{eq:l} \right|
	\end{eqnarray}
	We have:
	$$
	w = |\ev| = 2l + 2r + p \quad ; \quad j = |\ev_1 - \ev_2| = p + r
	$$
	We have therefore that $p \equiv w \mod 2$,  
	$r=j-p$ and  $l = (w+p)/2-j$. By summing over all possibilities for $p$, 
	it follows that the number of errors $\ev=(\ev_1,\ev_2)$ of weight $w$ such that $|\ev_1 - \ev_2| = j$ is given by $$
	\mathop{\sum}\limits_{\substack{p=0 \\
			p \equiv w \bmod 2}}^{j}\binom{n/2}{j}\binom{j}{p}4^{p}2^{j-p}\binom{n/2-j}{\frac{w+p}{2}-j}2^{\frac{w+p}{2}-j}=
	\mathop{\sum}\limits_{\substack{p=0 \\
			p \equiv w \bmod 2}}^{j}\binom{n/2}{j}\binom{j}{p}\binom{n/2-j}{\frac{w+p}{2}-j}2^{\frac{w+3p}{2}}
	$$
	which concludes the computation of $\qu_1$. 	
	Let us now compute the distribution $\qu_2$.
	
	\begin{lemma}\label{lemma:npst} Let $n'(s,t)$ be the number of words $\evu=(\ev_1,\ev_2)$ of weight $w$ that verify 
		$|\ev_2-\ev_1| = t$ and $\lw(\evu) = s$. We have,
	\begin{eqnarray*}
	n'(s,t) 
		& = & \left\{
	\begin{array}{ll}  \binom{n/2}{t}2^{w/2}\binom{t}{s}2^{3s/2}\binom{n/2-t}{\frac{w+s}{2} -t} \;\;\text{ if $s \equiv w \bmod 2$}\\
	0 \;\;\text{else.}
\end{array}
\right. 
	\end{eqnarray*}
	\end{lemma} 

	\begin{proof} We use the quantities defined in Equations \eqref{eq:p},\eqref{eq:r} and \eqref{eq:l}. Note that $\lw(\evu)=p$.
	For words which define $n'(s,t)$ we have $p=s$, $r=t-p=t-s$ and $l=\frac{w+p}{2}-t=\frac{w+s}{2}-t$.
	Moreover the constraint $p \equiv w \mod 2$ translates into $s \equiv w \mod 2$. \qed 
	\end{proof}

	This concludes the proof by noticing that
	$$
	 \mathbb{P}(\lw(\evu) = s \mid |\ev_V| = t) =  \frac{n'(s,t)}{\sum_p n'(p,t)}.
	$$

	\end{proof}

	\subsection{Proof of Proposition \ref{propo:q}}

	Our aim here is to prove Proposition \ref{propo:q}. It gives the weight distribution of $\DV{\cdot}$ as $q_1$ and $\lw(\cdot)$-distribution of $\DU{\cdot}$ as $q_2$. 
	Let us recall that algorithms $\DV{\cdot}$ and $\DU{\cdot}$ are given in Subsection \ref{subsec:prangeDecUV}. 
	We are now ready to prove:
	
	\propoq*

	\begin{proof} The computation of $q_1$ easily follows from the fact that $|\ev_V|$ (the output of Prange Algorithm, Line \ref{line:eVoutput} in Algorithm \ref{algo:DV}) can be written (Proposition \ref{propo:Prange} in Subsection \ref{subsec:prangeStep}) as 
		$
		S + T
		$
	where $S$ and $T$ are independent random variables such that $S$ denotes the weight of a vector that is uniformly distributed over $\F_3^{n/2-k_V'}$ and $T$ is distributed according to $\cD_V$ (in the Prange algorithm used in $\DV{\cdot}$ we uniformly picked $d$ symbols in the information set).
	To compute $q_2$ let us count the number $n(s,t,k_{\neq 0})$ of different $\ev_U$ that can be output by $\DU{\cdot}$ for a given value of $\ev_V$ 
(which is supposed to be of weight $t$) 
 and $\cJ$ (included in an information set $\Ic$) that is assumed to intersect the support of $\ev_V$ in exactly $k_{\neq 0}$ positions and that are such that
$\lw(\ev)=s$. We can 
partition $\IInt{1}{n/2}$ as 
$$
\IInt{1}{n/2} = \cJ \cup \Ic_1 \cup \Ic_2
$$
where $\Ic_1$ is the set of positions that are not in $\cJ$ but in the support of $\ev_V$, whereas $\Ic_2$ is 
the set of positions that are neither in $\cJ$ nor in the support of $\ev_V$. By assumption on $\ev_V$ we know that 
$|\Ic_1|=t-k_{\neq 0}$. Furthermore $|\cJ|=k_U-d$ and $\Ic_2=n/2-|\cJ|-|\Ic_1|=n/2-k_U+d-(t-k_{\neq 0})=n/2-t-k_0$ where
$k_0 \eqdef k_U - d - k_{\neq 0}$.
For $i \in\{0,1,2\}$ we let
$$
\Jc_i \eqdef \{i \in \IInt{1}{n/2}: |(e_i,e_{i+n/2})|=i\} \quad ; \quad 
j_i \eqdef |\Jc_i|. 
$$
We necessarily have
$$
j_1 = s \quad ; \quad 
n-w = j_1+2j_0.
$$
We derive from these equalities that
$$
j_0 = \frac{n-w-s}{2}
$$
Now we also have
$$
\Jc_1 \subseteq  \Ic_1 \quad ; \quad 
\Jc_0 \subseteq  \Ic_2.
$$
We can choose the $j_1=s$ positions of $\Jc_1$ as we wish among the $t-k_{\neq 0}$ positions of $\Ic_1$. Similarly 
we may choose the $j_0=\frac{n-w-s}{2}$ positions of $\Jc_0$ as we wish among the $n/2-t-k_0$ positions of $\Ic_2$. Vector $\ev_U$ is necessarily fixed over all positions in $\cJ$ by choice of the Prange algorithm, it is also necessarily fixed in the 
positions $\Ic_1 \setminus \Jc_1$ and $\Jc_0$. For positions $i$ in $\Jc_1 \cup (\Ic_2 \setminus \Jc_0)$ there are two possibilities for the
value $\ev_U(i)$.
This implies that
\begin{eqnarray*}
n(s,t,k_{\neq 0}) &= &\binom{t-k_{\neq 0}}{s}\binom{n/2-t-k_0}{ \frac{n-w-s}{2}}2^s2^{n/2-t-k_0 - \frac{n-w-s}{2}}\\
& = & \binom{t-k_{\neq 0}}{s}\binom{n/2-t-k_0}{\frac{n-w-s}{2}}2^{\frac{3s}{2}+\frac{w}{2}-t-k_0}.
\end{eqnarray*}
We therefore have
																																																						\begin{eqnarray*}
\prob(\lw(\ev)=s\mid |\ev_V|=t, \cJ \cap \supp(\ev_V)=k_{\neq 0}) & = & \frac{n(s,t,k_{\neq 0})}{\sum_p n(s,t,p)}\\
& =& 	\frac{\binom{t-k_{\neq 0}}{s}\binom{n/2-t-k_0}{\frac{n-w-s}{2}}2^{\frac{3s}{2}+\frac{w}{2}-t-k_0}}{\sum_{p}\binom{t-k_{\neq 0}}{p}\binom{n/2-t-k_0}{\frac{n-w-p}{2}}2^{\frac{3p}{2}+\frac{w}{2}-t-k_0}}\\
& = & 
\frac{\binom{t-k_{\neq 0}}{s}\binom{n/2-t-k_0}{\frac{n-w-s}{2}}2^{\frac{3s}{2}}}{\sum_{p}\binom{t-k_{\neq 0}}{p}\binom{n/2-t-k_0}{\frac{n-w-p}{2}}2^{\frac{3p}{2}}}.
\end{eqnarray*}
This concludes the proof by summing over all possibilities for $k_{\neq 0}$. 	\end{proof}

\section{Sketch of the proof of Theorem \ref{theo:trueRej}} \label{app:weightUnif}

Let us introduce a definition that will be useful.

\begin{definition}[Bad and Good Subsets] Let $d \leq k \leq n$ be integers and $\Hm \in \F_3^{(n-k)\times n}$. A subset $\cE \subseteq \llbracket 1,n \rrbracket$ of size $k-d$ is defined as a good set for $\Hm$ if $\Hm_{\overline{\cE}}$ is of full rank where $\overline{\cE}$ denotes the complementary of $\cE$. Otherwise, $\cE$ is defined as a bad set for $\Hm$.
\end{definition}

We summarize in Figures \ref{fig:decodeV} and \ref{fig:decodeU} how $\DV{\cdot}$ and $\DU{\cdot}$ work where the $\cJ$'s are random good sets for $\Hm_V$ and $\Hm_U$.

\begin{figure}\begin{center} 
	\begin{tikzpicture}
	\draw (0,0) rectangle (10,0.5);
	\draw (0,0) rectangle (3,0.5);
	\node at (1.5,0.25) {$\xv_V$};
	\draw[pattern=crosshatch, pattern color=lightgray] (3,0) rectangle (4,0.5);
	\draw[pattern=north west lines, pattern color=lightgray] (4,0) rectangle (10,0.5);
	\draw [thick, ,decorate,decoration={brace,amplitude=10pt,mirror},xshift=0.4pt,yshift=-0.4pt](0,-0.3) -- (2.8,-0.3) node[black,midway,yshift=-0.6cm] {\footnotesize $\cJ$};
	\draw [thick, ,decorate,decoration={brace,amplitude=5pt,mirror},xshift=0.4pt,yshift=-0.4pt](3,-0.3) -- (3.9,-0.3) node[black,midway,yshift=-0.6cm] {\footnotesize $\cI \mbox{\textbackslash}\cJ$};
	\draw [thick, ,decorate,decoration={brace,amplitude=10pt,mirror},xshift=0.4pt,yshift=-0.4pt](4,-0.3) -- (10,-0.3) node[black,midway,yshift=-0.6cm] {\footnotesize $\llbracket 1,n/2\rrbracket \mbox{\textbackslash}\cI$};
	\draw[pattern=north west lines, pattern color=lightgray] (0.2,1.5) rectangle (1.2,2
	);
	\node at (5.5,1.75) {\small Uniformly distributed by property of the Prange algorithm};
	\draw[pattern=crosshatch, pattern color=lightgray] (0.2,0.75) rectangle (1.2,1.25
	);
	\node at (5.2,1) {\small Uniformly distributed by specification of the algorithm};
	\end{tikzpicture}
	\end{center} 
	\caption{Decoding of the code $V$ \label{fig:decodeV}}
\end{figure}

\begin{figure}
	\begin{center}
		\begin{tikzpicture}
		\node at (5.5,2.75) {\small Uniformly distributed by property of the Prange algorithm};
		\draw[pattern=north west lines, pattern color=lightgray] (0.2,2.5) rectangle (1.2,3);
		\node at (5.2,2) {\small Uniformly distributed by specification of the algorithm};
		\draw[pattern=crosshatch, pattern color=lightgray] (0.2,1.75) rectangle (1.2,2.25);
		\draw (0,0) rectangle (10,0.5);
		\draw (4,0) rectangle (10,0.5);
		\draw [thick, ,decorate,decoration={brace,amplitude=10pt},xshift=0.4pt,yshift=-0.4pt](4,0.8) -- (10,0.8) node[black,midway,yshift=0.6cm] {\footnotesize $\Sp(\ev_V)$};
		\draw (0,0) rectangle (1.5,0.5);
		\node at (0.75,0.25) {$(\xv_U)_{\cJ_2}$};
		\draw (4,0) rectangle (5.5,0.5);
		\draw[pattern=crosshatch, pattern color=lightgray] (5.5,0) rectangle (6.5,0.5);
		\draw [thick, ,decorate,decoration={brace,amplitude=5pt,mirror},xshift=0.4pt,yshift=-0.4pt](5.5,-0.3) -- (6.5,-0.3) node[black,midway,yshift=-0.6cm] {\footnotesize $\cI_1 \mbox{\textbackslash}\cJ_1$};
		\draw[pattern=crosshatch, pattern color=lightgray] (1.5,0) rectangle (2.3,0.5);
		\draw [thick, ,decorate,decoration={brace,amplitude=5pt,mirror},xshift=0.4pt,yshift=-0.4pt](1.5,-0.3) -- (2.3,-0.3) node[black,midway,yshift=-0.6cm] {\footnotesize $\cI_2 \mbox{\textbackslash}\cJ_2$};
		\draw[pattern=north west lines, pattern color=lightgray] (2.3,0) rectangle (4,0.5);
		\draw[pattern=north west lines, pattern color=lightgray] (6.5,0) rectangle (10,0.5);
		\node at (4.75,0.25) {$(\xv_U)_{\cJ_1}$};
		\draw[<->] (4,-1.3) -- (5.5,-1.3);
		\node at (4.75,-1.6) {$k_{\neq 0}$};
		\node at (0,-2.2) {$\cJ_1 \eqdef \cJ \cap \Sp(\ev_V)$};
		\node at (1.5,-2.3){;};
		\node at (3,-2.2) {$\cJ_2 \eqdef \cJ \cap \overline{\Sp(\ev_V)}$};
		\node at (4.5,-2.3){;};
		\node at (6,-2.2) {$\cI_1 \eqdef \cI \cap \Sp(\ev_V)$};
		\node at (7.5,-2.3){;};
		\node at (9,-2.2) {$\cI_2 \eqdef \cI \cap \overline{\Sp(\ev_V)}$};
		\end{tikzpicture}
	\end{center}
	\caption{Decoding of the code $U$ \label{fig:decodeU}}
\end{figure}
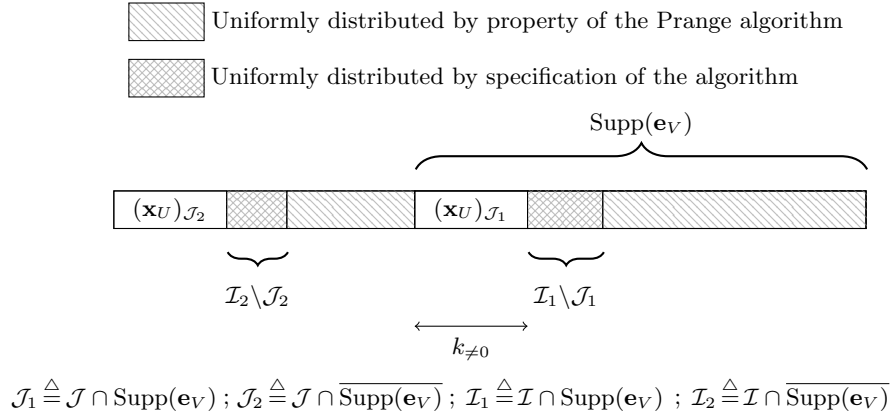

	We consider variations $\DVv{\cdot}$ and $\DUv{\cdot}$ of algorithms $\DV{\cdot}$ and $\DU{\cdot}$ respectively that work as $\DV{\cdot}$ and $\DU{\cdot}$
	when $\cJ$ is a good set and depart from it when $\cJ$ is a bad set. In the later case, the Prange decoder is not used anymore
	and an error is output that simulates what the Prange decoder would do with the exception that there is no guarantee that 
	the error $\ev_V$ that is output by $\DVv{\cdot}$ satisfies $\ev_V \Hm_V = \sv_V$ or that the 
	$\ev_U$  that is output by $\DUv{\cdot}$ satisfies $\ev_U \Hm_U = \sv_U$. The $\ev_V$ and $\ev_U$ that are output are chosen on the positions of $\cJ$ as $\DV$ and $\DU$ as would have done it, but the rest of the positions are chosen uniformly at random in $\F_3$.
	It is clear that in this case 
	\begin{fact}\label{fa:DVvDUv}
	$\DVv{\cdot}$ is weightwise uniform and $\DUv{\cdot}$ is $\lw$-uniform.
	  \end{fact} 
	The point of considering $\DVv{\cdot}$ and $\DUv{\cdot}$ is that they are very good approximations of 
	$\DV{\cdot}$ and $\DU{\cdot}$ that meet the uniformity conditions that ensure by using Lemma 
	\ref{lemm:rejSampl} that the output of Algorithm \ref{algo:skeleton} using $\DVv{\cdot}$ and $\DUv{\cdot}$
	instead of $\DV{\cdot}$ and $\DU{\cdot}$ produces an error $\ev$ that is uniformly distributed over the 
	words of weight $w$. The outputs of $\DVv{\cdot}$ and $\DUv{\cdot}$ only differ from the output
	of $\DV{\cdot}$ and $\DU{\cdot}$ when a bad set $\cJ$ is encountered. These considerations can be 
	used to prove the following proposition.
	\begin{proposition} \label{pr:statDist} 
	 Algorithm \ref{algo:skeleton} based on $\DVv{\cdot}$ and $\DUv{\cdot}$ produces uniformly distributed errors 
	 $\evu$ of weight $w$. Let  $\ev$ be the output of Algorithm \ref{algo:skeleton} with the use of $\DV{\cdot}$ and $\DU{\cdot}$. Let $\cJ_V^{\textup{unif}}$ be uniformly distributed over the subsets of $\IInt{1}{n/2}$ of size $k_V-d$ whereas $\cJ_V$ is uniformly distributed over 
		the same subsets that are good for $\Hm_V$. Let $\cJ^{\textup{unif}}_{U,\xv_V,\ell}$ be uniformly distributed over the subsets of $\IInt{1}{n/2}$ of size $k_U-d$ such that their intersection with $\xv_V$ is of size $\ell$ whereas $\cJ_{U,\xv_V,\ell}$ is the uniform distribution over 
		the same subsets that are good for $\Hm_U$. We have:
		\begin{equation*}
		\rho\left( \ev;\evu \right) \leq \rho\left( \cJ_V;\cJ_V^{\textup{unif}} \right) + \sum_{\xv_V,\ell} \rho \left( \cJ_{U,\xv_V,\ell};\cJ_{U,\xv_V,\ell}^{\textup{unif}} \right)\mathbb{P}\left( k_{\neq 0} = \ell \mid \evu_V = \xv_V \right) \mathbb{P}\left( \evu_V = \xv_V \right)
		\end{equation*}
	\end{proposition}

	\begin{proof}
	The first statement about the output of Algorithm \ref{algo:skeleton} is a direct consequence of Fact \ref{fa:DVvDUv}
	and Lemma \ref{lemm:rejSampl}.
	The proof of the rest of the proposition relies on the following proposition \cite[Proposition 8.10]{GM02}:
	\begin{proposition} \label{propo:Mic} Let X,Y be two  random variables over a common set $A$. For any randomized function $f$ with domain $A$ using internal coins independent from $X$ and $Y$, we have:
		$$
		\rho\left( f(X);f(Y) \right) \leq \rho\left( X;Y \right).
		$$
	\end{proposition}

	 	Let us define for $\xv_V \in \F_3^{n/2}$ and $\xv_U \in \F_3^{n/2}$,
		\begin{equation} 
		p(\xv_V) \eqdef \mathbb{P}\left( \ev_V = \xv_V \right) \quad ; \quad q(\xv_V) \eqdef \mathbb{P}\left( \evu_V = \xv_V \right)
		\end{equation} 
		\begin{equation}
			p(\xv_U|\xv_V) \eqdef \mathbb{P}\left( \ev_U = \xv_U \mid \ev_V = \xv_V \right) \quad ; \quad q(\xv_U|\xv_V) \eqdef \mathbb{P}\left( \evu_U = \xv_U \mid \evu_V = \xv_V \right)
		\end{equation}
		We have,
		\begin{align} \label{eq:statDistJ}
		\rho\left( \ev;\evu \right) &= \rho\left( \ev_U,\ev_V;\evu_U,\evu_V \right) \nonumber \\
		&= \sum_{\xv_V,\xv_U} \left| p(\xv_V)p(\xv_U|\xv_V) -  q(\xv_V)q(\xv_U|\xv_V) \right|\nonumber \\
		&= \sum_{ \xv_V,\xv_U} \left| (p(\xv_V)- q(\xv_V))p(\xv_U|\xv_V) + (p(\xv_U|\xv_V) - q(\xv_U|\xv_V))q(\xv_V) \right| \nonumber\\
		&\leq \sum_{ \xv_V,\xv_U} \left| (p(\xv_V)- q(\xv_V))p(\xv_U|\xv_V) \right| + \left| (p(\xv_U|\xv_V) - q(\xv_U|\xv_V)q(\xv_V) \right|\nonumber \\
		&= \sum_{ \xv_V} \left| (p(\xv_V)- q(\xv_V)) \right| + \sum_{\xv_V,\xv_U} \left| p(\xv_U|\xv_V) - q(\xv_V| \xv_U) \right|q(\xv_V)
		\end{align}
		where in the last line we used that $\sum_{\xv_U}|p(\xv_U|\xv_V)| = 1$ for any $\xv_V$. Thanks to Proposition \ref{propo:Mic}: 
		\begin{equation}\label{eq:StatV}
		\sum_{ \xv_V} \left| p(\xv_V)- p(\xv_V)^{\textup{unif}} \right| \leq \rho\left( \cJ_V;\cJ_V^{\textup{unif}} \right) 
		\end{equation}
		as the internal distribution $\cD_V$ of $\DV{\cdot}$ is independent of $\cJ_V$ and $\cJ_V^{\textup{unif}}$. Let us upper-bound the second term of the inequality. The distribution of $k_{\neq 0}$ is only function of the weight of the vector given as input to $\DU{\cdot}$ or $\DUv{\cdot}$. Therefore, 
				\begin{equation}\label{eq:kneq0} 
		\mathbb{P}\left( k_{\neq 0} = \ell \mid \ev_V = \xv_V \right) = \mathbb{P}\left( k_{\neq 0} = \ell \mid \evu_V = \xv_V \right) 
		\end{equation}
		Let us define,
		\begin{equation*}
		p(\xv_U|\xv_V,\ell) \eqdef \mathbb{P}(\ev_U = \xv_U \mid k_{\neq 0} = \ell,\ev_V = \xv_V) \quad ; \quad q(\xv_U|\xv_V,\ell) \eqdef \mathbb{P}(\evu_U = \xv_U \mid k_{\neq 0} = \ell,\evu_V = \xv_V)
		\end{equation*}
		With this notation we obtain from \eqref{eq:kneq0}
		\begin{equation}\label{eq:statDisty}
		p(\xv_U|\xv_V) - q(\xv_U|\xv_V) = \sum_{\ell} \left(p(\xv_U|\xv_V,\ell) -  q(\xv_U|\xv_V,\ell)\right) \mathbb{P}\left( k_{\neq 0} = \ell \mid \evu_V = \xv_V \right)
		\end{equation} 
		The internal coins of $\DU{\cdot}$ and $\DUv{\cdot}$ are independent of $\cJ_{U,\xv_V}$ and $\cJ_{U,\xv_V}^{\textup{unif}}$ and by using Proposition \ref{propo:Mic} we have for any $\xv_V$ and $\ell$:
		\begin{equation}\label{eq:Z} 
		\sum_{x_{U}} | p(\xv_U|\xv_V,\ell) -  q(\xv_U|\xv_V,\ell) | \leq \rho\left(\cJ_{U,\xv_V,\ell};\cJ_{U,\xv_V,\ell}^{\textup{unif}} \right) 
		\end{equation} 
		Combining Equations  \eqref{eq:statDistJ} and \eqref{eq:StatV}, \eqref{eq:statDisty} and  \eqref{eq:Z}
		concludes the proof. 
	\end{proof}

	Quantities  
	$$
	\rho\left( \cJ_V;\cJ_V^{\textup{unif}} \right) \quad \mbox{;} \quad  \sum_{\xv_V,\ell} \rho \left( \cJ_{U,\xv_V,\ell};\cJ_{U,\xv_V,\ell}^{\textup{unif}} \right)\mathbb{P}\left( k_{\neq 0} = \ell \mid \evu_V = \xv_V \right) \mathbb{P}\left( \evu_V = \xv_V \right)
	$$ 
	are functions of $\Hm_V$ and $\Hm_U$. We are going to show that their probabilities over $\Hm_V$ and $\Hm_U$ to be greater than $1/3^{d}$ is negligible. We will first need the following lemma . 
\begin{lemma}\label{lem:prob}
Let $d$ and $m$ be two positive integers with $d < m$ and let $\Mm$ be a matrix chosen uniformly at random in $\F_3^{(m-d)\times m}$. The probability that $\Mm$ is of rank $< m-d$ is upper-bounded by $\frac{1}{2 \cdot 3^d}$.
\end{lemma}
\begin{proof}
Let $\Mm_1,\dots,\Mm_{m-d}$ be the rows of $\Mm$. Let $V_i$ be the vector space spanned by $\Mm_1,\dots,\Mm_i$. If $\Mm$ is not of full rank then necessarily for at least one 
$i \in\llbracket 1,m-d \rrbracket$ we have $\dim V_i = \dim V_{i-1}=i-1$ where $V_{-1}\eqdef \{0\}$.The probability $P$ that $\Mm$ is not of full rank is therefore 
upper-bounded by 
\begin{eqnarray*}
P & \leq  & \sum_{i=1}^{m-d} \prob(\dim V_i=\dim V_{i-1}=i-1)\\
& \leq & \sum_{i=1}^{m-d} \prob(\dim V_i=i-1|\dim V_{i-1}=i-1)\\
& = & \sum_{i=1}^{m-d} \frac{1}{3^{m+1-i}}\\
& \leq  & \frac{1}{2 \cdot 3^d}. 
\end{eqnarray*}
\end{proof}

The following lemma will be useful too.

\begin{lemma} \label{lemm:theo1}
Let $\Hm$ be a matrix chosen uniformly at random in $\F_3^{(n/2-k)\times n/2}$ and let $d$ be an integer in the range $\IInt{1}{k}$. We define $R \eqdef k/(n/2)$ and $\delta = d/(n/2)$. Let $\cJ^{\textup{unif}}$ be uniformly distributed over the subsets of $\IInt{1}{n/2}$ of size $k_V-d$ whereas $\cJ_\Hm$ is uniformly distributed over 
the same subsets that are good for $\Hm$.
We have 
$$
\prob\left(\rho(\cJ_{\textup{unif}};\cJ_\Hm)>\frac{1}{3^d}\right) \leq \frac{2}{\binom{n/2}{k-d}} \left( 3^{d} + 2\cdot 3^{2d + \gamma n/2} \right) 
$$
where 
$$
\gamma \eqdef \min\limits_{x > 0} \left( (1-R+\delta)\log_3\left(\frac{1+3x}{x}\right) + (R-\delta)\log_{3}(1+x) \right) - 1 + R 
$$
\end{lemma}
\begin{proof}
Recall that the statistical distance between the uniform distribution over $\IInt{1}{s}$ and the uniform
distribution over $\IInt{1}{t}$ (with $t \geq s$) is equal to $\frac{t-s}{t}$.
Let $N$ be the number of subsets of $\IInt{1}{n/2}$ of size $k-d$ that are bad for $\Hm$.
By using the previous remark, we obtain
\begin{equation}\label{eq:statdist}
\rho(\cJ_{\textup{unif}};\cJ_\Hm)  =  \frac{N}{\binom{n/2}{k-d}}.
\end{equation}
Let us index from $1$ to $\binom{n/2}{k-d}$ the
subsets of size $k-d$ of $\IInt{1}{n/2}$ and let $X_i$ be the indicator of the event ``the subset of index $i$ is bad''. We have
\begin{equation}
\label{eq:N}
N = \sum_{i=1}^{\binom{n/2}{k-d}} X_i.
\end{equation}
We have by using Bienaym\'e-Tchebychev's inequality, that for any positive integer $t$:
\begin{align} \label{eq:Prod0} 
\prob(N > \mathbb{E}(N)+t )  \leq & \frac{\var(N)}{t^2} \nonumber \\
 = & \frac{\sum_i \var(X_i) + \sum_{i \neq j} \esp(X_iX_j) - \esp(X_i)(X_j)}{t^2} \nonumber \\
 \leq & \frac{\esp(N)}{t^{2}} + \frac{1}{t^{2}} \left( \sum_{i \neq j} \esp(X_iX_j) - \esp(X_i)(X_j) \right) 
\end{align}
where we use in the last line that $\var(X_i) \leq \esp(X_i^{2})$ and $\esp(X_i^2) = \esp(X_i)$. Let us now upper-bound the second term of the inequality. We first define for any $i\neq j$ the intersection of the complementary of the sets indexed by $i$ and $j$ as $\cE_{i,j}$.

By definition of a bad set, if $\cE_{i,j} = \emptyset$ then $X_i=1$ and $X_{j}=1$ are independent events and $\esp(X_iX_j) = \esp(X_i)\esp(X_j)$. Otherwise, let $e_{i,j} \eqdef |\cE_{i,j}| > 0$. We have:
\begin{equation}\label{eq:Prod4}
\esp(X_i X_j) = \mathbb{P}\left( X_i = X_j = 1 \mid \Hm_{\cE_{i,j}} \mbox{ of full rank} \right) \mathbb{P}\left( \Hm_{\cE_{i,j}} \mbox{ of full rank} \right) + \mathbb{P}\left( \Hm_{\cE_{i,j}} \mbox{ not of full rank} \right) 
\end{equation}
Let us define $\varepsilon$ as :
$$
\varepsilon \eqdef 1 - \mathbb{P}\left( \Hm_{\cE_{i,j}} \mbox{ of full rank} \right) 
$$
By using Lemma \ref{lem:prob} we get,
\begin{equation} \label{eq:epsi} 
\varepsilon \leq \frac{1}{2\cdot 3^{n/2-k-e_{i,j}}} 
\end{equation}
where the case $e_{i,j} > n/2 - k$ is trivial. 
We observe now that,
\begin{equation}
\mathbb{P}\left( X_i = X_j = 1 \mid \Hm_{\cE_{i,j}} \mbox{ of full rank} \right) = \mathbb{P}\left( X_i = 1 \mid \Hm_{\cE_{i,j}} \mbox{ of full rank} \right)\mathbb{P}\left( X_j = 1 \mid \Hm_{\cE_{i,j}} \mbox{ of full rank} \right)
\end{equation}
as events $X_i = 1$ and $X_j = 1$ conditioned on $(\Hm)_{\cE_{i,j}}$ being of full rank are independent. It gives with Equation \eqref{eq:Prod4},
\begin{equation}\label{eq:Prod1}
\esp(X_i X_j) = \mathbb{P}\left( X_i = 1 \mid \Hm_{\cE_{i,j}} \mbox{ of full rank} \right)\mathbb{P}\left( X_j = 1 \mid \Hm_{\cE_{i,j}} \mbox{ of full rank} \right) \left(1 - \varepsilon \right) + \varepsilon
\end{equation} 
We remark,
\begin{equation} \label{eq:Prod2} 
\esp(X_i)\esp(X_j) \geq \mathbb{P}\left( X_i = 1 \mid \Hm_{\cE_{i,j}} \mbox{ of full rank} \right)\mathbb{P}\left( X_j = 1 \mid \Hm_{\cE_{i,j}} \mbox{ of full rank} \right) \left(1 - \varepsilon \right)^{2}
\end{equation}
Therefore, by combining Equations \eqref{eq:Prod1} and \eqref{eq:Prod2} we get for $e_{ij} \geq 1$
\begin{align} \label{eq:Prod3}
\esp(X_iX_j) - \esp(X_i)\esp(X_j) &\leq \mathbb{P}\left( X_i = 1 \mid \Hm_{\cE_{i,j}} \mbox{ of full rank} \right)\mathbb{P}\left( X_j = 1 \mid \Hm_{\cE_{i,j}} \mbox{ of full rank} \right)\left( 1 - \varepsilon -  \left(1 - \varepsilon \right)^{2} \right) + \varepsilon \nonumber \\
&\leq 2\varepsilon \nonumber \\
&\leq \frac{1}{3^{n/2-k-e_{i,j}}}
\end{align}
where in the last line we used Equation \eqref{eq:epsi}. When $e_{ij}=1$, $X_i$ and $X_j$ are independent and we have in this case
$\esp(X_iX_j) - \esp(X_i)\esp(X_j)=0$. Let us make the following computations by using \eqref{eq:Prod3}:  
\begin{align}
\sum_{i \neq j} \esp(X_iX_j) - \esp(X_i)\esp(X_j) &= \sum_{i}\sum_{e=1}^{n/2-k + d} \sum_{ j : e_{i,j} = e} \esp(X_iX_j) - \esp(X_i)\esp(X_j) \nonumber \\
& \leq  \sum_{i}\sum_{e=1}^{n/2-k + d} \sum_{ j : e_{i,j} = e} \frac{1}{ 3^{n/2-k-e_{i,j}}} \nonumber \\
&\leq \frac{1}{ 3^{n/2-k}}\binom{n/2}{k-d}\sum_{e=0}^{n/2-k+d}3^{e}\binom{n/2-k+d}{e}\binom{k-d}{n/2 - k + d - e} \label{eq:partie0} 
\end{align}
In order to prove the last inequality, let 
\begin{eqnarray*}
a(x) &\eqdef &\sum_{e} \binom{n/2-k+d}{e}(3x)^{e}=(1+3x)^{n/2-k+d}\\
 b(x) &\eqdef &\sum_{e} \binom{k-d}{e}x^{e}=(1+x)^{k-d}.
 \end{eqnarray*}
 We have:
 \begin{eqnarray}
a(x)b(x) &=&(1+3x)^{n/2-k+d}(1+x)^{k-d}\label{eq:partie1}\\
 c_{n/2 - k + d}& = &\sum_{e=1}^{n/2-k+d}3^{e}\binom{n/2-k+d}{e}\binom{k-d}{n/2 - k + d - e}\label{eq:partie2}
\end{eqnarray}
where $\sum_e c_e x^e \eqdef a(x)b(x)$. From \eqref{eq:partie1} and \eqref{eq:partie2} we deduce that for any $x>0$ we have
\begin{equation}
\label{eq:partie3}
c_{n/2 - k + d} \leq \frac{(1+3x)^{n/2-k+d}(1+x)^{k-d}}{x^{n/2-k+d}}.
\end{equation}
Plugging this inequality in \eqref{eq:partie0} yields
\begin{equation}\label{eq:partie4}
\sum_{i \neq j} \esp(X_iX_j) - \esp(X_i)\esp(X_j) \leq \frac{1}{ 3^{n/2-k}}\binom{n/2}{k-d} \min\limits_{x > 0}\frac{(1+3x)^{n/2-k+d}(1+x)^{k-d}}{x^{n/2-k+d}}.
\end{equation}
From the definition of $\gamma$ we have
\begin{equation}
\min\limits_{x > 0} \frac{(1+3x)^{n/2-k+d}(1+x)^{k-d}}{ 3^{n/2-k}x^{n/2-k+d}}=3^{\frac{\gamma n}{2}}.
\end{equation}
Therefore, by plugging Equation \eqref{eq:partie4} in  \eqref{eq:Prod0}:
\begin{align*}  
\prob(N > \mathbb{E}(N)+t ) &\leq \frac{\esp(N)}{t^{2}} + \frac{1}{t^{2}} \binom{n/2}{k-d}3^{\gamma n/2}  \\
&\leq \frac{\binom{n/2}{k-d}}{2 \cdot 3^{d}t^{2}} + \frac{1}{t^{2}} \binom{n/2}{k-d}3^{\gamma n/2}
\end{align*} 
where in the last inequality we used that $\esp(N) \leq \frac{\binom{n/2}{k-d}}{2 \cdot 3^{d}}$ which is obtained thanks to Lemma \ref{lem:prob}..
Therefore, by choosing $t = \frac{\binom{n/2}{k-d}}{2 \cdot 3^{d}}$,
\begin{align*} 
\prob\left(N > \esp(N)+\frac{\binom{n/2}{k-d}}{2 \cdot 3^{d}}\right) 
&\leq  \frac{1}{\binom{n/2}{k-d}} \left(  2\cdot 3^{d} + 4 \cdot 3^{2d +\gamma n/2} \right) 
\end{align*} 
But now as $\esp(N) \leq \frac{\binom{n/2}{k-d}}{2 \cdot 3^{d}}$,
$$
\prob\left(N > \frac{\binom{n/2}{k-d}}{3^{d}}\right) 
\leq  \frac{2}{\binom{n/2}{k-d}} \left(   3^{d} + 2\cdot 3^{2d+\gamma n/2} \right) 
$$
from which we easily conclude the proof by using Equation \eqref{eq:statdist}. 

\end{proof}

	\begin{remark} Let us apply Lemma \ref{lemm:theo1} to $\rho\left( \cJ_V,\cJ_V^{\textup{unif}} \right)$ defined in Proposition \ref{pr:statDist} for the proposed parameters of Wave. It gives:
		\begin{equation}
		\mathbb{P}\left( \rho\left( \cJ_V,\cJ_V^{\textup{unif}} \right) > \frac{1}{3^{d}} \right) \leq \frac{1}{2^{726}} \quad \mbox{ where } d = 162
		\end{equation}
	\end{remark} 
The second  term $\sum_{\xv_V,\ell} \rho \left( \cJ_{U,\xv_V,\ell};\cJ_{U,\xv_V,\ell}^{\textup{unif}} \right)\mathbb{P}\left( k_{\neq 0} = \ell \mid \evu_V = \xv_V \right) \mathbb{P}\left( \evu_V = \xv_V \right)$ that appears in Proposition \ref{pr:statDist}
can be treated similarly by using the Bienaym\'e-Tchebychev inequality and is also typically of order $\frac{1}{3^{d}}$.
This implies Theorem \ref{theo:trueRej}.

\section{Proof of Proposition \ref{prop:statDist}} \label{app:propstatDist}

Our goal in this section is to prove Proposition \ref{prop:statDist} of \S\ref{sec:domSampl}. 
It is based on two lemmas, the first one is the following:
\begin{lemma}
	\label{lem:inrc}
	Let $\yv$ be a non-zero vector of $\mathbb{F}_{3}^{n}$ and $\sv$ an arbitrary element in $\mathbb{F}_{3}^{r}$. We choose a matrix $\Hm$ of size $r \times n$ uniformly at random among the set of
	$r \times n$ ternary matrices. In this case
	$$
	\prob \left( \yv \transpose{\Hm} = \sv \right) = \frac{1}{3^{r}}
	$$
\end{lemma}

\begin{proof}
	The coefficient of $\Hm$ at row $i$ and column $j$ is denoted by $h_{ij}$, whereas the coefficients of $\yv$ and $\sv$ are denoted 
	by $y_i$ and $s_i$ respectively.
	The probability we are looking for is the probability to have 
	\begin{equation}
	\label{eq:probabilite}
	\sum_{j} h_{ij} y_j = s_i
	\end{equation} for all $i$ in $\llbracket 1, r\rrbracket$. 
	Since $\yv$ is non zero, it has at least one non-zero coordinate. Without loss of generality, we may assume that $y_1=1$. We may rewrite 
	\eqref{eq:probabilite} as 
	$h_{i1} = \sum_{j>1} h_{ij} y_j$. This event happens with probability $\frac{1}{3}$ for a given $i$ and with probability $ \frac{1}{3^{r}}$
	on all $r$ events simultaneously due to the independence of the $h_{ij}$'s. 
\end{proof}

\subsection{Proof of the variation of the left-over hash lemma}
\label{ss:leftoverhash}
\lemleftoverHash*

\begin{proof}
	Let $q_{h,f}$ be the probability distribution of the discrete random variable
	$(h_0,h_0(e))$ where $h_0$ is drawn uniformly at random in $\Hc$ and $e$ drawn uniformly at random in $E$
	(i.e. $q_{h,f} = \prob_{h_0,e}(h_0=h,h_0(e)=f)$).
	By definition of the statistical distance we have
	\begin{eqnarray}
	\esp_h\left\{\rho(\Dc(h),\Uc)\right\} &= & \sum_{h \in \Hc} \frac{1}{|\Hc|} \rho(\Dc(h),\Uc) \nonumber \\
	&= & \sum_{h \in \Hc} \frac{1}{2|\Hc|} \sum_{f \in F} \left| \prob_e(h(e)=f) - \frac{1}{|F|} \right| \nonumber \\
	& = & \frac{1}{2} \sum_{(h,f) \in \Hc \times F}   \left| \prob_{h_0,e}(h_0=h,h_0(e)=f) - \frac{1}{|\Hc| \cdot |F|} \right| \nonumber \\
	& = & \frac{1}{2} \sum_{(h,f) \in \Hc \times F}   \left| q_{h,f} - \frac{1}{|\Hc| \cdot |F|} \right|.
	\end{eqnarray}
	Using the Cauchy-Schwarz inequality, we obtain 
	\begin{equation} \label{eq:cauchy-schwarz}
	\sum_{(h,f) \in \Hc \times F}   \left| q_{h,f} - \frac{1}{|\Hc| \cdot |F|} \right| \leq  \sqrt{\sum_{(h,f) \in \Hc \times F}
		\left( q_{h,f} - \frac{1}{|\Hc| \cdot |F|}\right)^2} \cdot \sqrt{|\Hc| \cdot |F|}.
	\end{equation}
	Let us observe now that
	\begin{align}
	\sum_{(h,f) \in \Hc \times F}
	\left( q_{h,f} - \frac{1}{|\Hc| \cdot |F|}\right)^2 & =  \sum_{h,f}
	\left( q_{h,f}^2 - 2 \frac{q_{h,f}}{|\Hc| \cdot |F|}+ \frac{1}{|\Hc|^2 \cdot |F|^2}\right) \nonumber \\
	& = \sum_{h,f} q_{h,f}^2 -2 \frac {\sum_{h,f} q_{h,f}}{|\Hc| \cdot |F|} + \frac{1}{|\Hc| \cdot |F|} \nonumber \\
	& = \sum_{h,f} q_{h,f}^2  - \frac{1}{|\Hc| \cdot |F|}  \label{eq:q}.
	\end{align}
	Consider for $i \in \{0,1\}$ independent random variables $h_i$ and $e_i$ 
	that are  drawn uniformly at random in $\Hc$ and $E$ respectively.
	We continue this computation by noticing now that
	\begin{align}
	\sum_{h,f} q_{h,f}^2 & =  \sum_{h,f} \prob_{h_0,e_0} (h_0 = h,h_0(e_0)=f) \prob_{h_1,e_1} (h_1 = h,h_1(e_1)=f)   \nonumber \\
	& = \prob_{h_0,h_1,e_0,e_1} \left(h_0=h_1,h_0(e_0)=h_1(e_1)\right) \nonumber  \\
	& = \frac{\prob_{h_0,e_0,e_1}\left(h_0(e_0)=h_0(e_1)\right) }{|\Hc|}\nonumber \\
	& = \frac{1+\varepsilon }{|\Hc| \cdot |F|}.\label{eq:collision}
	\end{align}
	By substituting for $\sum_{h,f} q_{h,f}^2$ the expression obtained in \eqref{eq:collision} into
	\eqref{eq:q} and then back into \eqref{eq:cauchy-schwarz} we finally obtain
	$$
	\sum_{(h,f) \in \Hc \times F}   \left| q_{h,f} - \frac{1}{|\Hc| \cdot |F|} \right|\leq  \sqrt{\frac{1+\varepsilon}{|\Hc| \cdot |F|}
		- \frac{1}{|\Hc| \cdot |F|}}\sqrt{|\Hc| \cdot |F|} = \sqrt{\frac{\varepsilon}{|\Hc| \cdot |F|}}\sqrt{|\Hc| \cdot |F|} = \sqrt{\varepsilon}.
	$$
	This finishes the proof of our lemma.
\end{proof}

\subsection{Proof of Lemma \ref{lem:syndromeDistribution}}

\lemSyndromeDistribution*

\begin{proof}  
By using Notation \ref{nota:euv} and Proposition \ref{prop:decomposition}, the probability we are looking for is:
	$$
	\mathbb{P}\left( (\xv_U - \yv_U)\transpose{\Hm}_{U} = \mathbf{0} \mbox{ and } (\xv_V-\yv_V)\transpose{\Hm}_V = \mathbf{0} \right) 
	$$
		where the probability is taken over $\Hm_U,\Hm_V,\xv,\yv$. 	To compute the previous probability we will use Lemma \ref{lem:inrc} which motivates to distinguish between four disjoint events:	
		$$
			\mbox{\bf Event 1: }
		\mathcal{E}_{1} \eqdef \{  \xv_U = \yv_U \quad \mbox{and} \quad \xv_V \neq \yv_V \} \quad ; \quad 
		\mbox{ \bf Event 2: }
		\mathcal{E}_{2} \eqdef \{  \xv_U \neq \yv_U \quad \mbox{and} \quad \xv_V = \yv_V \} 
		$$
		$$\mbox{\bf Event 3: }
		\mathcal{E}_{3} \eqdef \{  \xv_U \neq \yv_U \quad \mbox{and} \quad \xv_V \neq \yv_V \} \quad ; \quad 
		\mbox{\bf Event 4: } 
		\mathcal{E}_{4} \eqdef \{  \xv_U = \yv_U \quad \mbox{and} \quad \xv_V = \yv_V \} 
		$$	
		Under these events we get thanks to Lemma \ref{lem:inrc} and $k = k_U + k_V$:
		\begin{align} 
		&\mathbb{P}_{\Hsec,\xv,\yv} \left( \xv\transpose{\Hm}_{\textup{sk}} = \yv\transpose{\Hm}_{\textup{sk}} \right) \nonumber\\
		&= \sum_{i=1}^{4} \mathbb{P}_{\Hsec} \left(  \xv\transpose{\Hm}_{\textup{sk}} = \yv\transpose{\Hm}_{\textup{sk}} | \mathcal{E}_{i} \right) \mathbb{P}_{\xv,\yv} \left( \mathcal{E}_{i} \right) \nonumber\\ 
		&= \frac{\mathbb{P}_{\xv,\yv} \left( \mathcal{E}_{1} \right)}{3^{n/2 - k_{V}}}  + \frac{\mathbb{P}_{\xv,\yv } \left( \mathcal{E}_{2} \right)}{3^{n/2 - k_{U}}}  + \frac{\mathbb{P}_{\xv,\yv} \left( \mathcal{E}_{3} \right)}{3^{n - k}}  + \mathbb{P}_{\xv,\yv} \left( \mathcal{E}_{4} \right) \nonumber \\
				&\leq \frac{1}{3^{n-k}} \left( 1 +  3^{n/2-k_{U}}\mathbb{P} \left( \mathcal{E}_{1} \right) + 3^{n/2-k_{V}} \mathbb{P} \left( \mathcal{E}_{2} \right) + 3^{n-k} \mathbb{P}(\mathcal{E}_{4})  \right), \label{eq:cdtGGen}
		\end{align}	
		where we used for the last inequality the trivial upper-bound $	\mathbb{P} \left( \mathcal{E}_{3} \right) \leq 1$.
		Let us now upper-bound (or compute) the probabilities of the events $\Ec_1$, $\Ec_2$ and $\Ec_4$.
		For $\Ec_4$, recall that from the definition of normalized generalized $(U,U+V)$-codes, we clearly have
		\begin{equation} \label{eq:cdtUV}
		\mathbb{P}_{\xv,\yv } \left( \mathcal{E}_{4} \right) = \prob (\xv = \yv) = \frac{1}{2^{w}\binom{n}{w}}.
		\end{equation}
		Let us now estimate the probability of $\cE_2$ for which we first derive the following upper-bound:
		$$
		\mathbb{P}\left( \cE_2 \right) \leq \mathbb{P}\left( \xv_V = \yv_V \right) 
		$$
		To upper-bound this probability, we first observe that for any error $\ev \in \F_3^{n/2}$ of weight $j$:
		\begin{align*}
		\mathbb{P}(\xv_V = \ev) &= \mathbb{P}\left( \xv_V = \ev \mid |\xv_V| = j\right) \mathbb{P}(|\xv_V| = j) \\
		&= \frac{1}{2^{j}\binom{n/2}{j}} q_{1}(j)
		\end{align*}
		where $\qu_1(j)$ denotes $\mathbb{P}(|\evu_V|=j)$ and is computed in Proposition \ref{propo:qu}. From this we deduce that
		\begin{align*}
		\mathbb{P}(\xv_V = \yv_V) &= \sum_{j=0}^{n/2} \sum_{ \ev \in \F_3^{n/2} : |\ev| = j} \mathbb{P}_{\xv}(\xv_V = \ev)^{2} \\
		&= \sum_{j=0}^{n/2}\frac{1}{2^{j}\binom{n/2}{j}} \qu_1(j)^{2}
		\end{align*}
		which gives:
		\begin{equation} \label{eq:cdtV}
		\mathbb{P}\left( \cE_2 \right) \leq \sum_{j=0}^{n/2}\frac{ \qu_1(j)^{2}}{2^{j}\binom{n/2}{j}}.
		\end{equation} 
		Let us now estimate the probability of $\cE_1$ for which we derive the following upper-bound:
		\begin{equation*} 
		\mathbb{P}_{\xv,\yv}(\cE_1) \leq \mathbb{P}(\xv_U = \yv_U) 
		\end{equation*} 

	By definition of $\xv_U$ and $\yv_U$, the event we are looking for is $
	\left\{  \dv\hsp(\xv_{1} - \yv_{1}) = \bv\hsp(\xv_{2} - \yv_{2}) \right\}$ which is the same  (up to a permutation of indices of $\xv$ and $\yv$ and by multiplying some of their component by $-1$) as the case where we consider:
	$$
	b_1 = \cdots = b_{n_I} = 0 \quad ; \quad b_{n_I + 1} = \cdots = b_{n/2} = d_1 = \cdots = d_{n/2} = 1
	$$
	where $n_I$ is the number of blocks of type I. This gives  the following probability to upper-bound
	$$
	\mathbb{P}\left(\forall i \in \llbracket 1,n_{I} \rrbracket, \mbox{ } (\xv_1 - \yv_{1})(i) = 0, \forall i \in \llbracket n_{I}+1, n/2 \rrbracket, (\xv_1 - \yv_{1})(i) = (\xv_2 - \yv_{2})(i)\right) 
	$$
	We clearly have:
	\begin{align} \label{eq:min} 
		\mathbb{P}(&\forall i \in \llbracket 1,n_{I} \rrbracket, \mbox{ } (\xv_1 - \yv_{1})(i) = 0, \forall i \in \llbracket n_{I}+1, n/2 \rrbracket, \mbox{ } (\xv_1 - \yv_{1})(i) = (\xv_2 - \yv_{2})(i))  \nonumber \\
	    &\leq \sum_{\ev\in\F_3^{{n_I}}}\mathbb{P}\left(\forall i \in \llbracket 1,n_{I} \rrbracket,  \xv_1(i) = \ev(i) \right)^{2} \nonumber \\
	    &\leq \sum_{j=0}^{n_I} \sum_{\ev' \in \F_3^{n_I}  : |\ev'| = j }\mathbb{P}\left(\forall i \in \llbracket 1,n_{I} \rrbracket,  \xv_1(i) = \ev'(i) \right)^{2} \nonumber \nonumber \\
		&= \sum_{j=0}^{n_I} \sum_{\ev' \in \F_3^{n_I} : |\ev'|=j} \left( \frac{\binom{n-n_I}{w-j}2^{w-j}}{\binom{n}{w}2^{w}} \right)^{2} \nonumber \\
		&=  \sum_{j=0}^{n_I} \binom{n_I}{j}2^{j} \left( \frac{\binom{n-n_I}{w-j}}{\binom{n}{w}2^{j}} \right)^{2} \nonumber 
	\end{align}
	which gives:
	\begin{equation} \label{eq:cdtU} 
	\mathbb{P}(\cE_1) \leq \sum_{j=0}^{n_I} \binom{n_I}{j}2^{-j} \left( \frac{\binom{n-n_I}{w-j}}{\binom{n}{w}} \right)^{2} 
	\end{equation}

	Therefore, with Equations \eqref{eq:cdtGGen},\eqref{eq:cdtUV}, \eqref{eq:cdtV} and  \eqref{eq:cdtU} we finally conclude the proof. \qed
\end{proof}

Lemmas \ref{lem:syndromeDistribution} and \ref{lem:leftOver} imply directly  Proposition \ref{prop:statDist} as shown in the following proof.

\begin{proof}[Proposition \ref{prop:statDist}]
	Indeed we let in Lemma \ref{lem:leftOver}, $E \eqdef \F_3^n$, $F \eqdef \F_3^{n-k}$ and $\Hc$ be the set of functions associated to 
	the $4$-tuples $(\Hm_U,\Hm_V,\Sm,\Pm)$ used to generate a public parity-check matrix $\Hpub$. These functions
	$h$ are given by $h(\ev) = \ev \trHpub$.
	Lemma \ref{lem:syndromeDistribution} gives an upper-bound for the $\varepsilon$ term in Lemma \ref{lem:leftOver} and this finishes the
	proof of Proposition \ref{prop:statDist}. 	
\end{proof}

\section{Distinguishing a Permuted Normalized Generalized $(U,U+V)$-Code}
\label{sec:keyAtt}

\subsection{Proof of Proposition \ref{prop:density}}

Our aim here is to prove,
\propdensity*

\begin{proof}
Lemma \ref{lem:inrc} in \S\ref{app:propstatDist} will be useful four the proof. 
The last part of Proposition \ref{prop:density} is a direct application of this lemma. 
We namely have
\begin{proposition}$ $
	\label{prop:wDistribRCode}
	Let $a(z)$ be the expected number of codewords of weight $z$ in a ternary linear code $\cC$ of length $n$ whose parity-check matrix is chosen 
	$\Hm$ uniformly at random among all binary matrices of size $r \times n$. We have
	$$a(z) = \frac{\binom{n}{z}}{3^{r}}.$$
\end{proposition} 

We are ready now to prove Proposition \ref{prop:density} concerning the expected weight distribution of a random 
generalized normalized $(U,U+V)$-code, namely a code $(\av\hsp U + \bv \hsp V, \cv \hsp U + \dv \hsp V)$ that we will denote by $\cC$.

\noindent
{\bf Weight distributions of} $(\av\hsp U,\cv \hsp U) \eqdef \{(\av\hsp\uv,\cv \hsp \uv):
\uv \in U\}$ {\bf and} $(\bv \hsp V,\dv \hsp V) \eqdef \{(\bv \hsp \vv,\dv \hsp \vv): \vv
\in V\}$. Let us recall from the definition of normalized generalized codes that
$
a_ic_i \neq 0
$ for all $i \in \llbracket 1,n/2 \rrbracket$ 
and therefore it follows directly from Proposition \ref{prop:wDistribRCode} 
since $a_{(\uv,\mathbf{0})}(z)=0$ for odd and $a_{(\uv,\mathbf{0})}(z)$ is equal to the expected number of codewords of weight $z/2$ in 
a random linear code of length $n/2$ with a parity-check matrix of size $(n/2-k_U) \times n/2$ when $z$ is even. On the other hand, the weight distribution of $(\bv\hsp\vv,\dv\hsp\vv)$ for $\vv \in V$ is little more sophisticate. It depends of the number $n_I$ (see Definition \ref{def:Vpositions}) 
when either $b_i=0$ or $d_i=0$, the other one is necessarily different from $0$. In this way, $a_{(\mathbf{0},\vv)}(z)$ is equal to the expected number of weight $j + \frac{z-j}{2}$ for all $j$ in $\llbracket 1,n_I \rrbracket$ in a random linear code of length $n/2$ where $j$ positions correspond to the $n_I$ positions which gives the number of block of type $I$ and $\frac{z-j}{2}$ for the others as there are involved in components which count twice in the weight. Furthermore this code has a parity-check matrix of size $(n/2 -k_V) \times n/2$ which easily gives from Proposition \ref{prop:wDistribRCode} the expected result for $a_{(\mathbf{0},\vv)}$.
\newline

\noindent {\bf Weight distributions of} $\cC$.
The normalized generalized $\UV$-code is chosen randomly by picking up a parity-check matrix $\Hm_U$ of $U$ (resp. $\Hm_V$ of $V$) uniformly at random among the set of $(n/2-k_U)\times n/2$ (resp. $(n/2-k_V)\times n/2$) ternary 
matrices.  Let $Z \eqdef \sum_{ \xv \in \mathbb{F}_{3}^{n}: |\xv|=z} Z_{\xv}$ where $Z_{\xv}$ is the indicator function of ``$\xv\in \cC$''. Therefore,
\begin{eqnarray}
a_{(\uv,\vv)}(z)& =& \esp(Z) \nonumber \\
& = & \sum_{\xv \in \mathbb{F}_{3}^{n}:|\xv|=z} \prob(\xv \in \cC) \label{eq:auuvw}
\end{eqnarray}
Therefore, by Proposition \ref{prop:decomposition} we get:
$
\xv\in \cC \iff \xv_U\transpose{\Hm}_U = \mathbf{0} \mbox{ and } \xv_V\transpose{\Hm}_V = \mathbf{0}
$
which lead to three disjoint cases 
to (we use in each case Lemma \ref{lem:inrc}):
\newline

\noindent
{\bf Case 1:}  $\xv_U = \mathbf{0}$ and $\xv_V \neq \mathbf{0}$,
$
\prob(\xv \in \cC )= \prob(\xv_V\transpose{\Hm}_V=\mathbf{0}) = \frac{1}{3^{n/2-k_V}}
$
\newline

\noindent
{\bf Case 2:} $\xv_U \neq \mathbf{0}$ and $\xv_V = \mathbf{0}$, $
\prob(\xv \in \cC )= \prob(\xv_U\transpose{\Hm}_U = \mathbf{0} )= \frac{1}{3^{n/2-k_U}}
$
\newline

\noindent
{\bf Case 3:}  $\xv_U \neq \mathbf{0}$ and $\xv_V \neq \mathbf{0}$, $
\label{eq:prob3}
\prob(\xv \in\cC )= 
\prob(\xv_V\transpose{\Hm}_V=\mathbf{0}, \xv_U\transpose{\Hm}_U=\mathbf{0} ) = \frac{1}{3^{n/2-k_U}} \frac{1}{3^{n/2-k_V}}
$
\newline

By substituting $ \prob(\xv \in \cC )$  in \eqref{eq:auuvw} and using definition of number of blocks of type $I$ we conclude the proof. 
\qed 

\end{proof}

\subsection{Proof of Propositions \ref{propo:recovU} and \ref{propo:recovV}}

Our aim is to prove the following proposition. It give the expected number of iteration of Algorithm \ref{algo:ComputeU} to output a non zero list with probability $\Omega(1)$.

\proporecovU*

\begin{proof}	It will be helpful to recall \cite[Lemma 3]{OT11}
  
  \begin{lemma}
    \label{lem:lower_bound}
    Choose a random code $\Cr$ of length $n$ from a parity-check matrix of size $r \times n$ chosen uniformly at random
    in $\mathbb{F}_{3}^{r \times n}$.
    Let $X$ be some subset of $\mathbb{F}_{3}^n$ of size $m$.
    We have
    \begin{displaymath}
      \prob(X \cap \Cr \neq \emptyset) \geq f\left( \frac{m}{3^{r}}\right).
    \end{displaymath}
  \end{lemma}

  We say that two positions $i$ and $j$ are matched (for $U'$) if and only if
  there exists $\lambda \in \{ \pm 1 \}$ such that $c_i=\lambda c_j$ for every 
  $\cv \in U'$.
  From the fact that we only consider normalized generalized $(U,U+V)$-codes, 
  there are $n/2$ pairs of matched positions.
  $Z$ will now be defined by the number of matched pairs that are included in $\llbracket 1,n \rrbracket \setminus \cI$ where $\cI$ is the random set of size $n-k-\ell$ which is drawn in Instruction 4 of Algorithm \ref{algo:ComputeU}.  
  We compute the probability of success by conditioning on the values taken by $Z$:
  \begin{equation}
    \label{eq:Psucc2}
    \Psucc = \sum_{z=0}^{n/2}  \prob(Z=z) 
    \prob\left(\left.\exists \xv \in U' : |\xv_{\bar{\cI}}|=p \;\right|Z=z  \right)
  \end{equation}
  where $\bar{\cI} \eqdef \llbracket 1,n\rrbracket \setminus \cI$.    
  Notice that we can partition $\bar{\cI}$ as $\bar{\cI}= \cJ_1 \cup \cJ_2$
  where $\cJ_2$ consists in the union of the matched pairs in $\bar{\cI}$. Note that $|\cJ_2|=2z$.
  We may further partition $\cJ_2$ as $\cJ_2 = \cJ_{21} \cup \cJ_{22}$ where the elements of a matched pair are divided into the two sets.
  In other words, neither $\cJ_{21}$ nor $\cJ_{22}$ contains a matched pair. We are going to consider the codes
  $$
    U" \eqdef \punc_{\cI}(U') \quad ; \quad 
    U''' \eqdef \punc_{\cI \cup \cJ_{22}}(U')
  $$
  The last code is of length $n - (n-k-\ell + z) = k+\ell-z$ as $|\cJ_{22}| = z$ and $|\cI| = n-k-\ell$. The point of defining the first code is that 
  \begin{displaymath}
    \prob\left(\exists \xv \in U' : |\xv_{\bar{\cI}}|=p \mbox{ }| \mbox{ } Z=z  \right)
  \end{displaymath} 
  is equal to the probability that $U"$ contains a codeword of weight $p$.
  The problem is that we can not apply Lemma \ref{lem:lower_bound} to it due to the matched positions it contains (the code is not random). 
  This is precisely the point of defining $U'''$. In this case,  we can consider that it is a random code whose  parity-check matrix is chosen uniformly at random among 
  the set of matrices of size $\max(0,k+\ell-z-k_U) \times (k+\ell-z)$. We can therefore apply Lemma \ref{lem:lower_bound} to it.
  We have to be careful about the words of weight $p$ in $U"$ though, since they do not have the same probability of occurring in $U"$ 
  due to the possible presence of matched pairs in the support. This is why we introduce for $i$ in $\llbracket 0,\lfloor p/2 \rfloor\rrbracket$ the sets 
  $X_i$ defined as follows
  \begin{displaymath}
    X_i \eqdef \{\xv=(x_i)_{i \in \bar{\cI} \setminus \cJ_{22}} \in \mathbb{F}_{3}^{k+\ell-z}: |\xv_{\cJ_1}|=p-2i,\mbox{ } |\xv_{\cJ_{21}}|=i\}
  \end{displaymath}
  A codeword of weight $p$ in $U"$ corresponds to some word in one of the $X_i$'s by puncturing it in $\cJ_{22}$. We obviously have the lower bound
  \begin{equation} 
    \prob\left\{\exists \xv \in U' : |\xv_{\bar{\cI}}|=p \mbox{ } | \mbox{ } Z=z  \right\} \geq \mathop{\max}\limits_{i=0}^{\lfloor p/2 \rfloor} \left\{ \prob(X_i \cap U''' \neq \emptyset) \right\}
  \end{equation}
  By using Lemma \ref{lem:lower_bound} we have
  \begin{equation}
    \prob(X_i \cap U''' \neq \emptyset) \geq f\left(\frac{\binom{k+\ell-2z}{p-2i} \binom{z}{i} 2^{p-i} }{3^{\max(0,k+\ell-z-k_U)}}\right).
  \end{equation}
  On the other hand, we may notice that 
  \begin{displaymath}
    \prob(Z=z) =  \frac{\binom{n/2}{z}\binom{n/2-z}{k+\ell-2z}2^{k+\ell-2z}}{\binom{n}{k+\ell}}.
  \end{displaymath}
  Thanks to these considerations we conclude the proof. \qed 
\end{proof}

\proporecovV*

\begin{proof}

          We have $\frac{n}{2} - n_I$ pairs of matched positions $i$ and $j$ (it exists $\lambda \in \{\pm 1\}$ such that  $c_i = \lambda c_j$ for every $\cv \in V'$).   Let us define the following set: $\cJ$ is the set of positions that are of the images of the permutation $\Pm$ of the positions $1 \leq i\leq n/2$ such that $b_i\neq 0$ and the images of positions $n/2 + j$ with $0 \leq j \leq n/2$ such that $d_j\neq 0$.
  \begin{remark} From Definition \ref{def:Vpositions} and Remark \ref{rem:nI} in \S\ref{sec:domSampl} it follows that
    $
    |\cJ| = n - n_I.
    $
  \end{remark} 
  
  Let us now bring in the following random variables $\cI' \eqdef \cI \cap \cJ$, $Z \eqdef \left| \cI'\right|$
  and $M$ be the number of matched pairs which are included in $\cJ \setminus \cI'$. 
  $\cJ \setminus \cI'$ represents the set of positions that are not necessarily equal to $0$ in the punctured code 
  $\punc_{\cI}(V')$ (see Figure \ref{fig:support}).
  \begin{figure}
    \caption{A figure representing $\cJ$, $\cI$ and $\cI'$ and the form of a codeword in $V'$. \label{fig:support}}
    \centering
    \includegraphics[width=14cm]{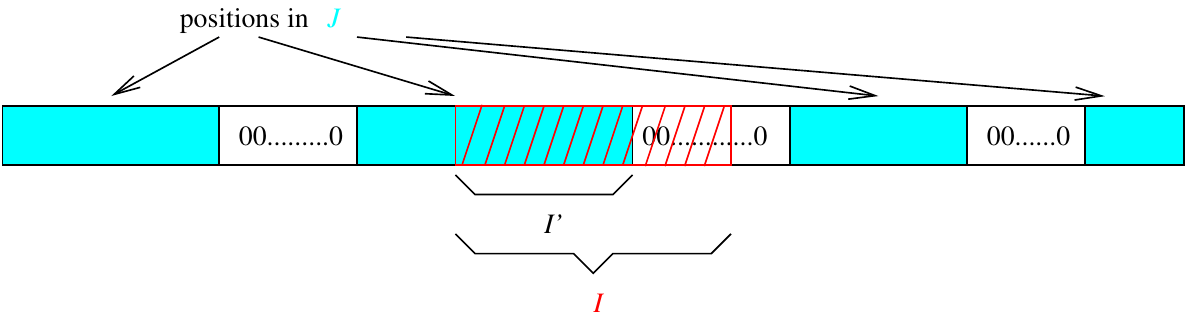}
  \end{figure}
  \textsc{ComputeV} outputs at least one element of $V'$ if
  there is an element of weight $p$ in $\punc_{\cI'}(V')$.
  Therefore the probability of success $\Psucc$ is given by
  \begin{equation}
    \label{eq:Psucc}
    \Psucc = \sum_{z=0}^{\min(n-k-\ell,n-n_I)}\sum_{m = 0}^{n/2-n_I}
    \prob\left(\exists \xv \in V' : |\xv_{\cJ'}|=p \mbox{ }| \mbox{ } Z=z, M = m  \right) \prob(Z=z, M = m) 
  \end{equation}
  where 
  \begin{displaymath}
  \cJ' \eqdef \cJ \setminus \cI'.
  \end{displaymath}
  Notice that we can partition $\cJ'$ as $\cJ' = \cJ_1 \cup \cJ_2$ where $\cJ_2$ consists in the union of the matched pairs in $\cJ'$. Note that $|\cJ_2| = 2m$. We may further partition $\cJ_2$ as $\cJ_2 = \cJ_{21} \cup \cJ_{22}$ where the elements of a matched pair are divided in two sets. In other words, neither $\cJ_{21}$ nor $\cJ_{22}$ contains a matched pair. We are going to consider the following codes
  $$
    V" \eqdef \punc_{\cI \cup \bar{\cJ}}(V') \quad ; \quad 
    V''' \eqdef \punc_{\cI \cup \bar{\cJ} \cup \cJ_{22}}(V').
  $$
  $V"$ is of length $n-n_I-z$, whereas the last code is of length $n-n_I-z-m$. 
  The point of defining the first code is that 
  \begin{displaymath}
    \prob\left(\exists \xv \in V' : |\xv_{\cJ'}|=p \mbox{ }| \mbox{ } Z=z  \right)
  \end{displaymath} 
  is equal to the probability that $V"$ contains a codeword of weight $p$.
  The problem is that we can not apply Lemma \ref{lem:lower_bound} to it due to the matched positions it contains. 
  This is precisely the point of defining $V'''$. In this case,  we can consider that it is a random code whose  parity-check matrix is chosen uniformly at random among 
  the set of matrices of size $\max(0,n-n_I-z-m-k_V) \times (n_V-z-m)$. We can therefore apply Lemma \ref{lem:lower_bound} to it.
  We have to be careful about the words of weight $p$ in $V"$ though, since they do not have the same probability of occurring in $V"$ 
  due to the possible presence of matched pairs in the support. This is why we introduce for $i$ in $\llbracket 0,\lfloor p/2 \rfloor\rrbracket$ the sets 
  $X_i$ defined as follows
  \begin{displaymath}
    X_i \eqdef \{\xv=(x_i)_{i \in \cJ' \setminus \cJ_{22}}\in \mathbb{F}_{3}^{n-n_I -z -m}: |\xv_{\cJ_1}|=p-2i,\mbox{ } |\xv_{\cJ_{21}}|=i\}
  \end{displaymath}
  A codeword of weight $p$ in $V"$ corresponds to some word in one of the $X_i$'s by puncturing it in $\cJ_{22}$. We obviously have the lower bound
  \begin{equation} 
    \prob\left\{\exists \xv \in V ' : |\xv_{\bar{\cI}}|=p \mbox{ } | \mbox{ } Z=z, M = m  \right\} \geq \mathop{\max}\limits_{i=0}^{\lfloor p/2 \rfloor} \left\{ \prob(X_i \cap V''' \neq \emptyset) \right\}
  \end{equation}
  By using Lemma \ref{lem:lower_bound} we have
  \begin{equation}
    \prob(X_i \cap V''' \neq \emptyset) \geq f\left(\frac{\binom{n - n_I - z - 2m}{p-2i}\binom{m}{i}2^{p-i} }{3^{\max(0,n-n_I - z - m - k_V)}}\right).
  \end{equation}
  On the other hand, we have 
  \begin{displaymath}
  \prob(Z=z, M = m) =\frac{\binom{\frac{n}{2} - n_I}{m}\binom{n_I}{n-k-\ell-z}}{\binom{n}{n-k-\ell}} \sum_{j = 0}^{n/2-n_{I} - m} \binom{n/2 - n_I-m}{j}2^{j}\binom{n_I}{z-n + 2n_I + 2m + j}
  \end{displaymath}

  Thanks to these considerations we conclude the proof. \qed 
\end{proof}

\subsection{Effective Estimate of the Security Exponent for the Recovery of $U$}
\label{app:CU} 
\subsubsection{Non Asymptotic Setting.}
Given $k,k_U$, we want to estimate $\min_{p,\ell}\WF_{p,\ell}$ where
\begin{displaymath}
  \begin{array}{rcccl}
    \WF_{p,\ell} & = & C_U(p,\ell) & = & {C_{p,\ell}}/{P_{p,\ell}} \\
    C_{p,\ell} & = & C_1(p,k,\ell) & = & \max\left(L_{p,\ell},{L_{p,\ell}^2}{3^{-\ell}}\right)
                                         \mbox{ with $L_{p,\ell}=\sqrt{\binom{k+\ell}{p}2^p}$} \\
    P_{p,\ell} & = & \Psucc & = & \displaystyle\sum_{z=0}^{n/2}\left(
                                  \frac{\binom{n/2}{z}\binom{n/2-z}{k+\ell-2z}2^{k+\ell-2z}}{\binom{n}{k+\ell}}
                                  \max_{0\le i\le p/2}f\left(\frac{\binom{k+\ell-2z}{p-2i}\binom{z}{i}2^{p-i}}{3^{\max(0,k+\ell-z-k_U)}}\right)\right)
  \end{array}
\end{displaymath}
with $f(x)=\max(1-1/x,x-x^2/2)$. We may simplify the function $f()$
which is equal up to a small constant factor (smaller than 3) to
$\min(1,x)$. We will now assume $f(x)=\min(1,x)$. We write
\begin{displaymath}
  P_{p,\ell} = \sum_{z=0}^{n/2} G_{\ell}(z) F_{p,\ell}(z),
\end{displaymath}
with
\begin{eqnarray*}
  G_{\ell}(z) & = & \frac{\binom{n/2}{z}\binom{n/2-z}{k+\ell-2z}2^{k+\ell-2z}}{\binom{n}{k+\ell}}, \\
  F_{p,\ell}(z) & = & \max_{0\le i\le p/2}f\left(\frac{\binom{k+\ell-2z}{p-2i}
                      \binom{z}{i}2^{p-i}}{3^{\max(0,k+\ell-z-k_U)}}\right) =
                      \min\left(1,\frac{\displaystyle\max_{0\le i\le p/2}
                      \phi_{p,\ell}(z,i)}{3^{k+\ell-z-k_U}}\right), \\
  \phi_{p,\ell}(z,i) & = & \textstyle \binom{k+\ell-2z}{p-2i} \binom{z}{i}2^{p-i}
\end{eqnarray*}
(the max in the denominator of $F_{p,\ell}$ can be removed because $\phi_{p,\ell}\ge1$).

\subsubsection{Asymptotic Setting.}
We are interested by the asymptotic behavior of the above quantities
when $n$ goes to infinity. For the sake of simplicity, we will use the
same notations, but all integers parameters $k,k_U,p,\ell,z,i$ are
replaced by their relative values, the letter
$x\in\{k,k_U,p,\ell,z,i\}$ now stands for $x/n$, and instead of an
integer it is a real number.

The functions $C_{p,\ell},L_{p,\ell},P_{p,\ell},G_{\ell},F_{p,\ell},\phi_{p,\ell}$ now stand
for for their relative asymptotic exponent, that is any $X$ above now
stands for $\lim_{n\rightarrow\infty}\frac{1}{n}\log_2X$.

We rewrite
\begin{eqnarray*}
  \WF_{p,\ell} & = & C_{p,\ell} - P_{p,\ell} \\
  C_{p,\ell} & = & \max\left(L_{p,\ell}, 2L_{p,\ell}-\ell\log_23\right) \mbox{ with }
  L_{p,\ell}=\frac{k+\ell}{2} h_3\left(\frac{p}{k+\ell}\right) \\
  G_{\ell}(z) & = & \frac{1}{2}h_2(2z)
                      + \left(\frac{1}{2}-z\right)h_3\left(\frac{k+\ell-2z}{\frac{1}{2}-z}\right)
                      - h_2(k+\ell) \\
  F_{p,\ell}(z) & = & \min\left(0,\tilde{F}_{p,\ell}(z)\right) \\
  \tilde{F}_{p,\ell}(z) & = & \max_{0\le i\le p/2}\phi_{p,\ell}(z,i) - (k+\ell-z-k_U)\log_23 \\
  \phi_{p,\ell}(z,i) & = & (k+\ell-2z)h_3\left(\frac{p-2i}{k+\ell-2z}\right)
                           + wh_3\left(\frac{i}{z}\right)
\end{eqnarray*}
where $h_q(x)=-x\log_2(x/(q-1))-(1-x)\log_2(1-x)$ is the $q$-ary
entropy function. The sum in the denominator of $P_{p,\ell}$ will be
replaced by a maximum over $z$
\begin{equation}\label{eq:Pz}
  P_{p,\ell} = \max_{0\le z\le 1/2} \left( G_{\ell}(z) + F_{p,\ell}(z) \right)
\end{equation}
To determine which value of $z$ dominates in the above maximum, we
need to study the variations of $z\mapsto G_{\ell}(z)$ and
$z\rightarrow {F}_{p,\ell}(z)$. But before that we need to study the
variation of $i\mapsto\phi_{p,\ell}(z,i)$ to determine the dominant
term in $\max_{0\le i\le p/2}\phi_{p,\ell}(z,i)$.
\begin{itemize}
\item The partial derivative of $\phi_{p,\ell}(z,i)$ with respect to
  $i$ is
  \begin{displaymath}
    \frac{\partial \phi_{p,\ell}}{\partial i}(z,i) = \log_2\frac{(p-2i)^2(z-i)}{2i(k+\ell-2z-p+2i)^2}
  \end{displaymath}
  It follows that the value of $i$ which maximizes
  $\phi_{p,\ell}(z,i)$ is the solution of a polynomial equation of
  degree 3.
  \begin{equation}\label{eq:Q}
    Q(i)=2i(k+\ell-2z-p+2i)^2-(p-2i)^2(z-i)
  \end{equation}
  An easy analysis shows that $Q$ admits a unique real
  root in the interval $[0,p/2]$. We denote it $i_0(z)$. We have
  \begin{displaymath}
    \tilde{F}_{p,\ell}(z) = \phi_{p,\ell}(z,i_0(z)) - (k+\ell-z-k_U)\log_23
  \end{displaymath}
\item The variations of $z\mapsto\tilde{F}_{p,\ell}(z)$ are dominated by the
  term $z\log_23$ and $\tilde{F}_{p,\ell}(z)$ is an increasing
  function of $z$. We denote $z_1$ the (unique) root of
  $\tilde{F}_{p,\ell}(z)$ in the range $]k+\ell-1/2,(k+\ell)/2[$. The
  function $F_{p,\ell}(z)$ is increasing (almost linearly) for
  $z\in]k+\ell-1/2,z_1]$ and is null for $z\in[z_1,(k+\ell)/2[$.
\item The derivative of $z\rightarrow \tilde{F}_{p,\ell}(z)$ is equal to
  \begin{eqnarray*}
    \frac{d \tilde{F}_{p,\ell}}{d z}(z) & = & \frac{d i_0}{d z}(z) \frac{\partial
      \phi_{p,\ell}}{\partial i}(z,i_0(z)) + \frac{\partial
      \phi_{p,\ell}}{\partial z}(z,i_0(z)) + \log_23 \\
    & = & \frac{\partial \phi_{p,\ell}}{\partial z}(z,i_0(z)) + \log_23
          = \log_2\frac{3z(k+\ell-2z-p+2i_0(z))^2}{(z-i_0(z))(k+\ell-2z)^2}.
  \end{eqnarray*}
\item The derivative of $z\rightarrow G_{\ell}(z)$ is equal to
  \begin{displaymath}
    \frac{d G_{\ell}}{d z}(z) = \log_2\frac{(k+\ell-2z)^2}{2z(1-2k-2\ell+2z)}
  \end{displaymath}
  and is null for $z_0=(k+\ell)^2/2$. The function
  $z\mapsto G_\ell(z)$ is increasing for $z\in[k+\ell-1/2,z_0]$,
  decreasing for $z\in[z_0,(k+\ell)/2]$, and $G_\ell(z_0)=0$.
\item The derivative of $z\rightarrow G_{\ell}(z) +
  \tilde{F}_{p,\ell}(z)$ is equal to
  \begin{equation}\label{eq:dP}
    P'_{p,\ell}(z) = \frac{d G_{\ell}}{d z}(z) + \frac{d \tilde{F}_{p,\ell}}{d z}(z) =
    \log_2 \frac{3(k+\ell-2z-p+2i_0(z))^2}{2(z-i_0(z))(1-2k-2\ell+2z)}.
  \end{equation}
  There exists a unique $z\in]k+\ell-1/2,(k+\ell)/2[$ which cancels
  the above derivative we denote it $z_2$.
\end{itemize}
For a given pair $(p,\ell)$,
\begin{itemize}
\item Compute $z_0$, if $F_{p,\ell}(z_0)=0$ then $P_{p,\ell}=0$ and
  $\WF_{p,\ell}=C_{p,\ell}$.
\item Compute $z_1$, $z_2$, and $z=\min(z_1,z_2)$
  \begin{displaymath}
    \WF_{p,\ell} = C_{p,\ell} - G_{\ell}(z) - F_{p,\ell}(z)
  \end{displaymath}
\end{itemize}

\begin{proposition}
  For any $(k,k_U,p,\ell)$ let $z_0=(k+\ell)^2/2$ and let $z_1$ and $z_2$
  denote respectively the roots of $z\mapsto\tilde{F}_{p,\ell}(z)$ and
  $z\mapsto P'_{p,\ell}(z)$ for $z$ in $]k+\ell-1/2,(k+\ell)/2[$. We have
  \begin{displaymath}
    W_{p,\ell} = C_{p,\ell} - G_{\ell}(z) - F_{p,\ell}(z), \mbox{
      where } z = \max(z_0,\min(z_1,z_2)).
  \end{displaymath}
\end{proposition}

\paragraph{Further Simplifications.}
\begin{itemize}
\item We have a very good approximation of $i_0(z)$ with
  \begin{displaymath}
    i_0(z) \approx \frac{p}{2}\frac{pw}{pw+(k+\ell-2z)^2}.
  \end{displaymath}
  The above assumes that $Q(i)$, given in \eqref{eq:Q}, is close to
  affine when $i\in[0,p/2]$. It is true enough in practice.
\item {\bf \em Get rid of parameter $p$.} We have
  \begin{eqnarray*}
    C_{p,\ell} & = & \max\left(L_{p,\ell}, 2L_{p,\ell} - \ell\log_23\right)
  \end{eqnarray*}
  In the $\max$ above, {\em and for the optimal values of the
    parameters $p$ and $\ell$}, the two terms are always equal. This
  gives us and additional identity
  \begin{displaymath}
    h_3\left(\frac{p}{k+\ell}\right) = \frac{2\ell\log_23}{k+\ell}
  \end{displaymath}
  which allows us to express the optimal value of $p$ as function of
  $\ell$.
\end{itemize}

\paragraph{Application to Wave.}
For Wave $k_U = 0.8451\, n/2$ and $k=0.676\, n$. In relative value
$k_U=0.42255$ and $k=0.676$. The minimal value for $W_{p,\ell}$ is
reached for $(p,\ell)=(0.0008048,0.003088)$ and the
dominant term in \eqref{eq:Pz} corresponds to $z=0.25135$. Finally
\begin{displaymath}
  \frac{1}{n}\log_2\min_{p,\ell}C_U(p,\ell)=0.01768.
\end{displaymath}
\paragraph{Application to Wave Dual Code.}
The above analysis must also be applied the dual code. In that case,
we replace $k$ by $n-k$ and $k_U$ by $n/2-k_V$ (in the dual $U$ is
replaced by $V^\perp$ and $V$ by $U^\perp$). We repeat the analysis
with $k_U=0.246545$ and $k=0.324$. The minimal value for $W_{p,\ell}$ is
reached for $(p,\ell)=(0.0004627,0.001737)$ and the
dominant term in \eqref{eq:Pz} corresponds to $z=0.07598$. Finally
\begin{displaymath}
  \frac{1}{n}\log_2\min_{p,\ell}C_{V^\perp}(p,\ell)=0.01811.
\end{displaymath}
\subsection{Security Exponent for the Recovery of $V$} \label{app:CV}
For the Wave parameters the cost $C_V(p,\ell)$ for recovering $V$ is
much larger than the cost $C_U(p,\ell)$ for recovering $U$. The same
holds for $U^\perp$ versus $V^\perp$. Finally, for Wave parameters,
the smallest of all is $C_U(p,\ell)$ and it will be used for selecting
the parameters.

\section{Proofs for \S\ref{sec:securityProof}}

\subsection{Basic Tools}

When we have probability distributions $\Dc_1$, $\Dc_2$, \dots, $\Dc_n$ over discrete sets 
$\Ec_1$, $\Ec_2$, \dots, $\Ec_n$, we denote by $\Dc_1 \otimes \Dc_2 \otimes \cdots \otimes \Dc_n$ the product probability distribution, i.e
$\Dc_1 \otimes \cdots \otimes \Dc_n(x_1,\dots,x_n) \eqdef \Dc_1(x_1) \dots \Dc_n(x_n)$ for 
$(x_1,\dots,x_n) \in \Ec_1 \times \cdots \times \Ec_n$. The $n$-th power product of a distribution $\Dc$ is denoted by $\Dc^{\otimes n}$, i.e.
$\Dc^{\otimes n} \eqdef \underbrace{\Dc \otimes \cdots \otimes \Dc}_{n \;\text{times}}$. 
Recall that the statistical distance $\rho$ is defined in Section
\S\ref{sec:nota}. We will need the following well known property for
the statistical distance which can be easily proved by induction.
\begin{proposition}
	\label{prop:product}
	Let $(\cD^0_1,\dots,\cD^0_n)$ and $(\cD^1_1,\dots,\cD^1_n)$ be two
	$n$-tuples of discrete probability distributions where $\cD^0_i$ and
	$\cD^1_i$ are distributed over a same space. For
	all positive integers $n$:
	\begin{displaymath}
	\rho\left(\cD^0_1 \otimes \dots \otimes \cD^0_n,\cD^1_1 \otimes \dots \otimes \cD^1_n \right) \leq \sum_{i=1}^n \rho(\cD^0_i,\cD^1_i).
	\end{displaymath}
\end{proposition}

\subsubsection{The Game Associated to Our Code-Based Signature Scheme.}
The modern approach to prove the security of cryptographic schemes is
to relate the security of its primitives to well-known problems that
are believed to be hard by proving that breaking the cryptographic
primitives provides a mean to break one of these hard problems.  In
our case, the security of the signature scheme is defined as a game
with an adversary that has access to hash and sign oracles.  It will
be helpful here to be more formal and to define more precisely the
games we will consider. They are games between two players, an {\em
	adversary} and a {\em challenger}. In a game $G$, the challenger
executes three kind of procedures:
\begin{itemize}
	\item an initialization procedure {\tt Initialize} which is called
	once at the beginning of the game.
	\item oracle procedures which can be requested at the will of the
	adversary. In our case, there will be two, {\tt Hash} and {\tt
		Sign}. The adversary $\cA$ which is an algorithm may call {\tt
		Hash} at most $\qhash$ times and {\tt Sign} at most $\qsig$ times.
	\item a final procedure {\tt Finalize} which is executed once $\cA$
	has terminated. The output of $\cA$ is given as input to this
	procedure.
\end{itemize}
The output of the game $G$, which is denoted $G(\cA)$, is the output
of the finalization procedure (which is a bit $b \in \{0,1\}$). The
game $G$ with $\cA$ is said to be successful if $G(\cA)=1$. The
standard approach for obtaining a security proof in a certain model is
to construct a sequence of games such that the success of the first
game with an adversary $\cA$ is exactly the success against the model
of security, the difference of the probability of success between two
consecutive games is negligible until the final game where the
probability of success is the probability for $\cA$ to break one of
the problems which is supposed to be hard. In this way, no adversary
can break the claim of security with non-negligible success unless it
breaks one of the problems that are supposed to be hard.

In the following, $\cS_{\textup{Wave}}$ will denote the signature scheme defined with the Wave-PSF family.

\begin{definition}[challenger procedures in the EUF-CMA Game]
	The challenger procedures for the \textup{EUF-CMA} Game corresponding to $\cS_{\textup{Wave}}$ are defined as:
	\begin{center} \tt
		\begin{tabular}{|l|l|l|l|}
			\hline
			\underline{proc Initialize$(\lambda)$} & \underline{proc Hash$(\mv,\rv)$} & \underline{proc Sign$(\mv)$} & \underline{proc Finalize$(\mv,\ev,\rv)$} \\ 
			$(pk,sk) \leftarrow \Gen(1^{\lambda})$ & {return} $\hash(\mv,\rv)$ & $\rv \Unif \{0,1\}^{\lambda_{0}}$ & $\sv \leftarrow {\tt Hash}(\mv,\rv)$ \\
			$\Hpub\leftarrow pk$ & & $\sv \leftarrow$ \texttt{Hash}$(\mv,\rv)$  & return \\
			$(\varphi,\Hm_U,\Hm_V,\Sm,\Pm)\leftarrow sk$ & & $\ev \leftarrow D_{\varphi,\Hm_U,\Hm_V}(\sv\transpose{\left(\Sm^{-1}\right)})$ & $\ev\transpose{\Hm}_{\textup{pk}} = \sv \wedge \wt{\ev} = w$ \\
			return $\Hpub$ & & return $(\ev\Pm,\rv)$ &   \\
			\hline  
		\end{tabular}
	\end{center}
\end{definition}

\subsection{The Proof}

We can now prove the following theorem
\secuRed*

\begin{proof}
	Let $\cA$ be a $(t,\qsig,\qhash,\varepsilon)$-adversary
	in the EUF-CMA model against $\cS_{\textup{Wave}}$ and let  $(\Hm_{0},\sv_{1},\cdots,\sv_{\qhash})$  be drawn uniformly at random among all instances of $\DOOM$ for parameters $n,k,\qhash,w$. We stress here that syndromes $\sv_{j}$ are random and independent vectors of $\mathbb{F}_{3}^{n-k}$. We write
	$\mathbb{P}\left( S_{i} \right)$ to denote the probability of
	success for $\cA$ of game $G_{i}$. Let
	\bigskip
	
	{\bf Game $0$} is the EUF-CMA game for $\cS_{\textup{Wave}}$.
	\bigskip

	{\bf Game $1$} is identical to Game $0$ unless the following failure event $F$ occurs: there is a collision in a signature query ({\em i.e.} two signatures queries for a same message $\mv$ lead to the same salt $\rv$). By using the difference lemma
	(see for instance \cite[Lemma 1]{S04a}) we get:
	\begin{displaymath}
	\mathbb{P}\left( S_{0} \right) \leq \mathbb{P}\left( S_{1} \right) + \mathbb{P} \left( F \right). 
	\end{displaymath}

	The following lemma (see \ref{subse:proofLemmaFail} for a proof) shows that in our case as $\lambda_{0} = \lambda + 2\log_{2}(\qsig)$, the probability of the event $F$ is negligible.
	\begin{lemma}
		\label{lemm:fEvent}
		For $\lambda_{0} = \lambda + 2\log_{2}(\qsig)$ we have:
		$
		\mathbb{P}\left( F \right) \leq \frac{1}{2^{\lambda}}.
		$
	\end{lemma}
	
	{\bf Game $2$} is modified from Game $1$ as follows:\smallskip
	\begin{center}{\tt
			\begin{tabular}{|l|l|}
				\hline
				\underline{proc Hash$(\mv ,\rv)$} & \underline{proc Sign$(\mv )$} \\
				if $\rv \in \listM$ & $\rv\gets \listM$.next$()$ \\
				\quad $\ev_{\mv ,\rv}\Unif S_w$ & $\sv\gets$ Hash$(\mv ,\rv)$ \\
				\quad return $\ev_{\mv ,\rv}\transpose{\Hm}_{\textup{pk}}$ & $\ev \gets D_{\varphi,\Hm_U,\Hm_V}(\sv\transpose{\left(\Sm^{-1}\right)})$\\
				else & return $\left(\ev\Pm,\rv\right)$ \\
				\quad $j\gets j+1$ & \\
				\quad return $\sv_j$ & \\
				\hline
		\end{tabular}}
		\newlength{\mylen}\settowidth{\mylen}{\tt
			\begin{tabular}{|l|l|}
				\hline
				\underline{proc Hash$(\mv ,\rv)$} & \underline{proc Sign$(\mv )$} \\
				if $\rv \in \listM$ & $\rv\gets \listM$.next$()$ \\
				\quad $\ev_{\mv ,\rv}\Unif S_w$ & $\sv\gets$ Hash$(\mv ,\rv)$ \\
				\quad return $\ev_{\mv ,\rv}\transpose{\Hm}_{\textup{pk}}$ & $\ev \gets D_{\Hsec,w}(\sv\transpose{(\Sm^{-1})})$ \\
				else & return $\left(\ev\Pm,\rv\right)$ \\
				\quad $j\gets j+1$ & \\
				\quad return $\sv_j$ & \\
				\hline
		\end{tabular}}
		\addtolength{\mylen}{-2.1\mylen}
		\addtolength{\mylen}{\linewidth}
		\addtolength{\mylen}{-15pt}\hfill
		\begin{minipage}{\mylen}
			To each message $\mv $ we associate a list $\listM$ containing  $\qsig$
			random elements of $\F_{2}^{\lambda_0}$. It is constructed the
			first time it is needed.  The call $\rv \in \listM$
			returns true if and only if $\rv$ is in the list. The call
			$\listM.{\tt next}()$ returns elements of $\listM$
			sequentially. The list is large enough to satisfy all queries.
		\end{minipage}
	\end{center}
	The {\tt Hash} procedure now creates the list $\listM$ if needed, then, if $\rv\in\listM$ it returns $\ev_{\mv,\rv}\transpose{\Hm}_{\textup{pk}}$ with $\ev_{\mv,\rv} \Unif S_{w}$. Although we do not use it in this game, we remark that $(\ev_{\mv,\rv},\rv)$ is a valid signature for $\mv$. The error value is stored. If $\rv\not\in\listM$ it outputs one of $\sv_j$ of
	the instance $(\Hm_{0},\sv_{1},\ldots,\sv_{\qhash})$ of the {DOOM}
	problem. The {\tt Sign} procedure is unchanged, except for $\rv$
	which is now taken in $\listM$.  The global index $j$ is set to 0 in {\tt
		proc Initialize}.
	This game can be related to the previous one through the following lemma.
	\begin{restatable}{lemma}{lemdistribi}
		\label{lem:distribi}
		\begin{displaymath}
		\mathbb{P}(S_{1})\leq \mathbb{P}(S_{2}) + \frac{\qhash}{2}\sqrt{\varepsilon} \mbox{ where } \varepsilon \mbox{ is given in Proposition \ref{prop:statDist}.} 
		\end{displaymath}
	\end{restatable}
        The proof of this lemma is given later in the appendix and relies among
        other things on the following points:
	\begin{itemize} 
		\item Proposition \ref{prop:product};
		\item Syndromes produced by matrices $\Hm_{\text{pk}}$ with errors of weight $w$ have average statistical distance from 
		the uniform distribution over $\mathbb{F}_{3}^{n-k}$  at most $\frac{1}{2} \sqrt{\varepsilon}$ (see Proposition \ref{prop:statDist}).
	\end{itemize}
	
	{\bf Game $3$} differs from Game $2$ by changing in {\tt proc
		Sign} calls ``$\ev \gets D_{\varphi,\Hm_U,\Hm_V}(\sv\transpose{\left(\Sm^{-1}\right)})$'' by
	``$ \ev \gets \ev_{\mv ,\rv}$'' and ``return $(\ev\Pm,\rv)$'' by
	``return $(\ev,\rv)$''. 
	Any signature $(\ev,\rv)$ produced by
	{\tt proc Sign} is valid.  
	The error $\ev$ is drawn according to
	the uniform distribution $\mathcal{U}_{w}$ while previously it was
	drawn according to Algorithm \ref{algo:skeleton} distribution, that
	is $\mathcal{D}_{w}$.
	By using Proposition \ref{prop:product} it follows that
	\begin{displaymath}
	\mathbb{P} \left( S_{2} \right) \leq \mathbb{P} \left( S_{3}
	\right) + \qsig \rho\left( \mathcal{U}_{w},\mathcal{D}_{w} \right).
	\end{displaymath}
	
	{\bf Game $4$} is the game where we replace the public matrix
	$\Hpub$ by $\Hm_{0}$. 
	In this way we will force the adversary to
	build a solution of the \DOOM\ problem. Here
	if a difference is detected between games it gives a
	distinguisher between distributions $\Drand$
	and $\Dpub$:
	\begin{displaymath}
	\mathbb{P} \left( S_{3} \right) \leq \mathbb{P} \left( S_{4} \right) + \rho_{c} \left( \Dpub,\Drand \right)\left(t_{c} \right).   
	\end{displaymath}

	We show in appendix how to emulate the lists $\listM$ in such a way
	that list operations cost, including its construction, is at most
	linear in the security parameter $\lambda$. Since $\lambda\le n$, it
	follows that the cost to a call to {\tt proc Hash} cannot exceed
	$O(n^2)$ and the running time of the challenger is
	$t_{c} = t +  O\left( \qhash \cdot n^{2} \right)$.
	\bigskip

	{\bf Game $5$} differs in the finalize procedure.
	\begin{center} {\tt
			\begin{tabular}{|l|}
				\hline
				\underline{proc Finalize}$(\mv ,\ev,\rv)$ \\
				$\sv\gets$ Hash$(\mv,\rv)$ \\
				$b \leftarrow \ev\transpose{\Hm}_{\textup{pk}} = \sv \wedge |\ev| = w$ \\  
				return $b \wedge \rv \notin \listM$ \\
				\hline
		\end{tabular}}
		\settowidth{\mylen}{\tt
			\begin{tabular}{|l|}
				\hline
				\underline{proc Finalize}$(\mv ,\ev,\rv)$ \\
				$\sv\gets$ Hash$(\mv ,\rv)$ \\
				$b \leftarrow \ev\transpose{\Hm}_{\text{pk}} = \sv \wedge |\ev| = w$ \\  
				return $b \wedge \rv \notin \listM$ \\
				\hline
		\end{tabular}}
		\addtolength{\mylen}{-2\mylen}
		\addtolength{\mylen}{\linewidth}
		\addtolength{\mylen}{-15pt}\hfill
		\begin{minipage}{\mylen}
			We assume the forger outputs a valid signature $(\ev,\rv)$ for
			the message $\mv $. The probability of success of Game $5$ is the
			probability of the event ``$S_4 \wedge(\rv\not\in\listM)$''.
		\end{minipage}
	\end{center}
	
	If the forgery is valid, the message $\mv $ has never been queried by
	{\tt Sign}, and the adversary never had access to any
	element of the list $\listM$. This way, the two events are
	independent and we get:
	\begin{displaymath}
	\mathbb{P} \left( S_{5} \right) = (1 - 2^{-\lambda_{0}})^{\qsig} \mathbb{P} \left( S_{4} \right).
	\end{displaymath}
	As we assumed $\lambda_{0}= \lambda + 2\log_{2}(\qsig) \geq  \log_{2}(\qsig^{2})$, we have:
	\begin{displaymath}
	\left( 1 - 2^{-\lambda_{0}} \right)^{\qsig} \geq \left( 1 - \frac{1}{\qsig^{2}}\right)^{\qsig} \geq \frac{1}{2}.
	\end{displaymath}
	Therefore
	\begin{equation}\label{eq:lower_bound_p5}
	\mathbb{P}\left( S_{5} \right) \geq \frac{1}{2}\mathbb{P}\left( S_{4} \right). 
	\end{equation}
	The probability $\mathbb{P} \left( S_{5} \right)$ is then exactly the probability for $\cA$ to output $\ev_{j} \in S_{w}$ such that $\ev_j\transpose{\Hm_{0}} = \sv_{j}$ for some $j$ which gives
	\begin{align}\label{eq:upper_bound_p5}
	\mathbb{P} \left( S_{5} \right) \leq\ Succ_{\DOOM}^{n,k,\qhash,w}(t_{c}).
	\end{align} 
	This  concludes the proof of Theorem \ref{theo:secRedu} by combining this together with all the bounds obtained for each of the previous games. \qed 
\end{proof}

\subsection{Proof of Lemma \ref{lemm:fEvent}} \label{subse:proofLemmaFail}

The goal of this subsection is to estimate the probability of a collision in a signature query for a message $\mv$ when we allow at most $\qsig$ queries. Recall that in $\mathcal{S}_{\textup{Wave}}$ for each signature query, we  pick $\rv$ uniformly at random in $\{0,1\}^{\lambda_{0}}$. Then the probability we are looking for is bounded by the probability to pick the same $\rv$ at least twice after $\qsig$ draws. The following lemma will be useful.

\begin{lemma} The probability to have at least one collision
	after drawing uniformly and independently $t$ elements in a
	set of size $n$ is upper bounded by ${t^{2}}/{n}$ for
	sufficiently large $n$ and $t^2< n$.
\end{lemma}

\begin{proof} The probability of no collisions after drawing
	independently $t$ elements among $n$ is:
	\[ p_{n,t} \eqdef \prod_{i=0}^{t-1}\left(1-\frac{i}{n}\right)
	\ge 1 - \sum_{i=0}^{t-1}\frac{i}{n} = 1 - \frac{t(t-1)}{2n} \]
	from which we easily get $1-p_{n,t}\le t^2/n$, concluding the proof. \qed 
\end{proof}

In our case, the probability of the event $F$ is bounded by
the previous probability for $t = \qsig$ and $n =
2^{\lambda_{0}}$, so, with $\lambda_0=\lambda+2\log_2\qsig$, we can conclude that
\begin{displaymath} 
\mathbb{P}\left( F \right) \leq \frac{\qsig^{2}}{2^{\lambda_{0}}}= \frac{1}{2^{\lambda_{0} - 2 \log_{2}(\qsig)}}= \frac{1}{2^{\lambda}}
\end{displaymath} 
which concludes the proof of Lemma \ref{lemm:fEvent}.

\subsection{List Emulation}
In the security proof, we need to build lists of indices (salts) in
$\F_3^{\lambda_0}$. Those lists have size $\qsig$, the maximum
number of signature queries allowed to the adversary, a number which
is possibly very large. For each message $\mv $ which is either hashed
or signed in the game we need to be able to
\begin{itemize}
	\item create a list $\listM$ of $\qsig$ random elements of
	$\F_3^{\lambda_0}$, when calling the constructor {\tt new list()};
	\item pick an element in $\listM$, using the method
	$\listM.\mathtt{next}()$, this element can be picked only once;
	\item decide whether or not a given salt $\rv$ is in $\listM$, when calling  $\listM.\mathtt{contains}(\rv)$.
\end{itemize}

The straightforward manner to achieve this is to draw $\qsig$ random
numbers when the list is constructed, this has to be done once for
each different message $\mv $ used in the game. This may result in a
quadratic cost $\qhash\qsig$ just to build the lists. Once the
lists are constructed, and assuming they are stored in a proper data
structure (a heap for instance) picking an element or testing
membership has a cost at most $O(\log \qsig)$, that is at most
linear in the security parameter $\lambda$.
\begin{figure}[h!]
	\centering
	\begin{tabular}{|l|l|}
		\hline
		\underline{class list} & \underline{method list.contains$(\rv)$} \\
		\quad elt, index & \quad return $\rv\in\{\mathtt{elt}[i],1\le i\le \qsig\}$ \\
		\quad list$()$ & \\\cline{2-2}
		\qquad $\mathtt{index}\gets0$ & \underline{method list.next$()$}  \\
		\qquad for $i=1,\ldots{},\qsig$ & \quad $\mathtt{index}\gets\mathtt{index}+1$ \\
		\qquad\quad $\mathtt{elt}[i]\gets\mathtt{randint}(2^{\lambda_0})$ & \quad return $\mathtt{elt[index]}$ \\
		\hline
	\end{tabular}
		\caption{Standard implementation of the list operations.\label{fig:standard}}
\end{figure}

Note that in our game we condition on the event that {\em all elements of $\listM$ are different}.
This implies that now  $\listM$ is obtained by choosing among the  subsets of size $\qsig$ of $\F_3^{\lambda_0}$ uniformly at random.
We wish to emulate the list operations and never construct them explicitly such that the probabilistic model for 
$\listM\mathtt{.next()}$ and $\listM\mathtt{.contains}(\rv)$ stays the same as above (but again conditioned on the event that  
all elements of $\listM$ are different).
For this purpose, we want to ensure that at any time we call either $\listM\mathtt{.contains}(\rv)$ or $\listM\mathtt{.next()}$ we have 
\begin{eqnarray}
\label{eq:list.contains}
\prob(\listM.\mathtt{contains}(\rv)=\mathtt{true})  &= & \prob(\rv \in \listM | \Qc)\\
\prob(\rv = \listM.\mathtt{next()}) & = & p(\rv|\Qc) \label{eq:list.next}
\end{eqnarray}
for every $\rv \in \F_3^{\lambda_0}$. Here $\Qc$ represents the queries to $\rv$ made so far and whether or not these $\rv$'s belong to
$\listM$. Queries to $\rv$ can be made through two different calls. The first one is a call of the form {\tt Sign}$(\mv)$ when it 
chooses $\rv$ during the random assignment $\rv \Unif \{ 0,1 \}^{\lambda_{0}}$. This results in a call to {\tt Hash}$(\mv,\rv)$ which queries itself
whether $\rv$ belongs to $\listM$ or not through the call $\listM\mathtt{.contains}(\rv)$. 
The answer is necessarily positive in this case. 
The second way to query $\rv$ is by calling 
{\tt Hash}$(\mv,\rv)$ directly. In this case, both answers {\tt true} and
{\tt false} are possible.
$p(\rv|\Qc)$ represents the probability distribution of $\listM\mathtt{.next()}$ that we have in the above implementation of the list operations 
given the previous queries $\Qc$.

A convenient way to represent $\Qc$ is  through three lists $\Lcts$, $\Lcth$ and $\Lcf$.
$\Lcts$ is the list of $\rv$'s that have been queried through a call {\tt Sign}$(\mv)$.
They belong necessarily to $\listM$. $\Lcth$ is the set of $\rv$'s that have not been queried so far through a call to {\tt Sign}$(\mv)$ but have been queried through a direct call {\tt Hash}$(\mv,\rv)$ and for which $\listM\mathtt{.contains}(\rv)$ returned {\tt true}.
$\Lcf$ is the list of $\rv$'s that have been queried by a call of the form {\tt Hash}$(\mv,\rv)$ and $\listM\mathtt{.contains}(\rv)$ returned 
{\tt false}. 

We clearly have
\begin{eqnarray}
\prob(\rv \in \listM | \Qc) &= & 0 \;\;\text{if $\rv \in \Lcf$} \label{eq:prrvinLm1}\\
\prob(\rv \in \listM | \Qc) &= & 1 \;\;\text{if $\rv \in \Lcts \cup \Lcth$} \label{eq:prrvinLm2}\\
\prob(\rv \in \listM | \Qc) &= & \frac{\qsig - |\Lcth| - |\Lcts|}{2^{\lambda_0}-|\Lcth| - |\Lcts|- |\Lcf|} \label{eq:prrvinLm3} \;\;\text{else.}
\end{eqnarray}
To compute the probability distribution $p(\rv|\Qc)$ it is helpful to notice that
\begin{equation}
\label{eq:nextisinLcth}
\prob(\listM\mathtt{.next()}\text{ outputs an element of $\Lcth$ }) = \frac{|\Lcth|}{\qsig - |\Lcts|}.
\end{equation}
This can be used to derive $p(\rv|\Qc)$ as follows
\begin{eqnarray}
p(\rv  | \Qc) &= & 0 \;\;\text{if $\rv \in \Lcf \cup \Lcts$} \label{eq:prq1}\\
p(\rv|\Qc) &= & \frac{1}{\qsig - \Lcts} \;\;\text{if $\rv \in \Lcth$} \label{eq:prq2} \\
p(\rv|\Qc) &= & \frac{\qsig - |\Lcts|-|\Lcth| }{(\qsig - \Lcts) (2^{\lambda_0}-|\Lcth| - |\Lcts|- |\Lcf|)} \;\;\text{else.}\label{eq:prq3}
\end{eqnarray}

\eqref{eq:prq1} is obvious. \eqref{eq:prq2} follows from that all elements of $\Lcth$ have the same probability to be chosen 
as return value for $\listM\mathtt{.next()}$ and \eqref{eq:nextisinLcth}. \eqref{eq:prq3} follows by a similar reasoning by arguing
(i) that all the elements of $\F_3^{\lambda_0} \setminus \left( \Lcts \cup \Lcth \cup \Lcf \right)$ have the same probability to be chosen as
return value for  $\listM\mathtt{.next()}$, (ii) the probability that $\listM\mathtt{.next()}$ outputs an element of $\F_3^{\lambda_0} \setminus \left( \Lcts \cup \Lcth \cup \Lcf \right)$ is the probability that it does not output an element of $\Lcth$ which is 
$1-\frac{|\Lcth|}{\qsig - |\Lcts|} = \frac{\qsig - |\Lcts| - |\Lcth|}{\qsig - |\Lcts|}$.

Figure \ref{fig:implementation} explains how we perform the emulation of the 
list operations so that they perform similarly to genuine list operations as specified above. 
The idea is to create and to operate explicitly on the lists $\Lcts$, $\Lcth$ and $\Lcf$ described earlier.
We have chosen there
\begin{displaymath}
\beta = \frac{\qsig - |\Lcth| - |\Lcts|}{2^{\lambda_0}-|\Lcth| - |\Lcts|- |\Lcf|}
\mbox{ and } \gamma = \frac{\wt{\Lcth}}{\qsig-\wt{\Lcts}}.
\end{displaymath}
we also assume that when we call {\tt randomPop()} on a list it outputs an element of the list uniformly at random
and removes this element from it.
The method {\tt
	push} adds an element in a list. The procedure $\mathtt{rand}()$
picks a real number between 0 and 1  uniformly at random. 

\begin{figure}
	\centering
	\begin{tabular}{|l|l|l|}
		\hline
		\underline{class list} & \underline{method list.contains$(\rv)$} & \underline{method list.next$()$} \\
		\quad $\Lcth$, $\Lcf$, $\Lcts$ & \quad if $\rv\not\in\Lcth \cup\Lcf \cup\Lcts $ & \quad if $\mathtt{rand}()\le\gamma$ \\
		\quad list$()$ & \qquad if rand$()\le\beta$ & \qquad $\rv\gets\Lcth.\mathtt{randomPop}()$ \\
		\qquad $\Lcth \gets\emptyset$ & \quad\qquad $\Lcth.\mathtt{push}(\rv)$ & \quad else \\
		\qquad $\Lcf \gets\emptyset$ & \qquad else & \qquad $\rv\Unif\F_3^{\lambda_0}\setminus(\Lcth \cup\Lcts  \cup\Lcf)$ \\
		\qquad $\Lcts \gets\emptyset$ & \quad\qquad $\Lcf.\mathtt{push}(\rv)$ & \quad $\Lcts.\mathtt{push}(\rv)$ \\
		& \quad return $\rv\in\Lcth \cup\Lcts$ & \quad return $\rv$ \\
		\hline
	\end{tabular}
	\caption{Emulation of the list operations. 
		\label{fig:implementation}}
\end{figure}

The correctness of this emulation follows directly from the calculations given above.
For instance the correctness of the call $\listM.\mathtt{next()}$ follows from the fact that with probability 
$\frac{\wt{\Lcth}}{\qsig-|\Lcts|}=\gamma$ it outputs an element of $\Lcth$ chosen uniformly at random (see \eqref{eq:nextisinLcth}).
In such a case the corresponding element has to be moved from $\Lcth$ to $\Lcts$ (since it has been queried now through a call 
to {\tt Sign}$(\mv)$). The correctness of $\listM.\mathtt{contains}(\rv)$ is a direct consequence of the formulas for 
$\prob(\rv \in \listM | \Qc)$ given in \eqref{eq:prrvinLm1}, \eqref{eq:prrvinLm2} and \eqref{eq:prrvinLm3}.
All {\tt push}, {\tt pop}, membership testing above can be implemented
in time proportional to $\lambda_0$. 
\subsection{Proof of Lemma \ref{lem:distribi}}
\label{app:lemm:disitribi}

Let us prove now Lemma \ref{lem:distribi} which is consequence of Propositions \ref{prop:statDist} and \ref{prop:product}.

\lemdistribi*

\begin{proof}
	To simplify notation we let $q \eqdef \qhash$.
	Then we notice that 
	\begin{equation}
	\label{eq:first}
	\prob(S_1) \leq \prob(S_2) + \rho(\Dpubwq,\Dpub\otimes \Uc^{\otimes q}),
	\end{equation}
	where 
	\begin{itemize}
		\item $\Uc$ is the uniform distribution over $\F_3^{n-k}$;
		\item $\Dpubwq$ is the distribution of the $(q+1)$-tuples
		$(\Hpub,\ev_1\trHpub ,\cdots,\ev_q\trHpub)$ where the $\ev_i$'s are independent and uniformly distributed in $S_w$;
		\item
		$\Dpub\otimes \Uc^{\otimes q}$  is the distribution of the $(q+1)$-tuples
		$(\Hpub,\sv_1,\cdots, \sv_q)$ where the $\sv_i$'s are independent and uniformly distributed in $\F_3^{n-k}$.
	\end{itemize}
	We now observe that 
	\begin{eqnarray*}
				\rho(\Dpubwq,\Dpub\otimes \Uc^{\otimes q}) &= & \sum_{\Hm \in \F_3^{(n-k) \times n}} \prob(\Hpub=\Hm) \rho((\Dc_w^\Hm)^{\otimes q},\Uc^{\otimes q}) \\
		& \leq & q \sum_{\Hm \in \F_3^{(n-k) \times n}} \prob(\Hpub=\Hm) \rho(\Dc_w^{\Hm},\Uc)\;\;\text{(by Prop. \ref{prop:product})} \\
		& = & q \esp_{\Hpub} \left\{\rho(\Dpubw,\Uc)\right\} \\
		& \leq & q \frac{\sqrt{\varepsilon}}{2} \;\;\text{(by Prop. \ref{prop:statDist})}.
	\end{eqnarray*}
	
\end{proof}

\end{document}